%% file: paper.tex
\documentclass[letterpaper,twocolumn,10pt]{article}
\usepackage{usenix-2020-09}
\usepackage{tikz}
\usepackage{amsmath}
\usepackage{filecontents}
\usepackage{amsfonts}
\usepackage{amssymb}
\usepackage{algorithmic}
\usepackage{graphicx}
\usepackage{textcomp}
\usepackage{xcolor}

\usepackage{mathrsfs}
\usepackage{amsthm}
\usepackage{epstopdf}
\usepackage{balance}
\usepackage{multirow}

\usepackage{epsfig,endnotes}
\usepackage{subfig}
\usepackage{grffile}
\usepackage[font=bf]{caption}
\usepackage{color, url}
\usepackage{xspace} 
\usepackage[ruled,linesnumbered]{algorithm2e}
\usepackage{epstopdf}
\usepackage{balance}
\usepackage{bm}
\usepackage{rotating}

\usepackage{enumitem}
\usepackage{makecell}
\usepackage{pdfx}
\usepackage{bm}

\newcommand\CR[1]{\textcolor{black}{#1}}

\newtheorem{theorem}{Theorem}
\newtheorem{lemma}{Lemma}
\newtheorem{definition}{Definition}

\renewcommand{\mathbf}[1]{\bm{#1}}

\newcommand{\argmax}{\operatornamewithlimits{argmax}}

\newcommand{\neil}[1]{{#1}}

\newcommand{\myparatight}[1]{\smallskip\noindent{\bf {#1}:}~}

\allowdisplaybreaks

\begin{document}
\pagestyle{empty}

\date{}

\title{\Large \bf PORE: Provably Robust Recommender Systems against  Data Poisoning Attacks}

\author{
{\rm Jinyuan Jia\thanks{Equal contribution.}$\:\;^{1}$, Yupei Liu\textcolor{green!80!black}{\footnotemark[1]}${\:\;^{2}}$, Yuepeng Hu$^2$, Neil Zhenqiang Gong$^2$} \\
$^1$The Pennsylvania State University, $^2$Duke University\\
$^1$jinyuan@psu.edu, $^2$\{yupei.liu, yuepeng.hu, neil.gong\}@duke.edu}

\maketitle

\input{abstract}

\input{introduction}
\input{preliminary}

\input{method}

\input{evaluation}

\input{relatedwork}
\input{discussion}

\input{conclusion}

\bibliographystyle{plain}
\bibliography{refs}

\input{appendix.tex}

\end{document}

%% file: abstract.tex
\begin{abstract}

 \emph{Data poisoning attacks} spoof a recommender system to make arbitrary, attacker-desired recommendations 
via injecting fake users with carefully crafted rating scores into the recommender system. 
We envision a cat-and-mouse game for such data poisoning attacks and their defenses, i.e., new defenses are designed to defend against existing attacks and new attacks are designed to break them. 
To prevent such cat-and-mouse game,  we propose PORE, the first framework to build \emph{provably} robust recommender systems in this work. 
PORE can transform \emph{any} existing recommender system  to be provably robust against \emph{any} \neil{untargeted} data poisoning attacks, \neil{which aim to reduce the overall performance of a recommender system}. 
Suppose PORE recommends top-$N$  items to a user when there is no attack. We prove that PORE still recommends at least $r$ of the $N$ items to the user  under {any} data poisoning attack, where $r$ is a function of the number of fake users in the attack. Moreover, we design an efficient algorithm to compute $r$ for each user. We empirically evaluate PORE on popular benchmark datasets.     
    
\end{abstract}

%% file: introduction.tex
\section{Introduction}

Many web service platforms (e.g., Amazon, YouTube, and TikTok) leverage recommender systems to engage users and improve user experience. 
Typically, a platform first collects a large amount of rating scores that users gave to items, which is known as  a \emph{rating-score matrix}.  
Then, the platform uses them to build a recommender system that models the complex relationships between user interests and item properties. Finally, the recommender system recommends top-$N$  items to each user that match his/her interests. 

However, due to its openness, i.e., anyone can register users and provide rating scores to items, recommender system is fundamentally not robust to 
\emph{data poisoning attacks}~\cite{lam2004shilling,mobasher2007toward,li2016data,yang2017fake,fang2018poisoning,fang2020influence,huangdata2021}. 
Specifically, in a data poisoning attack, an attacker  creates fake users in a recommender system and assigns carefully crafted rating scores to \neil{items}. Different data poisoning attacks essentially use different methods to craft the fake users' rating scores.  When a recommender system is built based on the \emph{poisoned-rating-score matrix}, which includes the rating scores of both genuine and fake users, it recommends attacker-chosen, arbitrary top-$N$ items to a user. As a result, the recommendation performance (e.g.,  Precision@$N$,  Recall@$N$, and  F1-Score@$N$) is substantially degraded. Data poisoning attacks pose severe challenges to the robustness/security of recommender systems. 

Many defenses have been proposed to enhance the robustness of recommender systems against data poisoning attacks. In particular, one family of defenses~\cite{burke2006classification,zhang2006attack,wu2012hysad,zhang2014hht,fang2018poisoning} aim to detect fake users before building a recommender system. These methods rely on the assumption that the rating scores of fake users and genuine users have statistically different patterns, which they leverage to distinguish between fake and genuine users. 
\neil{Another family of defenses~\cite{sandvig2007robustness,mehta2007robust,tang2019adversarial,chen2019adversarial,yuan2019adversarial,liu2020certifiable,hidano2020recommender} aim to design new methods of training  recommender systems such that they have good recommendation performance even if they are trained on a poisoned-rating-score matrix, e.g., using trim learning~\cite{hidano2020recommender}.  
However, these defenses only achieve  \emph{empirical} robustness, leading to an endless cat-and-mouse game between attacks and defenses: a new empirical defense is proposed to mitigate existing attacks but can be broken by new attacks that adapt to the defense.} For instance,  fake users can adapt their rating scores  such that they cannot be detected based on the rating scores' statistical patterns ~\cite{fang2018poisoning,fang2020influence,huangdata2021}.  As a result, a recommender system's performance is still substantially degraded by strong, adaptive attacks.

\myparatight{Our work} In this work, we aim to end such cat-and-mouse game via proposing PORE, the first framework to build \emph{provably} robust recommender systems against \emph{any} \neil{untargeted} data poisoning attacks. 
Suppose, under no attacks, a recommender system algorithm  trains a recommender system on a clean rating-score matrix,  which recommends a set of top-$N$ items (denoted as $\Gamma_u$) to a user $u$. Under attacks, the recommender system algorithm trains a recommender system on a poisoned-rating-score matrix, which recommends a set of top-$N$ items (denoted as $\Gamma_u'$) to the user.  
We say the recommender system algorithm is $(e,r)$-provably robust for the user $u$ if the intersection between $\Gamma_u$ and $\Gamma_u'$ includes at least $r$ items when there are at most $e$ fake users, no matter how an attacker crafts the fake users' rating scores. In other words, an $(e,r)$-provably robust recommender system guarantees that at least $r$ of the recommended top-$N$ items are unaffected by $e$ fake users no matter what rating scores they use. We note that $r$ depends on the number of fake users $e$ and we call $r$ \emph{certified intersection size}. A provably robust recommender system can guarantee a lower bound of recommendation performance under \emph{any} data poisoning attack, i.e.,  no matter how fake users craft their rating scores.

Suppose a \emph{submatrix} consists of $s$ rows of the  rating-score matrix, i.e., a submatrix includes rating scores of $s$ users. Intuitively, when the fraction of fake users is bounded, a randomly sampled  submatrix is likely to not contain fake users and thus a recommender system built based on the submatrix is not affected by fake users.  Based on this intuition, \neil{PORE uses \emph{bagging}~\cite{breiman1996bagging}, a well-known ensemble method,} to achieve provable robustness. In particular, PORE aggregates recommendations from multiple base recommender systems to recommend top-$N$ items to each user.  
Specifically, we can use any recommender system algorithm (called \emph{base algorithm}) to build a recommender system (called \emph{base recommender system}) on a submatrix. Therefore, we could build ${n \choose s}$ base recommender systems since there are ${n \choose s}$ submatrices, where $n$ is the total number of users. Each base recommender system makes recommendations to users.  We denote by $p_i$ the fraction of the ${n \choose s}$ base recommender systems that recommend item $i$ to a user $u$. We call $p_i$ \emph{item probability}.\footnote{Item probability $p_i$ also depends on user $u$, but we omit it for simplicity.} PORE recommends the top-$N$ items with the  largest item probabilities to user $u$.

Our major theoretical result is that we prove 
PORE is $(e,r)$-provably robust, no matter what base algorithm is used to train the base recommender systems. Moreover, for any given number of fake users $e$, we derive the certified intersection size $r$ for each genuine user, which is the solution to an optimization problem. 
 PORE relies on the item probabilities $p_i$'s to make recommendations. 
Moreover, the optimization problem to calculate $r$ also involves item probabilities. However, it is challenging to compute the exact item probabilities as it requires building ${n \choose s}$  base recommender systems. To address the challenge, we design an efficient algorithm to estimate the lower/upper bounds of the item probabilities via building $T \ll {n \choose s}$ base recommender systems, where 
the $T$ base recommender systems can be built in parallel. 
PORE makes recommendations based on the estimated item probabilities in practice. 
Moreover, we use the estimated item probabilities to solve the optimization problem to obtain $r$ for each user.

We empirically evaluate  PORE on three benchmark datasets, i.e., MovieLens-100k, MovieLens-1M, \neil{and MovieLens-10M}. Moreover, we consider two state-of-the-art base algorithms, i.e., Item-based Recommendation (IR)~\cite{argyriou2020microsoft} and Bayesian Personalized Ranking (BPR)~\cite{rendle2012bpr}, to show the generality of PORE. 
\CR{We also generalize  state-of-the-art provably robust defense~\cite{jia2020intrinsic} against data poisoning attacks for machine learning classifiers to recommender systems and compare PORE with it.}
We have \CR{three} key observations from our experimental results. \CR{First, PORE substantially outperforms the defense generalized from classifiers.} Second, when there are no data poisoning attacks, PORE has comparable recommendation performance (i.e., Precision@$N$,  Recall@$N$, and  F1-Score@$N$) with a standard recommender system built on the entire rating-score matrix. Third, under any data poisoning attacks, PORE can guarantee a lower bound of recommendation performance, while the standard recommender systems cannot. 

Our key contributions are summarized as follows:
\begin{itemize}
    \item We propose PORE, the first framework to build  recommender systems that are provably robust against \neil{untargeted} data poisoning attacks. 
    
    \item We prove the robustness guarantees of PORE and derive its certified intersection size. Moreover, we design an algorithm to compute the certified intersection size. 
    
    \item We perform extensive evaluation on popular benchmark datasets using two state-of-the-art base recommender system algorithms.  
\end{itemize}

%% file: preliminary.tex
\section{Background}
\label{problem}

\begin{figure*}[!t]
	 \centering
{\includegraphics[width=0.9\textwidth]{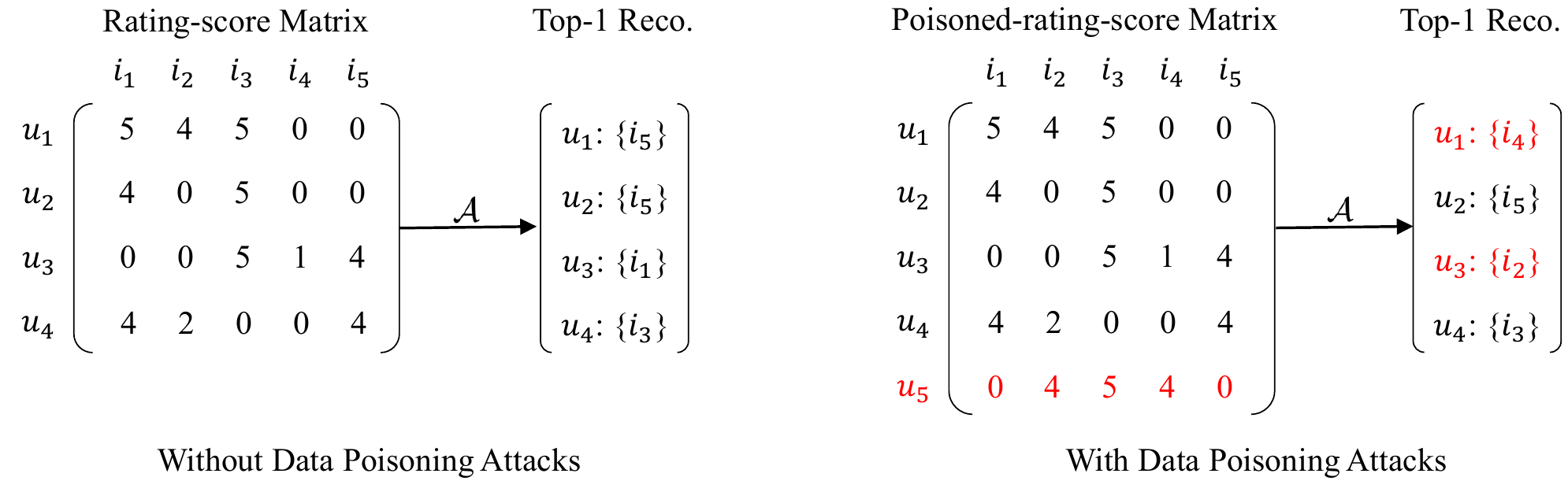}} 
\caption{\emph{Left}: a recommender system without data poisoning attacks.  \emph{Right}: an attacker manipulates the recommended items for  users $u_1$ and $u_3$ by injecting a fake user (i.e., $u_5$) into the system.}
\label{fig:attacks}
\end{figure*}

\subsection{Recommender Systems}
\myparatight{Rating-score matrix} Suppose we have $n$ users and $m$ items which are denoted as $\mathcal{U}=\{u_1, u_2, \cdots, u_n\}$ and $\mathcal{I}=\{i_1, i_2, \cdots, i_m\}$, respectively. We use $\mathbf{M}$ to denote the  rating-score matrix which has $n$ rows and $m$ columns, where a row corresponds to a user and a column corresponds to an item. Essentially, the rating-score matrix $\mathbf{M}$ captures the users' interests towards different items.  In particular, an entry $\mathbf{M}_{ui}$  represents the rating score that the user $u$ gave to the item $i$. For instance, the rating score could be an integer in the range $[1, 5]$, where $1$ is the lowest rating score and denotes that the item does not attract the user's interest, and $5$ is the largest rating score and denotes that the item attracts the user's interest substantially. We note that our method is applicable to any type of rating scores, e.g., binary, integer-valued, and continuous. $\mathbf{M}_{ui} = 0$ means that the user $u$ has not rated the item $i$ yet. For convenience, we denote by $\mathcal{R}$ the domain of a rating score including 0, i.e., $\mathbf{M}_{ui} \in \mathcal{R}$.  

\myparatight{Top-$N$ recommended items} A recommender system aims to help users discover new items that may arouse his/her interests. 
A recommender system algorithm takes the rating-score matrix $\mathbf{M}$ as input and recommends top-$N$ items to each user that he/she has not rated yet but is potentially interested in. 
For simplicity, we use $\mathcal{A}$ to denote a recommender system algorithm. Moreover, we use $\mathcal{A}(\mathbf{M},u)$ to denote the set of top-$N$ items recommended to  user $u$ when the recommender system  is built by $\mathcal{A}$ on $\mathbf{M}$.

\myparatight{Recommender system algorithms} Many algorithms have been proposed to build recommender systems such as Item-based Recommendation (IR)~\cite{sarwar2001item,linden2003amazon,argyriou2020microsoft}, Bayesian Personalized Ranking (BPR)~\cite{rendle2012bpr}, Matrix Factorization~\cite{koren2009matrix}, Neural Collaborative Filtering (NCF)~\cite{he2017neural}, and LightGCN~\cite{he2020lightgcn}. For instance, {IR} 
calculates the similarities between different items based on their rating scores, predicts users' missing rating scores using  such similarities, and recommends a user the $N$ items that the user has not rated yet but have the largest predicted rating scores. Due to its scalability, IR has been widely deployed in industry, e.g., Amazon~\cite{linden2003amazon}. 
According to recent benchmarks released by Microsoft~\cite{Recommenders_url}, BPR achieves state-of-the-art  performance, e.g., BPR even outperforms more complex algorithms such as NCF~\cite{he2017neural} and LightGCN~\cite{he2020lightgcn}.

\subsection{Data Poisoning Attacks}
Many studies~\cite{lam2004shilling,mobasher2007toward,li2016data,yang2017fake,fang2018poisoning,fang2020influence,huangdata2021} 
showed that recommender systems are not robust to data poisoning attacks (Section~\ref{related} discusses more details).  
In a data poisoning attack, an attacker creates fake users in a recommender system and assigns carefully crafted rating scores to them, such that the recommender system, which is built based on the rating scores of  genuine and fake users,  makes attacker-desired, arbitrary recommendations. 
For instance, a data poisoning attack could substantially degrade the performance of a recommender system; and a data poisoning attack could promote certain items (e.g., videos on YouTube and products on Amazon) via spoofing a recommender system  to recommend  them to many genuine users. \CR{Figure~\ref{fig:attacks} illustrates data poisoning attacks.}

We denote by $\mathbf{M}'$ the \emph{poisoned-rating-score matrix}.  A data poisoning attack  aims to reduce the intersection between  $\mathcal{A}(\mathbf{M},u)$ and $\mathcal{A}(\mathbf{M}',u)$ via carefully designing the rating scores of the fake users. Different attacks essentially assign different rating scores to the fake users.

\section{Problem Formulation}
\myparatight{Threat model}
We assume an attacker can inject fake users into a recommender system via registering and maintaining fake accounts~\cite{thomas2013trafficking}. 
We consider an attacker can inject at most $e$ fake users into a recommender system, e.g., because of limited resources to register and maintain fake accounts. However, we assume each fake user can arbitrarily rate as many items as the attacker desires. Moreover, we assume the attacker has whitebox access to the recommender system, e.g., the attacker has access to the rating scores of all genuine users as well as the recommender system algorithm and its parameters. In other words, we consider strong attackers, who can perform any data poisoning attacks. 

A poisoned-rating-score matrix $\mathbf{M}'$ extends the rating-score matrix  $\mathbf{M}$ by at most $e$ rows, which correspond to the rating scores of the at most $e$ fake users. Different data poisoning attacks essentially select different rating scores for the fake users and  result in different poisoned-rating-score matrix $\mathbf{M}'$. We use $\mathcal{L}(\mathbf{M},e)$ to denote the set of all possible poisoned-rating-score matrices when the clean rating-score matrix is $\mathbf{M}$ and the number of fake users is at most $e$. $\mathcal{L}(\mathbf{M},e)$ essentially denotes all possible data poisoning attacks with at most $e$ fake users.  Formally, we define $\mathcal{L}(\mathbf{M},e)$ as follows: 
\begin{align}
   \mathcal{L}(\mathbf{M},e) = \{\mathbf{M}' | \mathbf{M}'_{ui}=\mathbf{M}_{ui} \text{ and } \mathbf{M}'_{vi} \in \mathcal{R}, \nonumber\\ \forall u\in \mathcal{U}, v\in \mathcal{V}, i\in \mathcal{I} \}, 
\end{align}
where $\mathcal{R}$ is the domain of a rating score, $\mathcal{U}$ is the set of genuine users, $\mathcal{V}$ is the set of at most $e$ fake users (i.e., $|\mathcal{V}|\leq e$), and $\mathcal{I}$ is the set of items.

\begin{figure*}[!t]
	 \centering
{\includegraphics[width=0.9\textwidth]{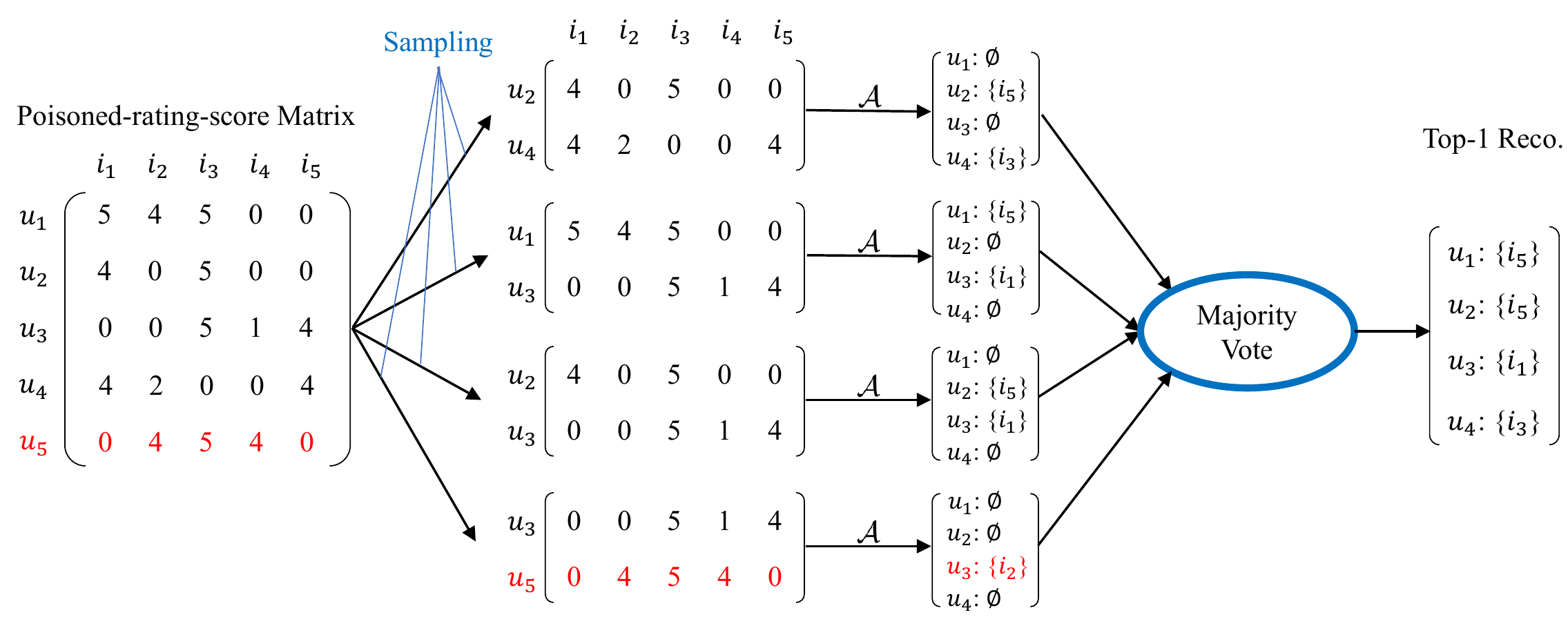}} 
\caption{Robustness of our ensemble recommender system against data poisoning attacks.}
\label{fig:overrview}
\end{figure*}

\myparatight{Provably robust recommender system algorithm}
We say a recommender system algorithm is provably robust against data poisoning attacks if a certain number of its recommended top-$N$ items for a user are provably unaffected by any data poisoning attacks. 
Specifically, given a set of items $\mathcal{I}_u$, we say a recommender system algorithm $\mathcal{A}$ is provably robust for a user $u$ if it satisfies the following property: 
\begin{align}
\label{intersectionsizeproperty}
     \min_{\mathbf{M}' \in \mathcal{L}(\mathbf{M},e)} |\mathcal{I}_{u} \cap \mathcal{A}(\mathbf{M}',u)|\geq r,
\end{align}
where $\mathcal{L}(\mathbf{M},e)$ is the set of all possible poisoned-rating-score matrices (i.e., all possible data poisoning attacks with at most $e$ fake users),  $|\mathcal{I}_{u} \cap \mathcal{A}(\mathbf{M}',u)|$ is the size of  the intersection between $\mathcal{I}_{u}$ and the top-$N$ items recommended to $u$ by $\mathcal{A}$ under attacks, and $r$ is called \emph{certified intersection size}. Note that $r$ may depend on the user $u$ and the number of fake users $e$, but we omit its explicit dependency on $u$ and $e$ for simplicity.  

When $\mathcal{I}_u$ is the set of top-$N$ items recommended to user $u$ by $\mathcal{A}$ under no attacks, i.e., $\mathcal{I}_u=\mathcal{A}(\mathbf{M},u)$, our provable robustness means that at least $r$ of the $N$ items in $\mathcal{A}(\mathbf{M},u)$ are still recommended to $u$ under any attacks with at most $e$ fake users. When $\mathcal{I}_u$ is the set of ground truth test items for $u$ (i.e., the set of items that $u$ is indeed interested in), our provable robustness means that at least $r$ of the ground truth test items are recommended to $u$ under attacks. As we will discuss more details in experiments,   $r$ in the latter case can be used to derive a lower bound of recommendation performance such as Precision@$N$,  Recall@$N$, and F1-Score@$N$ under any data poisoning attacks. We formally define a provably robust recommender system algorithm as follows: 
\begin{definition}[$(e,r)$-Provably Robust Recommender System]
Suppose the number of fake users is at most $e$ and a data poisoning attack can arbitrarily craft the rating scores for the fake users. We say a recommender system algorithm $\mathcal{A}$ is $(e,r)$-provably robust for a user $u$ if its certified intersection size for $u$ is at least $r$, i.e., if Equation~(\ref{intersectionsizeproperty}) is satisfied.
\end{definition}

%% file: method.tex
\section{Our PORE}
\label{method}
\CR{We first give an overview of our PORE, then define our ensemble recommender system and show it is $(e,r)$-provably robust, and finally describe our algorithm to compute the certified intersection size $r$ for each user. }

\subsection{\CR{Overview}}
Our key intuition is that, when the number of fake users is bounded, a random subset of a small number of users is likely to only include genuine users. Therefore, a recommender system built using the rating scores of such a random subset of users is not affected by fake users. 
Based on the intuition,  PORE builds multiple recommender systems using random subsets of users and takes majority vote among them to recommend items for users.

Specifically, we create multiple submatrices from a  rating-score matrix, where each submatrix contains the rating scores of $s$ different randomly selected users, i.e., $s$ rows randomly selected from the   rating-score matrix. Then, we use an arbitrary base algorithm  to build a recommender system (called \emph{base recommender system}) on each submatrix and use it to recommend items for users. Finally, we build an \emph{ensemble recommender system}, which takes a majority vote among the base recommender systems as the final recommended items for each user. Figure~\ref{fig:overrview} shows our ensemble recommender system and  its robustness against data poisoning attacks.

Next, we first formally define our ensemble recommender system in PORE. Then, we show   PORE is provably robust against data poisoning attacks. In particular, given an arbitrary set of items $\mathcal{I}_u$ for a user $u$, we prove that at least $r$ of the top-$N$ items recommended to $u$ by PORE are guaranteed to be among $\mathcal{I}_u$ under any data poisoning attacks.  Moreover, we derive the certified intersection size $r$. Finally, we design an algorithm to compute the certified intersection sizes $r$ for all users simultaneously.

\subsection{Our Ensemble Recommender System}
\vspace{-2mm}
\myparatight{Item probability $p_i$} Given a  rating-score matrix $\mathbf{M}$, 
we randomly sample a submatrix with $s$ rows of $\mathbf{M}$, i.e., the submatrix consists of the rating scores of $s$ different randomly selected users. 
For convenience, we denote by $\mathbf{X}$ the sampled submatrix. 
Then, we use an arbitrary base algorithm $\mathcal{A}$ to build a base recommender system on the sampled submatrix $\mathbf{X}$. 
We use the base recommender system to recommend top-$N'$ items (denoted as $\mathcal{A}(\mathbf{X}, u)$) to each user $u$ in the sampled submatrix. Since the submatrix is randomly sampled, the recommended top-$N'$ items are also random. To consider such randomness, we denote by $p_{i}$ the probability that item $i$ is recommended to a user $u$. Formally, we define $p_{i}$ as follows: $p_{i} = \text{Pr}(i \in \mathcal{A}(\mathbf{X},u))$. We call $p_i$ \emph{item probability}. 

Note that we consider $\mathcal{A}(\mathbf{X},u)$ is an empty set when $u$ is not in the sampled submatrix $\mathbf{X}$, as many recommender systems make recommendations for users in the rating-score matrix that was used to build the recommender systems.  
Since a base recommender system is built using $s$ rows of the rating-score matrix $\mathbf{M}$, we can build ${n \choose s}$ base recommender systems in total, where $n$ is the total number of users/rows in $\mathbf{M}$. Essentially, our item probability $p_i$ is the fraction of the ${n \choose s}$ base recommender systems that recommend item $i$ to the user $u$.

\myparatight{Poisoned item probability $p_{i}'$} Under data poisoning attacks, the rating-score matrix $\mathbf{M}$ becomes a poisoned version $\mathbf{M}^{\prime}$. 
We denote by $\mathbf{Y}$ a random submatrix with $s$ rows sampled from $\mathbf{M}^{\prime}$. Moreover, we define \emph{poisoned item probability} $p_{i}' =  \text{Pr}(i \in \mathcal{A}(\mathbf{Y},u))$, i.e., $p_{i}'$ is the probability that the item $i$ is in the top-$N'$ items recommended to $u$ when the base recommender system is built based on  $\mathbf{Y}$.

\myparatight{Our ensemble recommender system $\mathcal{T}$} Our ensemble recommender system (denoted as $\mathcal{T}$) recommends the top-$N$ items with the largest item probabilities $p_{i}$'s for a user $u$. Essentially, our ensemble recommender system takes a majority vote among the ${n \choose s}$ base recommender systems. In particular, our ensemble recommender system essentially recommends the top-$N$ items that are the most frequently recommended by the ${n \choose s}$ base recommender systems to a user. 
For simplicity, we use $\mathcal{T}(\mathbf{M},u)$ to denote the set of top-$N$ items recommended to the user $u$ by our ensemble recommender system $\mathcal{T}$ when the rating-score matrix is $\mathbf{M}$. Note that $N'$ and $N$ are different parameters, i.e., $N'$ is the number of items recommended to a user by a base recommender system while $N$ is the number of items recommended to a user by our ensemble recommender system. We will explore their impact on  our ensemble recommender system in experiments. 

 Under data poisoning attacks, our ensemble recommender system $\mathcal{T}$ uses the poisoned item probabilities to make recommendations. Specifically, $\mathcal{T} (\mathbf{M}', u)$ is the set of top-$N$ items with the largest poisoned item probabilities $p_{i}'$'s that are recommended to $u$ by  $\mathcal{T}$ under attacks. 

\subsection{Deriving the Certified Intersection Size}
\label{sec:derive_certified_performance}

We show that our ensemble recommender system algorithm $\mathcal{T}$ is $(e,r)$-provably robust. In particular, for any given number of fake users $e$,  we can derive the certified intersection size $r$ of  $\mathcal{T}$ for any user $u$. Specifically, given an arbitrary set of items $\mathcal{I}_u$, 
   we show that at least $r$ of the recommended top-$N$ items $\mathcal{T} (\mathbf{M}', u)$ are in $\mathcal{I}_u$ when there are  at most $e$ fake users, no matter what rating scores they use. \neil{Later, we can replace $\mathcal{I}_u$ as our desired sets of items.} Formally, we aim to show the following: $\min_{\mathbf{M}' \in \mathcal{L}(\mathbf{M},e)} |\mathcal{I}_{u} \cap \mathcal{T} (\mathbf{M}',u)| \geq r$, 
where $\mathcal{L}(\mathbf{M},e)$ is the set of all possible poisoned-rating-score matrices  and denotes all data poisoning attacks with at most $e$ fake users. 
Next, we first overview our main idea  and then show our theorem.

\myparatight{Overview of our derivation} Our proof is based on the \emph{law of contraposition}. Suppose we have a statement: $U \longrightarrow V$, whose contraposition is $\neg V \longrightarrow \neg U$. The law of contraposition means that a statement is true if and only if its contraposition is true. We define a predicate $V$ as $V: |\mathcal{I}_{u} \cap \mathcal{T} (\mathbf{M}',u)| \geq r$. The predicate $V$ is true if at least $r$ of the top-$N$ items recommended by $\mathcal{T}$ for $u$ are   in $\mathcal{I}_{u}$ when the poisoned-rating-score matrix is a given $\mathbf{M}'$. Then, we derive a necessary condition (denoted as $\neg U$)  for $\neg V$ to be true, i.e., we have $\neg V \longrightarrow \neg U$.   By the law of contraposition, we know that $V$ is true if  $U$ is true. 
Roughly speaking, \neil{$U$} means that the $r$th largest poisoned item probability for items in $\mathcal{I}_u$ is \neil{larger} than the $(N - r +1)$th largest poisoned item probability for items in $\mathcal{I}\setminus\mathcal{I}_u$, when the poisoned-rating-score matrix is $\mathbf{M}'$. 

The challenge in deriving the condition \neil{for $U$ to be true} is that it is hard to compute the poisoned item probabilities $p_{i}^{\prime}$'s due to the complexity of recommender system. To address the challenge, we resort to derive a lower bound of  $p_{i}^{\prime}$ for each $i\in \mathcal{I}_u$ and an upper bound of $p_j^{\prime}$ for each $j\in \mathcal{I} \setminus \mathcal{I}_u$. In particular, we derive the lower/upper bounds of poisoned item probabilities using  lower/upper bounds of item probabilities. We consider lower/upper bounds of item probabilities instead of their exact values, because it is challenging to compute them exactly. Suppose we have a lower bound $\underline{p_i}$ of $p_i$ for each $i \in \mathcal{I}_u$ and an upper bound $\overline{p}_j$ of $p_j$ for each $j \in \mathcal{I} \setminus \mathcal{I}_u$, i.e., we have the following: 
\begin{align}
\label{probability_upper_low_bound_theorem_1}
 p_i \geq  \underline{p_i}  \text{ and } p_j \leq \overline{p}_j, 
\end{align}
where $i \in \mathcal{I}_u$ and $j \in \mathcal{I} \setminus \mathcal{I}_u$. 
In the next section, we design an algorithm to estimate such lower/upper bounds of item probabilities. 
Given the lower/upper bounds $\underline{p_i}$ and $\overline{p}_j$, we derive a lower bound of $p_i^{\prime}$ for each $i\in \mathcal{I}_u$ and an upper bound of $p_j^{\prime}$ for each $j\in \mathcal{I} \setminus \mathcal{I}_u$ via  a variant Neyman-Pearson Lemma~\cite{neyman1933ix} that we develop. Our variant  is applicable to multiple functions, while the standard Neyman-Pearson Lemma is only applicable to one function.

Next, we show our intuition to derive the upper and lower bounds (please refer to the proof of Theorem~\ref{theorem_certified_size} for formal analysis) of the poisoned item probabilities. 
We denote by $\Phi$ the union of the domain spaces of $\mathbf{X}$ and $\mathbf{Y}$, i.e., each element in $\Phi$ is a submatrix with $s$ rows sampled from $\mathbf{M}$ or $\mathbf{M}^{\prime}$. Our idea is to find subsets in  $\Phi$ such that we can apply our variant of the Neyman-Pearson Lemma to derive the upper/lower bounds of the poisoned item probabilities. Moreover, the upper/lower bounds are related to the probabilities that the random submatrices $\mathbf{X}$ and $\mathbf{Y}$ are in the subsets, which can be easily computed. 
We denote by $P$ (or $A$)  the set of submatrices sampled from $\mathbf{M}$ (or $\mathbf{M}^{\prime}$) that include the user $u$. 

\myparatight{Deriving a lower bound of $p_i^{\prime}$, $ i \in \mathcal{I}_u$} We can find a subset $C_i \subseteq P$ such that we have $\text{Pr}(\mathbf{X} \in C_i) = p_i^* \triangleq \frac{\lfloor \underline{p_i} \cdot {n \choose s}\rfloor}{{n \choose s}}$. 
Note that we can find such a subset because $p_i^*$ is an integer multiple of $1/{n \choose s}$. Then, via our variant of the Neyman-Pearson Lemma, we can derive a lower bound of $p_i^{\prime}$ using the probability that the random submatrix $\mathbf{Y}$ is in the subset $C_i$, i.e., we have:  $p_i^{\prime} \geq \text{Pr}(\mathbf{Y} \in C_i), \ \forall i\in \mathcal{I}_u$. 

\myparatight{Deriving an upper bound of $p_j^{\prime}$, $j \in \mathcal{I} \setminus \mathcal{I}_u$} We first find a subset $C_j^{\prime} \subseteq P$ such that we have the following: $\text{Pr}(\mathbf{X} \in C_j^{\prime}) = \overline{p}^{*}_j = \frac{\lceil \overline{p}_j \cdot {n \choose s}\rceil}{{n \choose s}}$. 
Given the subset $C_j^{\prime}$, we further define a subset $C_j = C_j^{\prime} \cup (A \setminus P)$. Then, based on our variant of the Neyman-Pearson Lemma,  we  derive an upper bound of $p_j^{\prime}$ using the probability that the random submatrix $\mathbf{Y}$ is in the subset $C_j$, i.e., we have:
\begin{align}
\label{upper_bound_single_mainbody}
    p_j^{\prime} \leq \text{Pr}(\mathbf{Y} \in C_j).
\end{align} 
In our derivation, we further improve the upper bound via jointly considering multiple items in $\mathcal{I} \setminus \mathcal{I}_u$. Suppose $\mathcal{H}_{c} \subseteq \mathcal{I}\setminus\mathcal{I}_u$ is a set of $c$ items.  We denote
$\overline{p}_{\mathcal{H}_{c}}=\sum_{j \in \mathcal{H}_{c}}\overline{p}_{j}$.
Then, we can find a subset $C_{\mathcal{H}_{c}}^{\prime}$ such that we have the following:  $\text{Pr}(\mathbf{X}\in C^{\prime}_{\mathcal{H}_{c}}) = \overline{p}^{*}_{\mathcal{H}_{c}} \triangleq \frac{\lceil (\overline{p}_{\mathcal{H}_{c}}/N')\cdot {n \choose s} \rceil}{ {n \choose s}  }$. 
Given the subset $C_{\mathcal{H}_{c}}^{\prime}$, we further define a subset $C_{\mathcal{H}_{c}} = C_{\mathcal{H}_{c}}^{\prime} \cup (A \setminus P)$.
Then,  we have the following upper bound for the smallest poisoned item probability in the set $\{p_j^{\prime} | j \in \mathcal{H}_{c}\}$:
\begin{align}
\label{upper_bound_combined_mainbody}
    \min_{j \in \mathcal{H}_{c}} p_{j}^{\prime} \leq \frac{N' \cdot \text{Pr}(\mathbf{Y}\in C_{\mathcal{H}_{c}})}{c}.
\end{align}
Finally, we can combine the upper bounds in Equation~(\ref{upper_bound_single_mainbody}) and~(\ref{upper_bound_combined_mainbody}) to derive an upper bound of the $(N-r+1)$th largest poisoned item probability in $\mathcal{I} \setminus \mathcal{I}_u$. Note that we don't jointly consider multiple items in $\mathcal{I}_u$ when deriving the lower bounds for poisoned item probabilities in $\mathcal{I}_u$ because it does not improve the lower bounds.

Formally, we have the following theorem: 
\begin{theorem}
\label{theorem_certified_size}
Suppose we have a rating-score matrix $\mathbf{M}$, a user $u$, 
 and an arbitrary set of $k$ items $\mathcal{I}_u=\{\mu_1,\mu_2,\cdots,\mu_k\}$.  
Furthermore, we have a lower bound  $\underline{p_{i}}$ for each $i \in \mathcal{I}_u$ and an upper bound $\overline{p}_{j}$ for each $j \in \mathcal{I}\setminus \mathcal{I}_u$ that satisfy  Equation~(\ref{probability_upper_low_bound_theorem_1}). 
Without loss of generality, we assume $\underline{p_{\mu_1}}\geq \underline{p_{\mu_2}} \geq \cdots \geq \underline{p_{\mu_k}}$.
Under any data poisoning attacks with at most  $e$ fake users, we have the following guarantee: $\min_{\mathbf{M}' \in \mathcal{L}(\mathbf{M},e)} |\mathcal{I}_{u} \cap \mathcal{T} (\mathbf{M}', u)| \geq r$,  
where $r$ is the solution to the following optimization problem or $0$ if it does not have a solution:
\begin{align}
\label{optimization_problem_theorem_1}
     r & =\argmax_{r'\in\{1,2,\cdots,\min(k, N)\}} r' \nonumber \\
     \text{s.t. } & 
           \underline{p^{*}_{\mu_{r'}}} >   \min ( \min_{c=1}^{N-r' +1}\frac{ N' \cdot(\overline{p}^{*}_{\mathcal{H}_{c}} + \sigma)}{c}, 
     \overline{p}^{*}_{v_1}+ \sigma),  
\end{align}
where $n' = n + e$, $\sigma = \frac{s}{n'} \cdot \frac{{n' \choose s}}{{n \choose s}} -\frac{s}{n}$ ,  $ \underline{p^{*}_{\mu_{r'}}} =\frac{\lfloor \underline{p_{\mu_{r'}}} \cdot {n \choose s}\rfloor}{{n \choose s}}$, $\mathcal{H}_{c}=\{v_1, v_2, \cdots, v_{c}\}$ is the set of $c$ items that have the smallest item-probability upper bounds among the $N-r'+1$ items with the largest item-probability upper bounds in $\mathcal{I}\setminus \mathcal{I}_u $, $v_1$ is the item in $\mathcal{H}_1$, $\overline{p}_{\mathcal{H}_{c}}=\sum_{j \in \mathcal{H}_{c}}\overline{p}_{j}$, $ \overline{p}^{*}_{\mathcal{H}_{c}} = \frac{\lceil (\overline{p}_{\mathcal{H}_{c}}/N')\cdot {n \choose s} \rceil}{ {n \choose s}  }$, and $\overline{p}^{*}_{v_1} = \frac{\lceil \overline{p}_{v_1} \cdot {n \choose s}\rceil}{{n \choose s}} $. 
\end{theorem}
\begin{proof}
See Appendix~\ref{proof_of_certified_theorem}. 
\end{proof}

\subsection{Computing the Certified Intersection Size}
Given a base algorithm $\mathcal{A}$, a  rating-score matrix $\mathbf{M}$, a set of genuine users $\mathcal{U} = \{u_1, u_2, \cdots, u_n\}$, a set of items $\mathcal{I}_{u}$ for each genuine user $u$, the maximum number of fake users $e$, and a sampling size $s$, we aim to compute the certified intersection size of our ensemble recommender system for each user in $\mathcal{U}$. The key to compute the certified intersection size is to solve $r$ in the optimization problem in Equation~(\ref{optimization_problem_theorem_1}). Specifically, given a user $u$, the key challenge to solve the optimization problem in Equation~(\ref{optimization_problem_theorem_1}) is how to estimate the 
item-probability lower bounds $\underline{p_i}$ for $\forall i \in \mathcal{I}_u$ and upper bounds $\overline{p}_j$ for $\forall j \in \mathcal{I}\setminus \mathcal{I}_{u}$. One naive way is to build ${n \choose s}$ base recommender systems and compute the exact item probabilities. However, such approach is computationally infeasible as ${n \choose s}$ is huge. To address the challenge, we design an algorithm to estimate lower/upper bounds of item probabilities via building $T << {n \choose s}$ base recommender systems. Next, we introduce  estimating the lower/upper bounds of the item probabilities,  solving $r$ using the estimated item-probability bounds, and our complete algorithm.

\myparatight{Estimating the item-probability lower/upper bounds} We randomly sample $T$ submatrices from $\mathbf{M}$, where each submatrix contains  $s$ rows of $\mathbf{M}$. For simplicity, we denote them as $\Gamma_1, \Gamma_2, \cdots, \Gamma_T$. Then, we build a base recommender system for each submatrix $\Gamma_{t}$ using the base  algorithm $\mathcal{A}$, where $t=1, 2, \cdots, T$. Given a user $u$, we use each base recommender system to recommend top-$N'$ items for the user. We denote by $\mathcal{A}(\Gamma_t, u)$  the set of top-$N'$ items recommended to the user $u$ by the base recommender system built on the submatrix $\Gamma_t$. Note that  $\mathcal{A}(\Gamma_t, u)$ is empty if the user $u$ is not in the submatrix $\Gamma_t$. We denote by $T_i$ the frequency of an item $i$  among the recommended top-$N'$ items of the $T$ base recommender systems, i.e., $T_i$ is the number of base recommender systems whose top-$N'$ recommended items for $u$ include $i$. Based on the definition of  item probability $p_i$, the frequency $T_i$ follows a binomial distribution with parameters $T$ and $p_i$,  i.e., we have the following: $\text{Pr}(T_i = t) = {T \choose t} \cdot p_i^t \cdot (1 - p_i)^{T - t}$, 
where $t=0, 1,  \cdots, T$. 
 Our goal is to estimate a lower or upper bound of $p_i$ given $T_i$ and $T$. This is essentially a \emph{binomial proportion confidence interval} estimation problem. Therefore, we can leverage the standard Clopper-Pearson method~\cite{clopper1934use} to estimate a lower or upper bound of $p_i$ from a given $T_i$ and $T$. Formally, we have the following:
\begin{align}
\label{cp_lower_bound}
    \underline{p_i} & = \text{Beta}(\frac{\alpha_u}{m}; T_i, T-T_i+1), i\in \mathcal{I}_u, \\
    \label{cp_upper_bound}
    \overline{p}_j  &= \text{Beta}(1-\frac{\alpha_u}{m}; T_j, T-T_j+1), j \in \mathcal{I}\setminus \mathcal{I}_{u},  
\end{align}
where  $1 - \alpha_u/m$ is the confidence level for estimating the lower/upper bound of one item probability, $m$ is the total number of items, and $\text{Beta}(\beta;\varsigma, \vartheta)$ is the $\beta$th quantile of the Beta distribution with shape parameters $\varsigma$ and $\vartheta$. Based on  \emph{Bonferroni correction} in statistics, the \emph{simultaneous confidence level} of estimating the lower/upper bounds of the $m$ item probabilities is at least $1 - \alpha_u$. Given an item set $\mathcal{H}_{c}$ defined in Theorem~\ref{theorem_certified_size}, we can estimate $\overline{p}_{\mathcal{H}_{c}}$ as  $\overline{p}_{\mathcal{H}_{c}} = \min(\sum_{j \in \mathcal{H}_{c}} \overline{p}_{j}, N' - \sum_{i \in \mathcal{I}_u} \underline{p_i} )$,  where both $\sum_{j \in \mathcal{H}_{c}} \overline{p}_{j}$ and $N' - \sum_{i \in \mathcal{I}_u} \underline{p_i}$ are upper bounds of ${p}_{\mathcal{H}_{c}}$, and we use the smaller one.

\begin{algorithm}[tb]
   \caption{\textsc{BinarySearch}}
   \label{alg:binary_search}
\begin{algorithmic}
   \STATE {\bfseries Input:} $e$, $s$, $N'$, $N$, $ \mathcal{I}_{u},\{\underline{p_i}|i\in \mathcal{I}_{u}\}$, and $ \{\overline{p}_j|j\in \mathcal{I}\setminus\mathcal{I}_{u}\}$
   \STATE {\bfseries Output:}  $r_u$ \\
   $low, high \gets 1, \min(|\mathcal{I}_u|,N)$ \\
   \WHILE{$low < high $}
   \STATE $r' =\lceil (low+high)/2 \rceil$ \\
   \IF{\textsc{VerifyConstraint($r',e, s, N', N,\mathcal{I}_{u},\{\underline{p_i}|i\in \mathcal{I}_{u}\},\{\overline{p}_j|j\in \mathcal{I}\setminus\mathcal{I}_{u}\}$)} $== 1$}
   \STATE $low \gets r' $ \\
   \ELSE 
   \STATE $high \gets r' - 1$ \\
   \ENDIF
   \ENDWHILE
    \IF{\textsc{VerifyConstraint($r',e, s, N', N,\mathcal{I}_{u},\{\underline{p_i}|i\in \mathcal{I}_{u}\},\{\overline{p}_j|j\in \mathcal{I}\setminus\mathcal{I}_{u}\}$)} $== 1$}
   \STATE return $ r' $ \\
   \ELSE 
   \STATE return $ 0$ \\
   \ENDIF
\end{algorithmic}
\end{algorithm}

\myparatight{Solving the optimization problem}  We note that  Equation~(\ref{optimization_problem_theorem_1}) has the following property: its left-hand side and right-hand side respectively decreases and increases as $r'$ increases. Thus, given the estimated item-probability bounds, we can efficiently solve the optimization problem in Equation~(\ref{optimization_problem_theorem_1}) via binary search to obtain $r_u$ for each user $u$. Algorithm~\ref{alg:binary_search} shows our \textsc{BinarySearch} algorithm. The function \textsc{VerifyConstraint}  verifies whether the constraint in Equation~(\ref{optimization_problem_theorem_1}) is satisfied for a given $r'$ and returns 1 if so.

\begin{algorithm}[tb]
   \caption{\textsc{Compute $r$}}
   \label{alg:certify}
\begin{algorithmic}
   \STATE {\bfseries Input:} $\mathbf{M}$, $s$, $T$, $\mathcal{A}$, $N'$, $\alpha$, $e$, $N$, $\mathcal{U}$, and $\{\mathcal{I}_u| u\in \mathcal{U}\}$
   \STATE {\bfseries Output:}  $r_u$ for each user $u\in \mathcal{U}$ \\
   $\Gamma_1, \Gamma_2,\cdots,\Gamma_T \gets  \textsc{RandomSample}(\mathbf{M},s)$ \\
   \FOR{$u$ {\bfseries in} $\mathcal{U}$}
   \STATE counts$[i] \gets \sum_{t=1}^{T}\mathbb{I}(i \in \mathcal{A}(\Gamma_t, u)), i\in \{1,2,\cdots,m\} $ \\
   \STATE $\underline{p_{i}}, \overline{p}_{j} \gets \textsc{BoundEst}(\text{counts},\frac{\alpha}{n}), i \in \mathcal{I}_{u}, j \in \mathcal{I}\setminus \mathcal{I}_{u}$ \\
   \STATE $r_{u} = \textsc{BinarySearch}(e, s, N', N,\mathcal{I}_{u},\{\underline{p_i}|i\in \mathcal{I}_{u}\},\{\overline{p}_j|j\in \mathcal{I}\setminus\mathcal{I}_{u}\})$
   \ENDFOR
   \STATE \textbf{return}
 $\{r_u|u\in \mathcal{U}\}$
\end{algorithmic}
\end{algorithm}

\myparatight{Complete algorithm} Algorithm~\ref{alg:certify} shows our complete algorithm to compute the certified intersection size $r_u$ for each user $u\in \mathcal{U}=\{u_1, u_2, \cdots, u_n\}$.
The function \textsc{RandomSample} randomly samples $T$ submatrices,   each of which contains $s$ rows sampled from $\mathbf{M}$ uniformly at random. 
\textsc{BoundEst} estimates the item-probability bounds with a confidence level $1-\frac{\alpha}{n}$, i.e., $\alpha_u = \frac{\alpha}{n}$, for each user $u \in \mathcal{U}$,  based on Equation~(\ref{cp_lower_bound})~-~(\ref{cp_upper_bound}). 
 \textsc{BinarySearch} solves the optimization problem in Equation~(\ref{optimization_problem_theorem_1}) via binary search to obtain $r_u$ for $u$ based on the estimated item-probability bounds. \neil{Note that Algorithm~\ref{alg:certify} requires a clean rating-score matrix $\mathbf{M}$, which may be sampled from the clean data distribution.}

Due to randomness, the estimated item-probability bounds may be incorrect, e.g., an estimated item-probability lower bound is larger than the true item probability for some item and some user or an estimated item-probability upper bound is smaller than the true item probability for some item and some user. When such estimation error happens for a user $u$, the solved certified intersection size $r_u$ is incorrect for $u$. Since the simultaneous confidence level of estimating the item-probability bounds in our algorithm is at least $1-\frac{\alpha}{n}$ for any user $u \in \mathcal{U}$, the probability of having at least one incorrectly estimated item-probability bound  for any user $u \in \mathcal{U}$ is at most $\frac{\alpha}{n}$.  
Moreover, our following theorem shows that the probability of having an incorrect certified intersection size $r_u$ for at least one user in $\mathcal{U}$ is bounded by $\alpha$: 
\begin{theorem}
\label{proposition_1}
The probability that our Algorithm~\ref{alg:certify} returns an incorrect $r_u$
for at least one user in $\mathcal{U}$ is at most $\alpha$.
\end{theorem}
\begin{proof}
See Appendix~\ref{proof_of_proposition}. 
\end{proof}

%% file: evaluation.tex
\section{Evaluation}
\label{evaluation}
\subsection{Experimental Setup}
\vspace{-2mm}
\myparatight{Datasets} We mainly evaluate  PORE on MovieLens-100k and MovieLens-1M benchmark datasets~\cite{movielens,harper2015movielens}, which consist of around $100,000$ and $1,000,000$ rating scores, respectively.  Specifically, MovieLens-100k  contains $943$ users and $1,682$ items, where each user rated $106$ items on average. MovieLens-1M  contains $6,040$ users and $3,952$ items, where each user on average rated $166$ items. Following~\cite{argyriou2020microsoft}, for each user, we sample $75\%$ of its rating scores as training data and treat its remaining rated items as test items. The users' training data form the rating-score matrix and are used to build recommender systems, while their test items are used to evaluate the performance of the recommended top-$N$ items.  

 \myparatight{Base algorithms} PORE is applicable to any base algorithm. To show such generality, we evaluate  two  base  algorithms, i.e.,  Item-based  Recommendation  (IR) ~\cite{argyriou2020microsoft} and Bayesian Personalized Ranking (BPR)~\cite{rendle2012bpr}. We adopt their public implementations~\cite{argyriou2020microsoft}. 
{We adopt IR because it  has been widely deployed in industry~\cite{linden2003amazon}. We adopt BPR because it achieves state-of-the-art  performance according to recent benchmarks released by Microsoft~\cite{Recommenders_url}.}

\begin{figure*}[!t]
	 \centering
	 \vspace{-2mm}
{\includegraphics[width=0.16\textwidth]{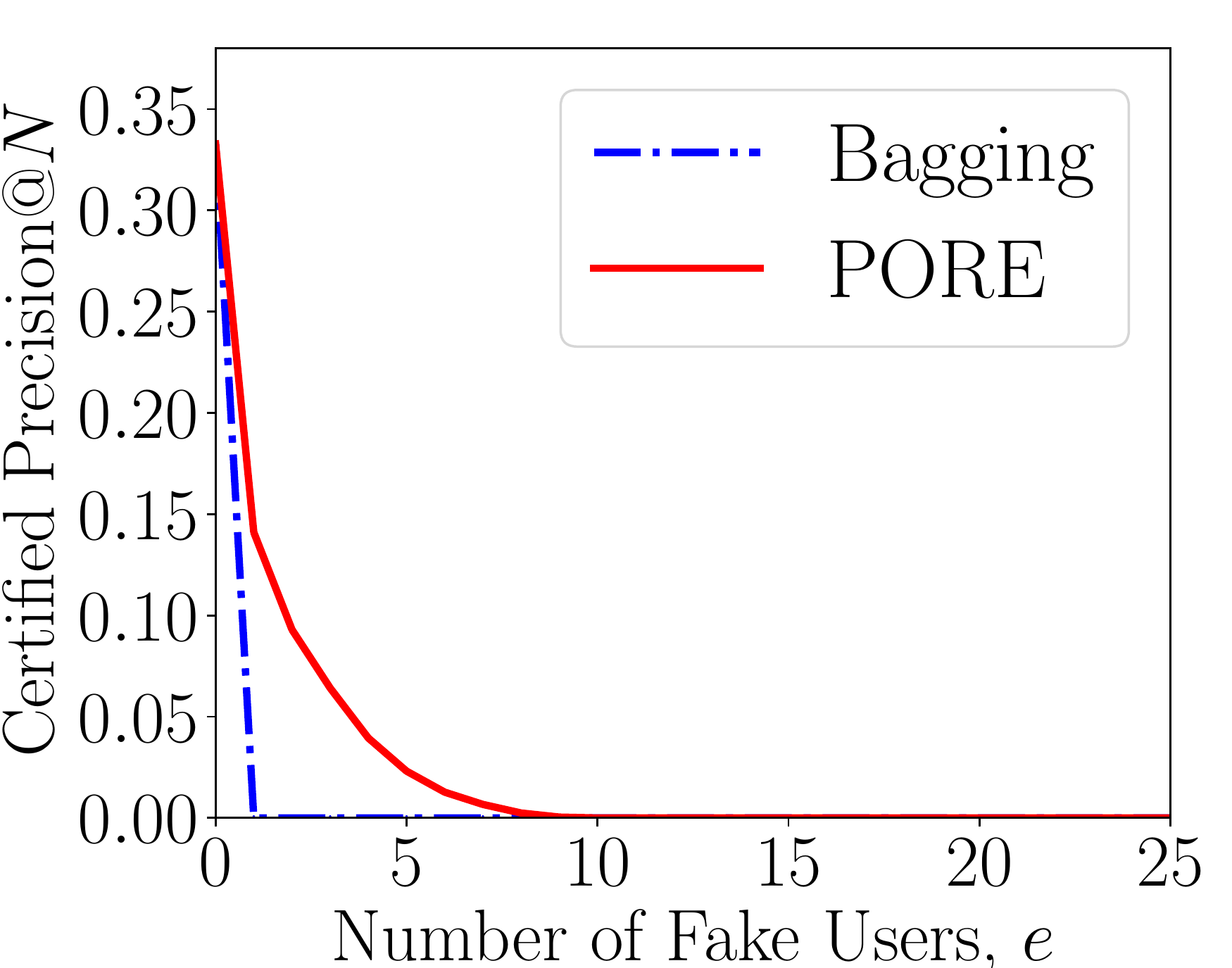}}
{\includegraphics[width=0.16\textwidth]{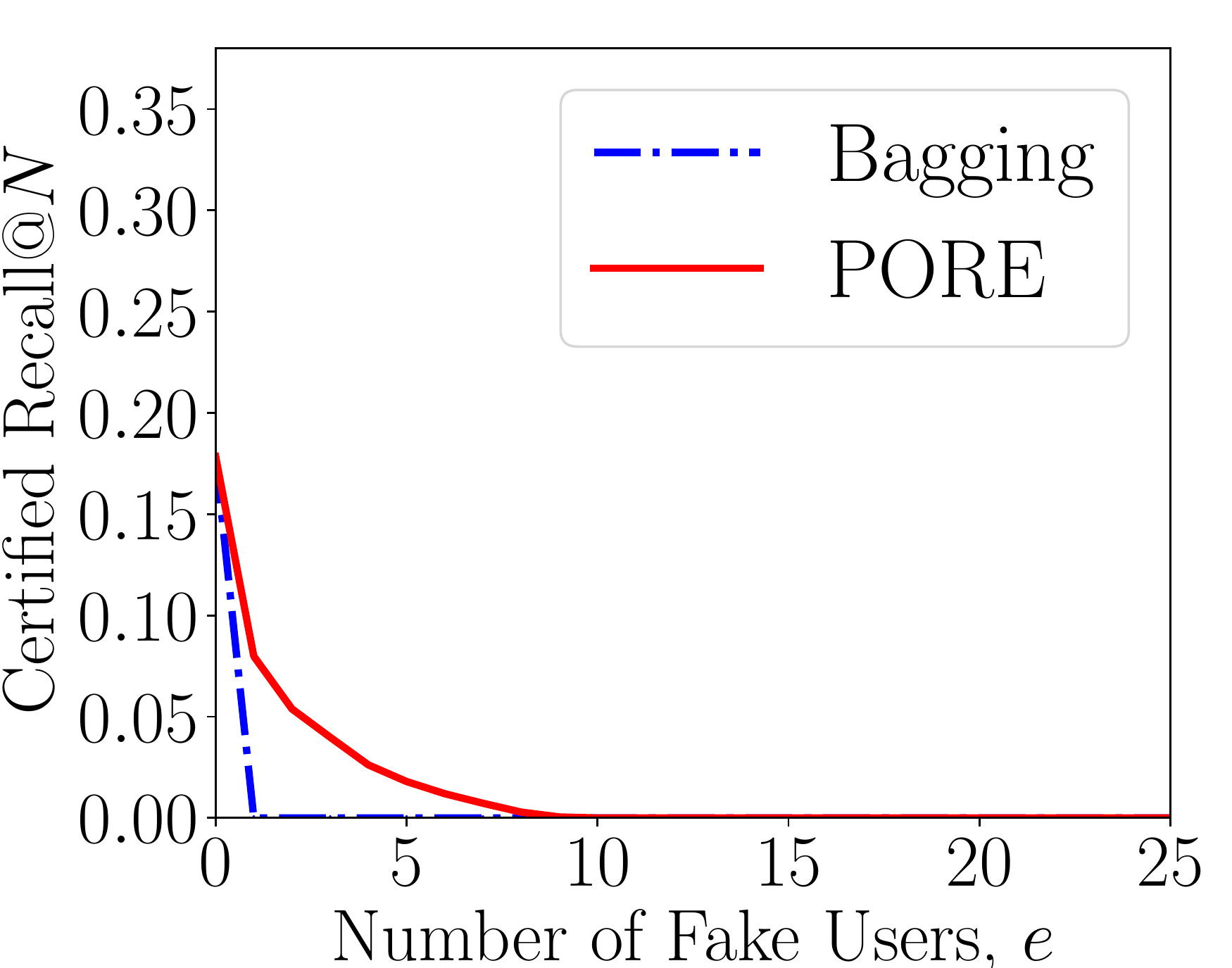}}
{\includegraphics[width=0.16\textwidth]{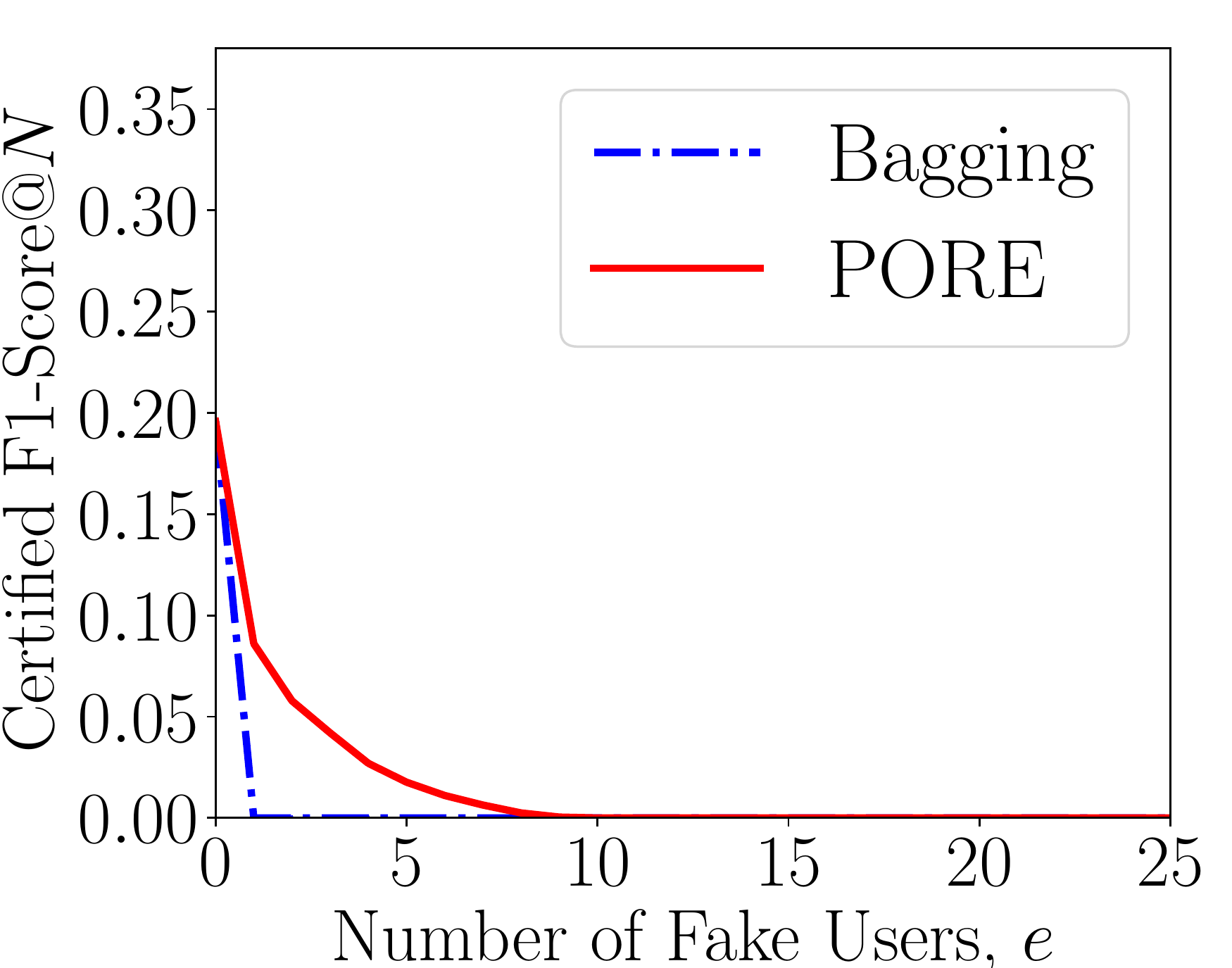}}
{\includegraphics[width=0.16\textwidth]{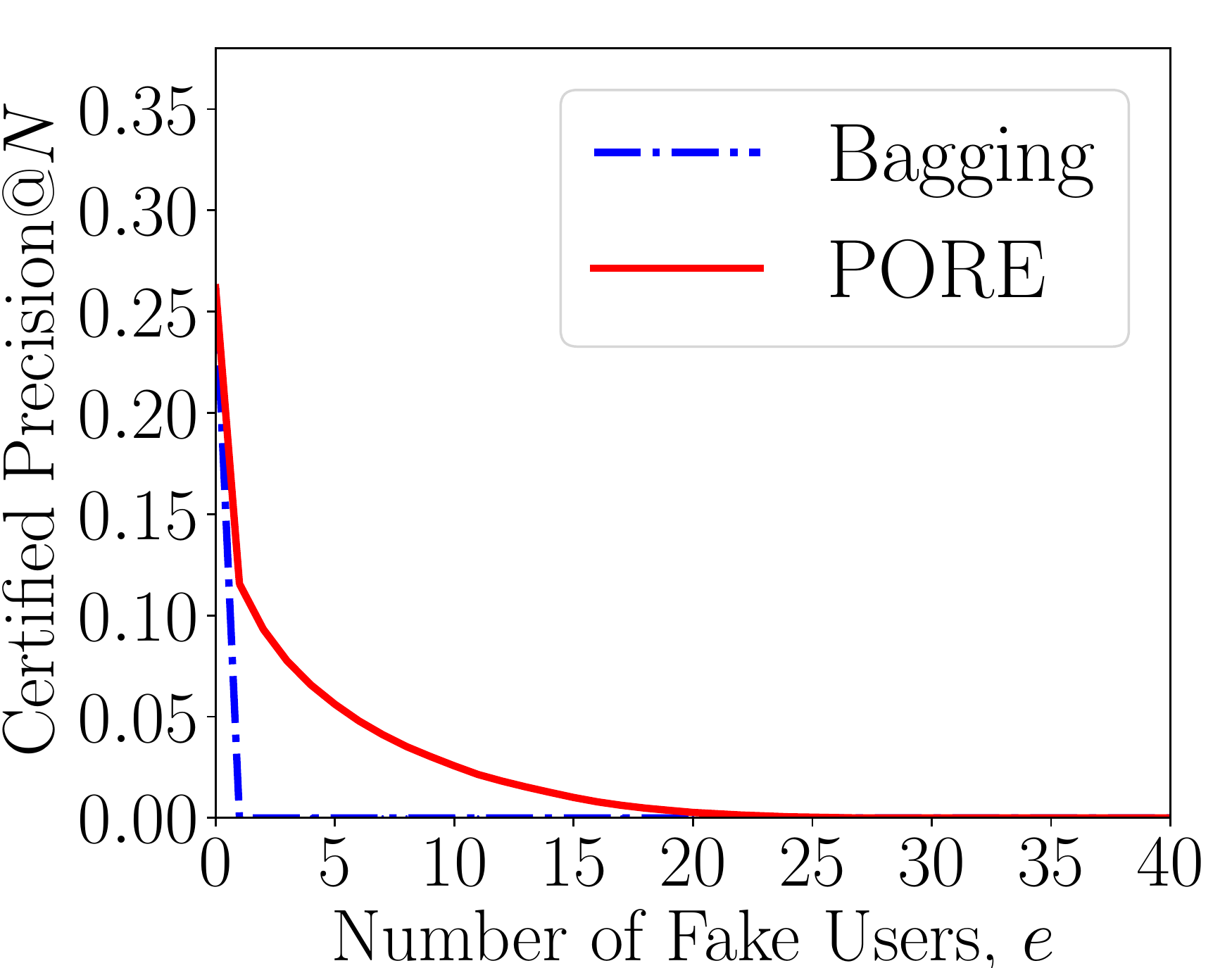}}
{\includegraphics[width=0.16\textwidth]{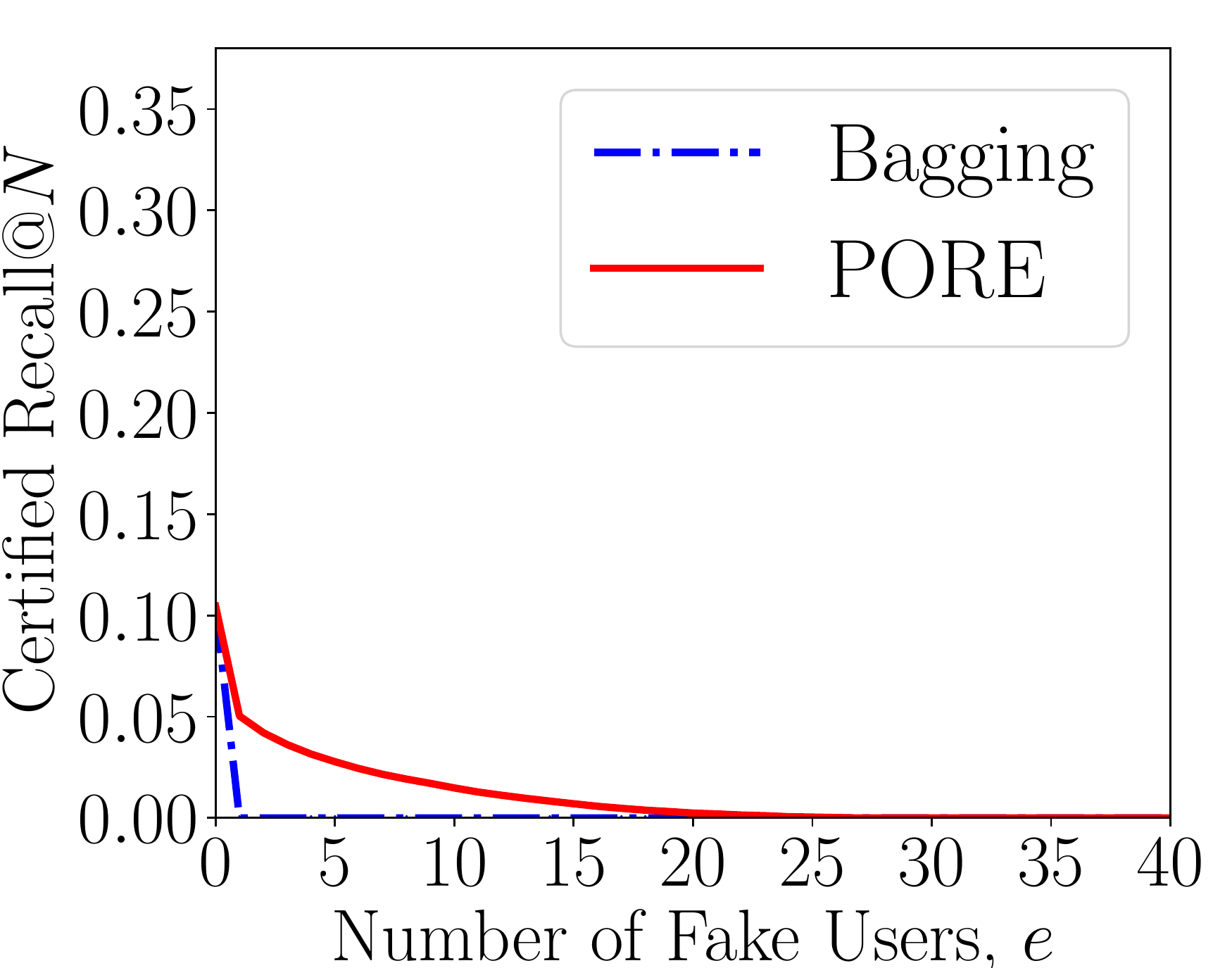}}
{\includegraphics[width=0.16\textwidth]{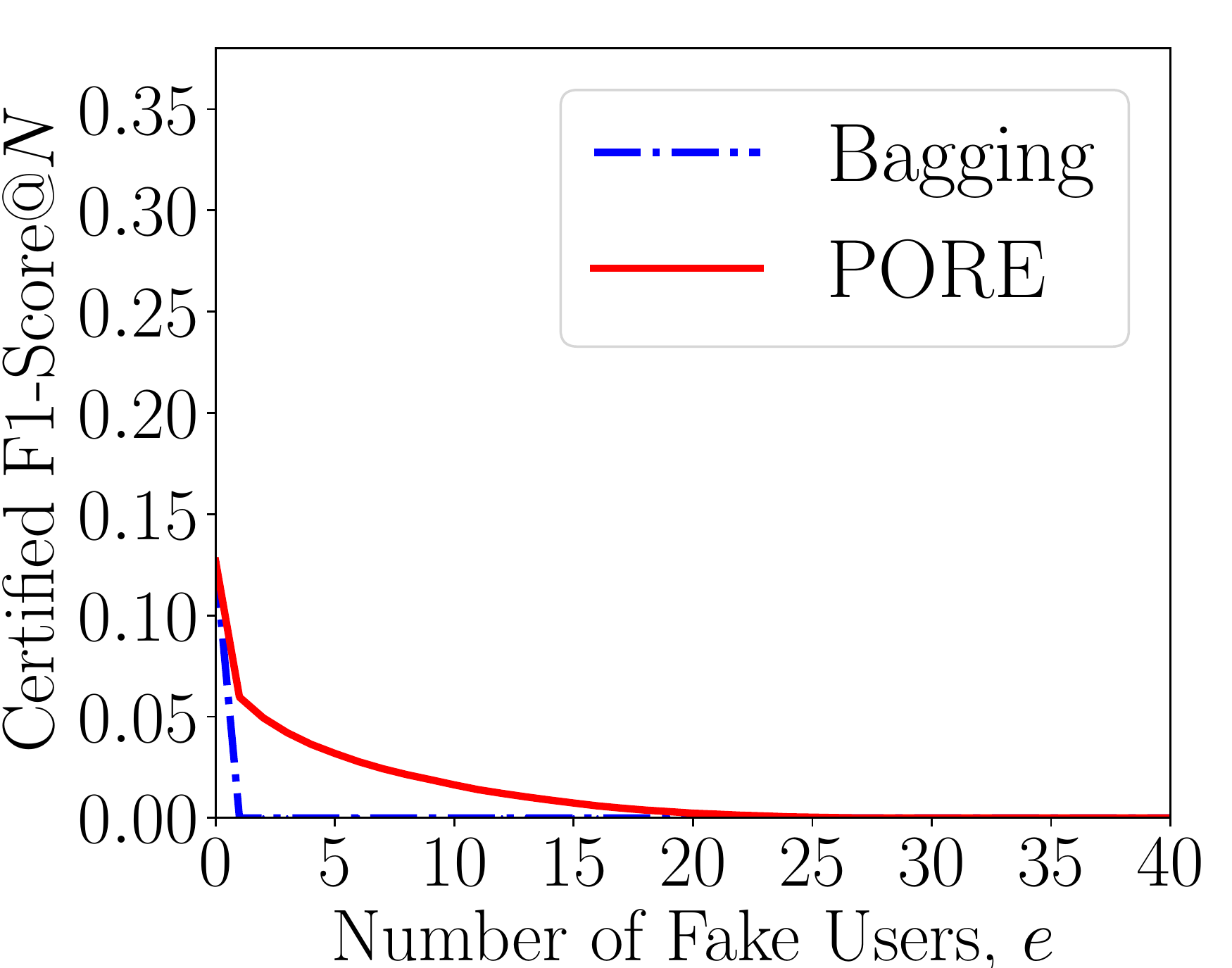}}
\vspace{-3mm}
\caption{Our PORE outperforms bagging when extended from classifiers to recommender systems  on MovieLens-100k (left three) and MovieLens-1M (right three), where $N=10$ and base algorithm is IR.}
\label{compare_with_existing_defense}
\end{figure*}

\begin{figure*}[!t]
	 \centering
	 \vspace{-2mm}
{\includegraphics[width=0.16\textwidth]{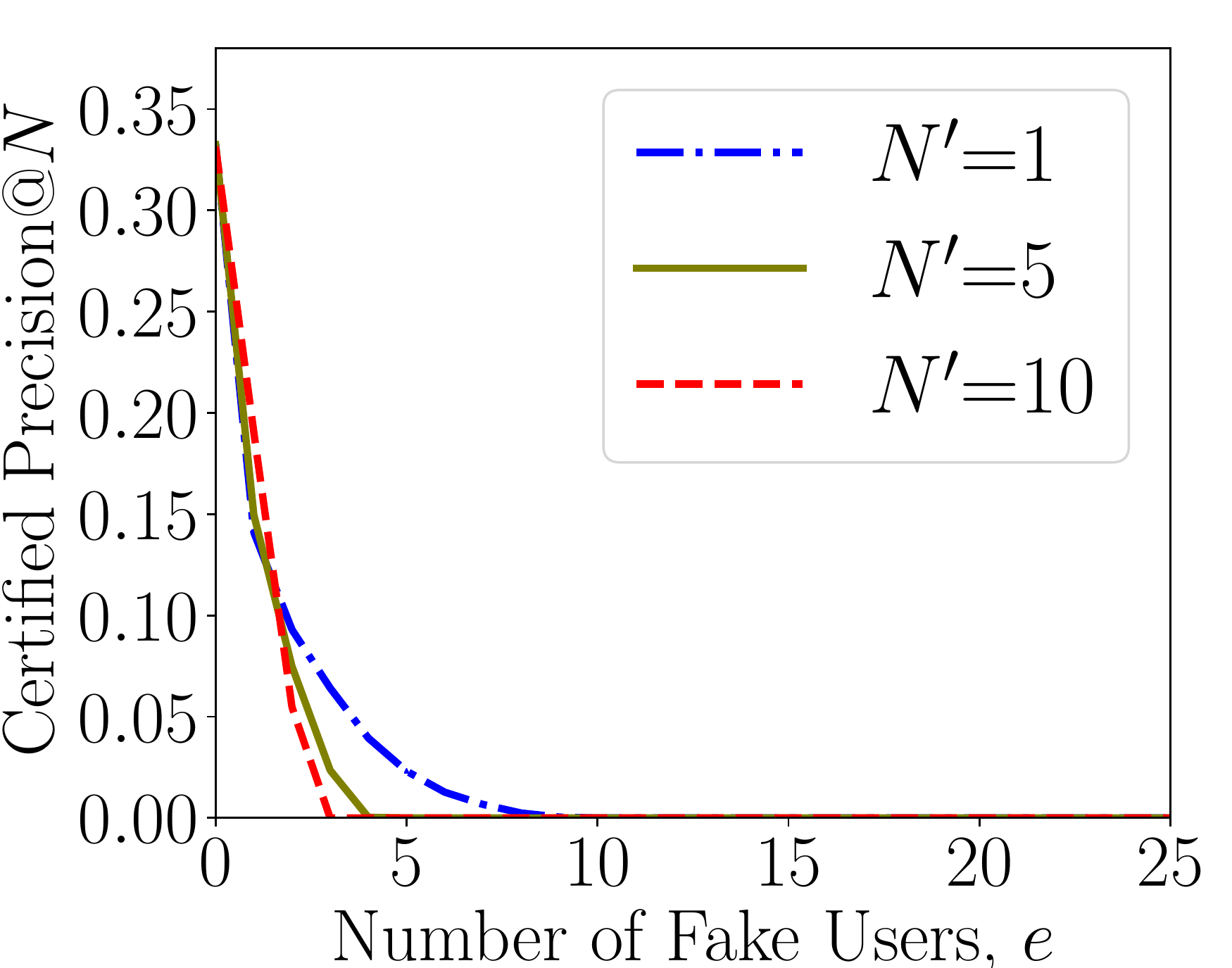}}
{\includegraphics[width=0.16\textwidth]{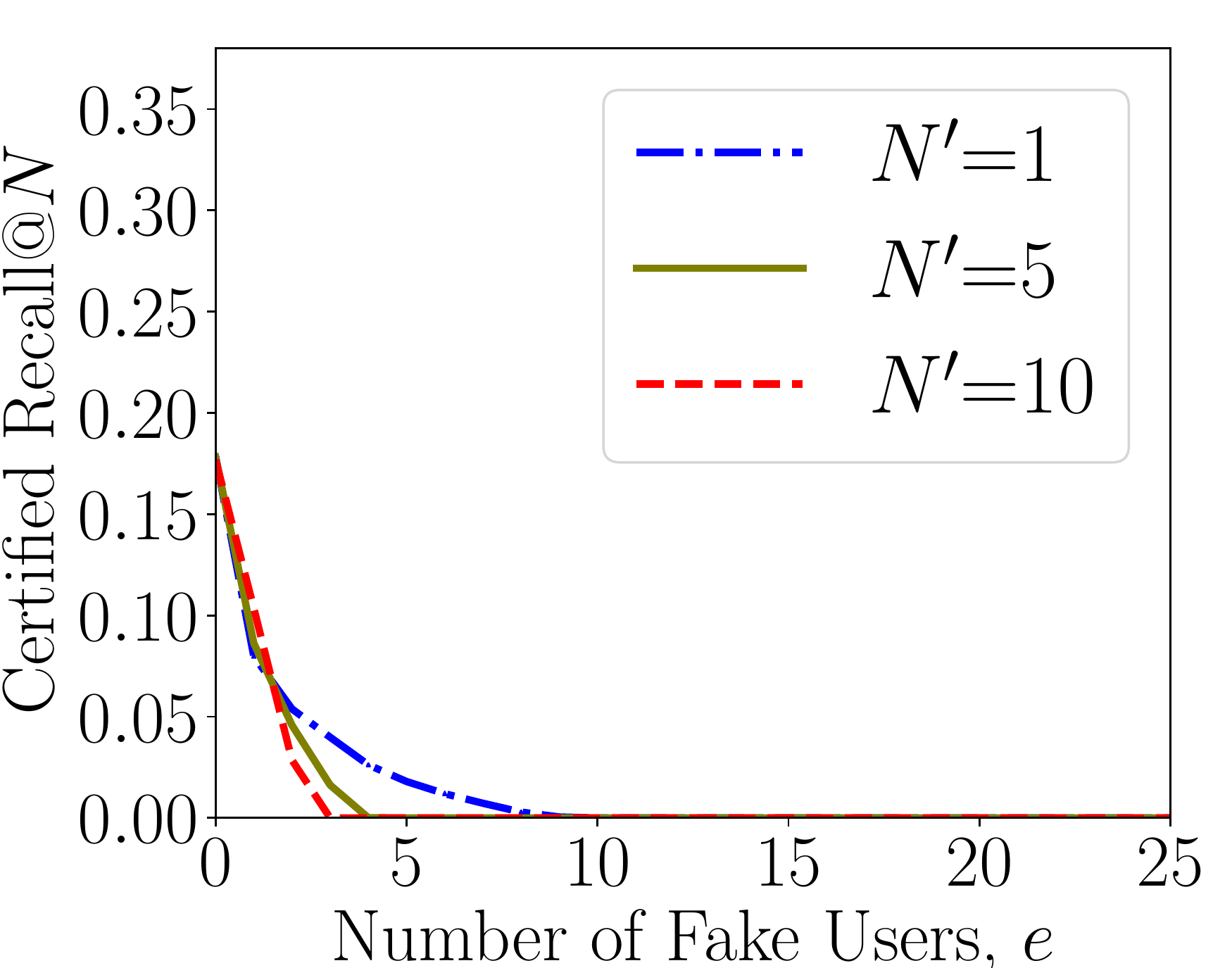}}
{\includegraphics[width=0.16\textwidth]{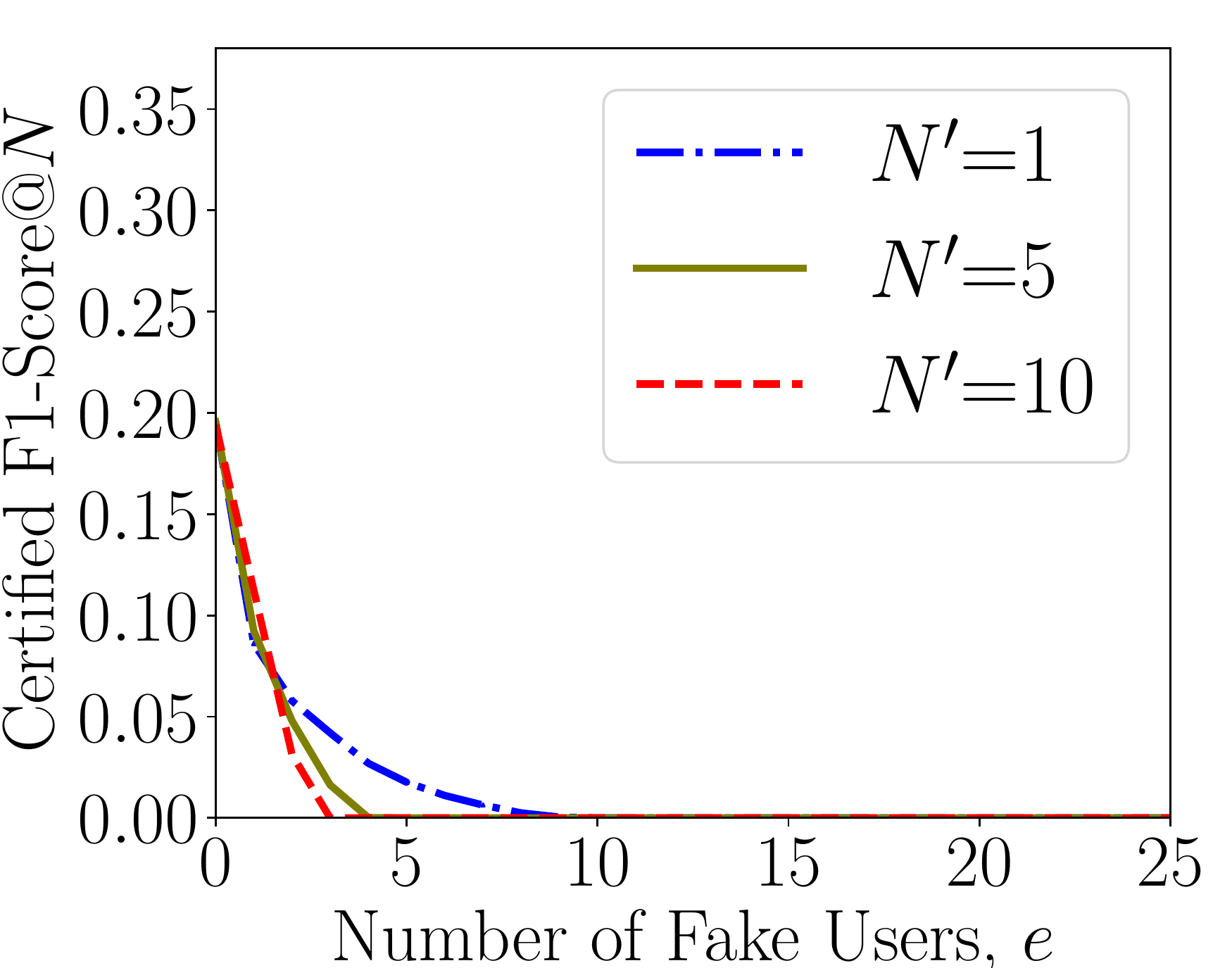}}
{\includegraphics[width=0.16\textwidth]{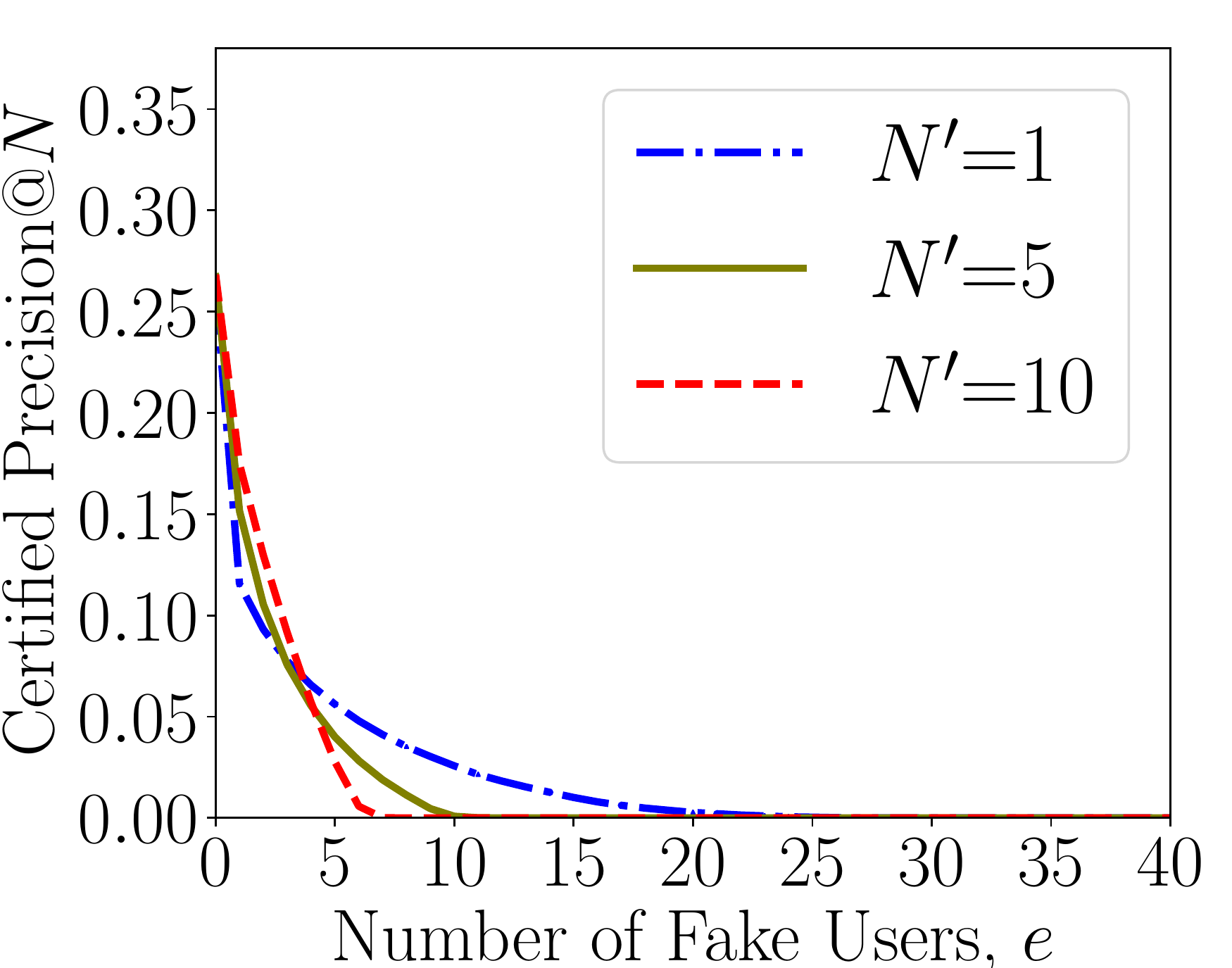}}
{\includegraphics[width=0.16\textwidth]{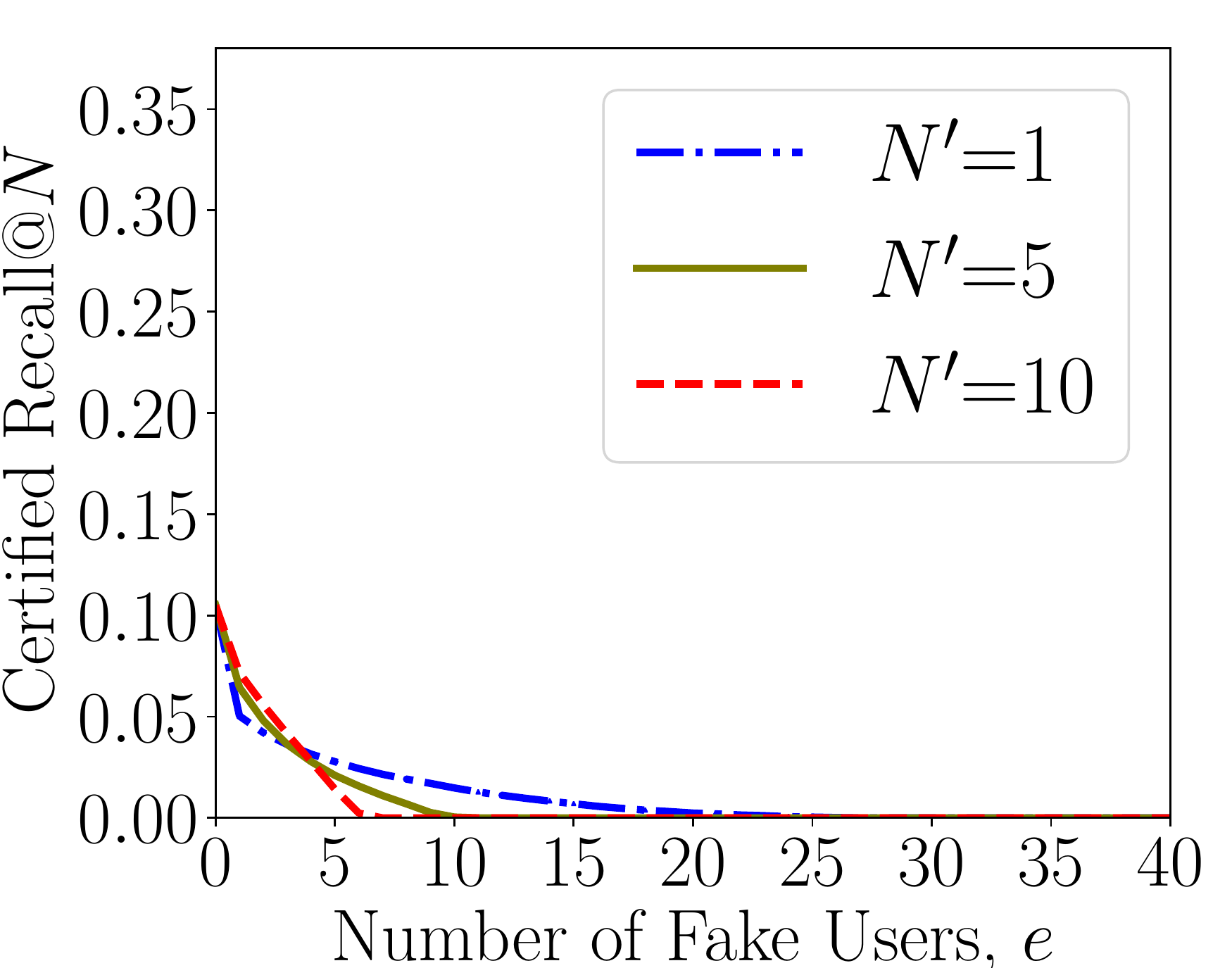}}
{\includegraphics[width=0.16\textwidth]{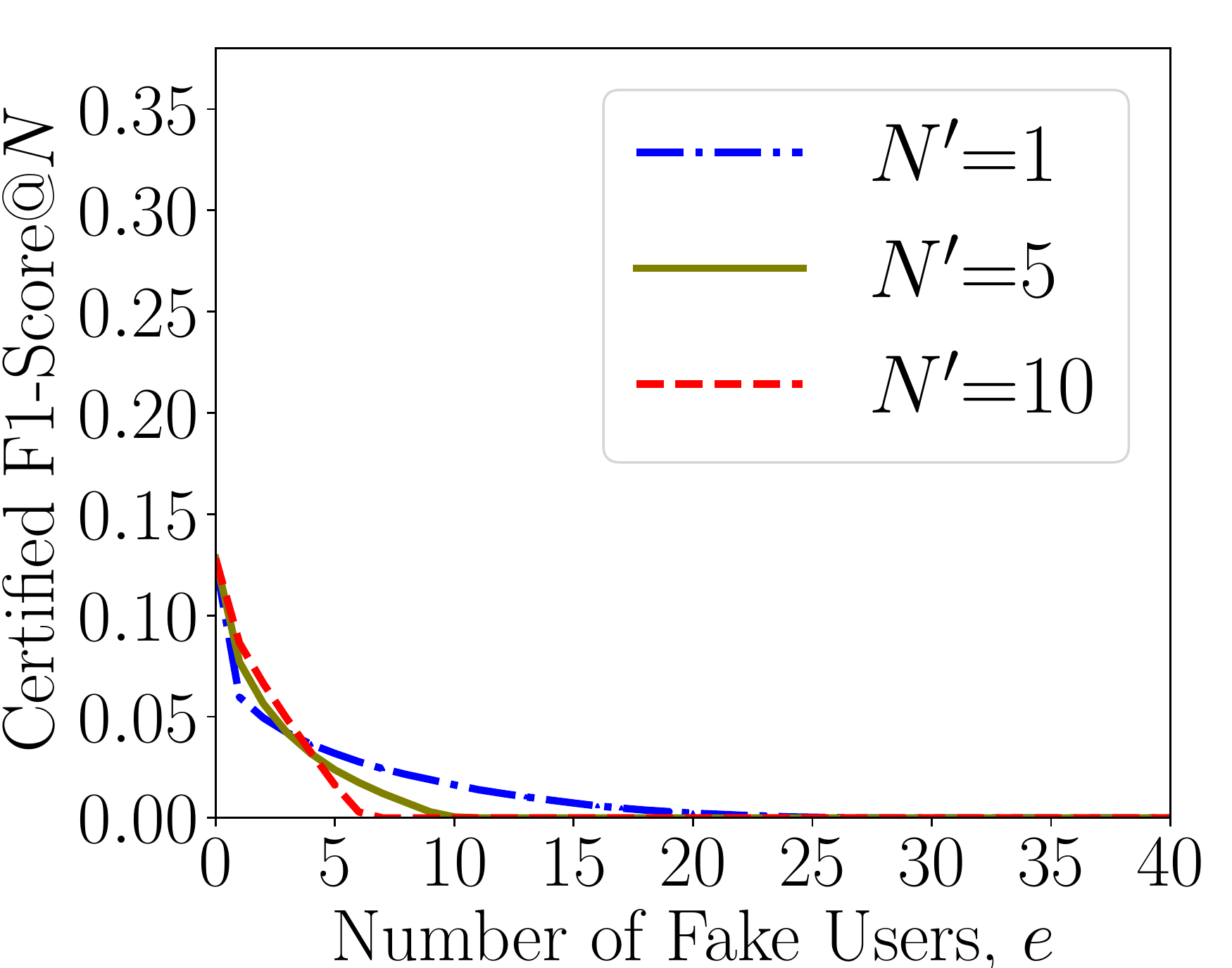}}
\vspace{-3mm}
\caption{Impact of $N'$ on the certified Precision@$N$, certified Recall@$N$, and certified F1-Score@$N$ of ensemble IR for MovieLens-100k (left three) and MovieLens-1M (right three), where $N=10$.}
\label{impact_of_N_prime}
\end{figure*}

\myparatight{Evaluation metrics} When there are no data poisoning attacks,  \emph{Precision@$N$},  \emph{Recall@$N$}, and  \emph{F1-Score@$N$} are standard metrics to evaluate the performance of a recommender system. We denote by $\mathcal{E}_u$ the  set of test items for a user $u$. Precision@$N$ for a user $u$ is the fraction of the top-$N$ items recommended for $u$ that are  in $\mathcal{E}_u$,   Recall@$N$ for $u$ is the fraction of the items in $\mathcal{E}_u$ that are  in the top-$N$  items recommended for $u$, while F1-Score@$N$ for a user $u$ is the harmonic mean of the user's Precision@$N$ and Recall@$N$. The Precision@$N$ (or Recall@$N$ or F1-Score@$N$) of a recommender system algorithm is the users' average  Precision@$N$ (or Recall@$N$ or F1-Score@$N$). 

Under data poisoning attacks,  Precision@$N$, Recall@$N$, and F1-Score@$N$ are insufficient to evaluate a recommender system algorithm. This is because they may be different under different data poisoning attacks, and it is infeasible to enumerate all possible attacks.  To address the challenge, 
we propose to evaluate a recommender system algorithm using \emph{certified Precision@$N$},  \emph{certified Recall@$N$}, and  \emph{certified F1-Score@$N$} under attacks. Like the standard Precision@$N$ (or Recall@$N$ or F1-Score@$N$), calculating our {certified Precision@$N$},  {certified Recall@$N$}, and  {certified F1-Score@$N$} also only requires a clean rating-score matrix and thus does not depend on any specific data poisoning attack.  Certified Precision@$N$ (or certified Recall@$N$ or certified F1-Score@$N$) is a \emph{lower bound} of Precision@$N$ (or Recall@$N$ or F1-Score@$N$) under any data poisoning attacks with at most $e$ fake users. 
For instance, a certified Precision@$N$ of 0.3 means that a recommender system achieves at least Precision@$N$ of 0.3 when the number of fake users is at most $e$, no matter what rating scores they use.  Specifically, our certified Precision@$N$ for a user $u$ is the least fraction of the top-$N$ recommended items  for $u$ that are guaranteed to be in $\mathcal{E}_u$ when there are at most $e$ fake users;  our certified Recall@$N$ for $u$ is the least fraction of $u$'s test items $\mathcal{E}_u$ that are guaranteed to be in the top-$N$ recommended items; while certified F1-Score@$N$ for a user is the harmonic mean of the user's certified Precision@$N$ and certified Recall@$N$. Formally, we have the following for a user $u$: 
\begin{align}
\label{definition_of_cpn}
 &   \text{Certified Precision@}N = \min_{\mathbf{M}' \in \mathcal{L}(\mathbf{M},e)} \frac{|\mathcal{E}_{u} \cap \mathcal{T} (\mathbf{M}', u)|}{N}, \\
 \label{definition_of_crn}
&    \text{Certified Recall@}N = \min_{\mathbf{M}' \in \mathcal{L}(\mathbf{M},e)} \frac{|\mathcal{E}_{u} \cap \mathcal{T} (\mathbf{M}', u)|}{|\mathcal{E}_{u}|}, \\
\label{definition_of_cfn}
&    \text{Certified F1-Score@}N = \min_{\mathbf{M}' \in \mathcal{L}(\mathbf{M},e)} \frac{2 \cdot |\mathcal{E}_{u} \cap \mathcal{T} (\mathbf{M}', u)|}{|\mathcal{E}_{u}|+N}, 
\end{align}
where $|\cdot|$ is the size of a set. We can compute the certified Precision@$N$, certified Recall@$N$, and certified F1-Score@$N$ for each user by Algorithm~\ref{alg:certify}. In particular, we can compute certified intersection size $r_u$ for each user $u$ using Algorithm~\ref{alg:certify} by letting $\mathcal{I}_u = \mathcal{E}_u$. Given $r_u$, the certified Precision@$N$, certified Recall@$N$, and certified F1-Score@$N$ for a user $u$ are at least $\frac{r_u}{N}$, $\frac{r_u}{|\mathcal{E}_{u}|}$, and $\frac{2r_u}{|\mathcal{E}_{u}|+N}$, respectively.
A recommender system algorithm's certified Precision@$N$ (or  Recall@$N$ or  F1-Score@$N$) is the average of the \emph{genuine users'} certified Precision@$N$ (or  Recall@$N$ or  F1-Score@$N$).

\myparatight{Compared methods} We note that PORE is the first provably robust recommender system algorithm against data poisoning attacks. Therefore, there are no prior recommender system algorithms we can compare with in terms of \emph{certified} Precision@$N$, Recall@$N$, and F1-Score@$N$. However, Jia et al.~\cite{jia2020intrinsic} showed that bagging can be used to build provably robust defense against data poisoning attacks for \emph{machine learning classifiers}, which we extend to recommender systems and compare with our PORE.  
Roughly speaking, given a training dataset, bagging trains multiple base classifiers,  each of which is trained on a random subset of training examples in the training dataset. Given a testing input, bagging uses each base classifier to predict its label and takes a majority vote among the predicted labels as the final predicted label for the testing input. Jia et al. showed that bagging can guarantee that the predicted label for an input is provably unaffected by a bounded number of fake training examples injected into the training dataset. 

We generalize their provable guarantee to derive certified intersection size for each genuine user in recommender systems, which can be then used to compute certified Precision@$N$, Recall@$N$, and F1-Score@$N$ for bagging. Specifically, we treat a user as a testing input, an item as a label, and a base recommender system as a base classifier in the terminology of bagging. Like our PORE, the generalized bagging builds $T$ base recommender systems and takes majority vote among them to recommend top-$N$ items to each user.  Note that since a base classifier predicts one label for a testing input, we set $N'=1$, i.e., a base recommender system (i.e., a base classifier) recommends top-1 item (i.e., predicts one label) for a user (i.e., a testing input). Finally, for each user, bagging recommends him/her the $N$ items with the largest (poisoned) item probabilities. Given a set of items $\mathcal{I}_u$ for a user $u$, we denote by $\underline{p_i}$  a lower bound of item probability $p_i$ of item $i \in \mathcal{I}_u$. We use $\overline{p}_l$ to denote the largest upper bound of item probabilities of items in $\mathcal{I} \setminus \mathcal{I}_u$, i.e., $\overline{p}_l = \max_{j \in \mathcal{I} \setminus \mathcal{I}_u } \overline{p}_j$. We can estimate these item-probability bounds using our method in Equation~\ref{cp_lower_bound} and~\ref{cp_upper_bound}.  Given  $\underline{p_i}$ and $\overline{p}_l$, we can compute an integer $Z_i$ based on Theorem 1 in bagging~\cite{jia2020intrinsic}. Roughly speaking, bagging can guarantee the poisoned item probability $p_i'$ is larger than $p_l'$ under any data poisoning attacks with at most $Z_i$ fake users. Therefore, given at most  $e$ fake users, the certified intersection size of bagging for a user $u$ can be computed as $\min\{\sum_{i \in \mathcal{I}_u }\mathbb{I}(Z_i \geq e), N\}$, where $\mathbb{I}$ is an indicator function. 

\neil{We note that when bagging is extended to recommender systems, both $N'$ and $N$ can only be 1 when using the techniques in~\cite{jia2020intrinsic} to derive its provable robustness guarantees.  PORE can be viewed as an extension of bagging to recommender systems, but $N'$ and $N$ can be arbitrary positive integers. Due to such differences, we propose new techniques to derive the robustness guarantee of PORE. Our major technical contribution is to derive a better guarantee for bagging applied to recommender systems.}

\begin{figure*}[!t]
	 \centering
{\includegraphics[width=0.16\textwidth]{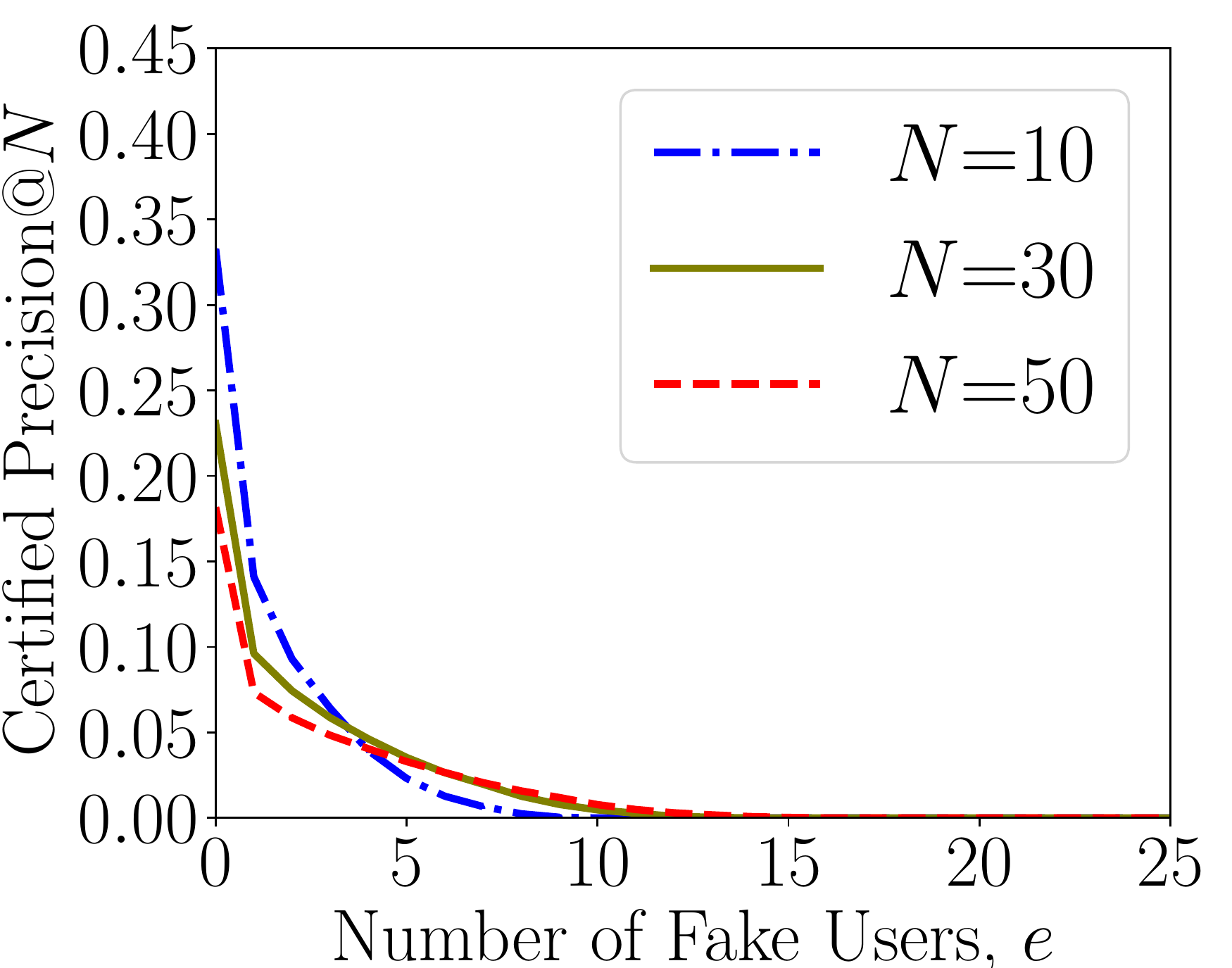}} 
{\includegraphics[width=0.16\textwidth]{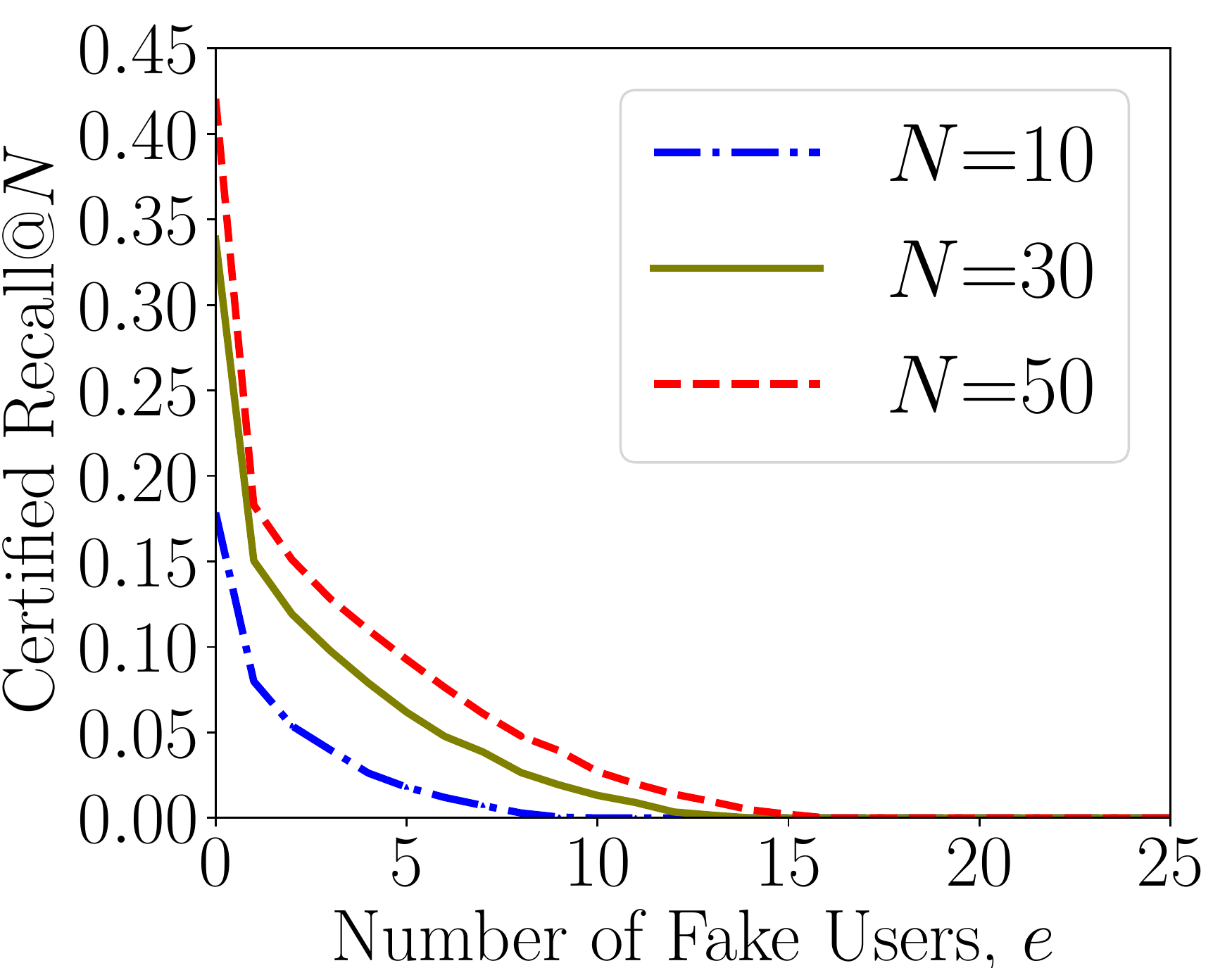}}
{\includegraphics[width=0.16\textwidth]{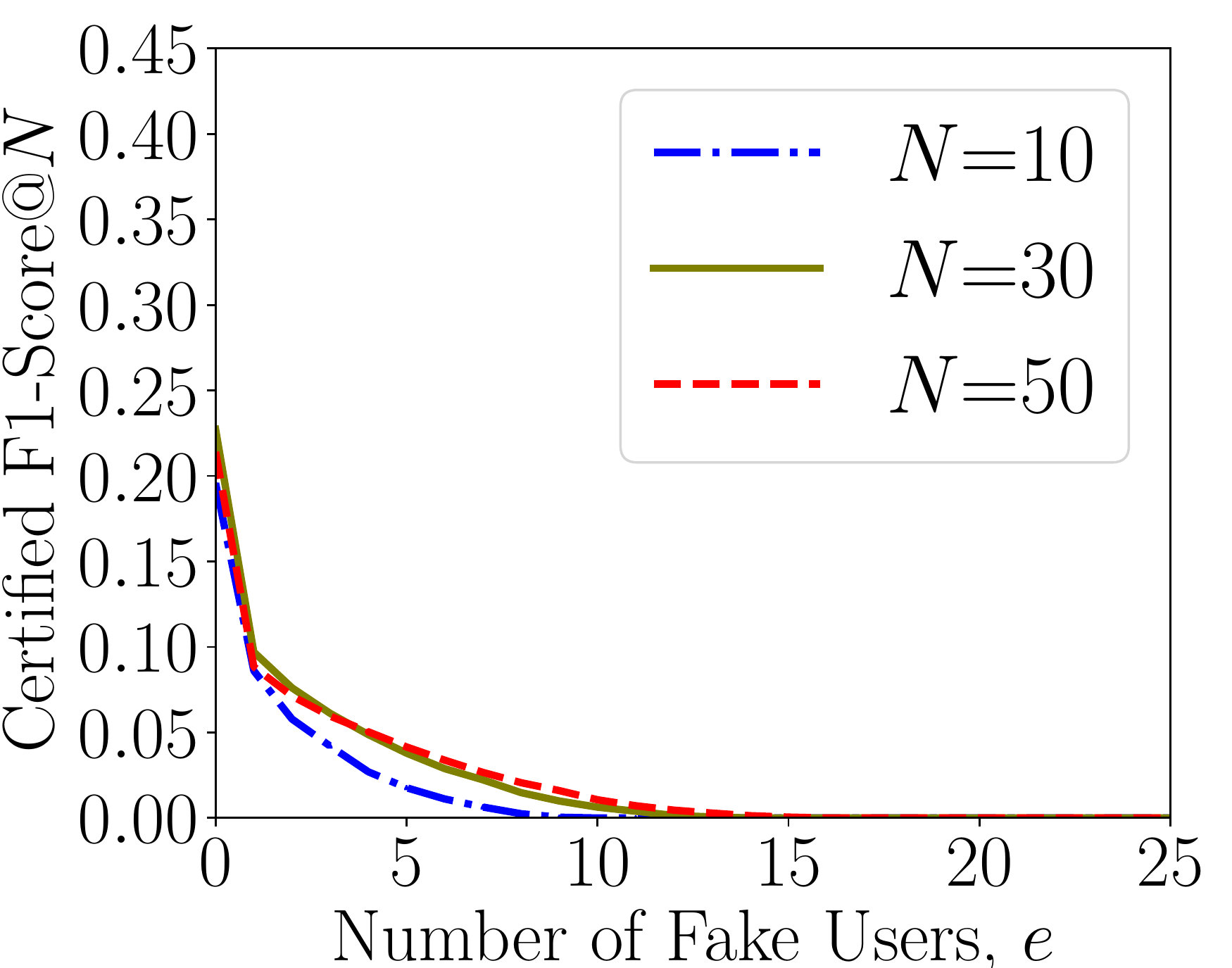}} 
{\includegraphics[width=0.16\textwidth]{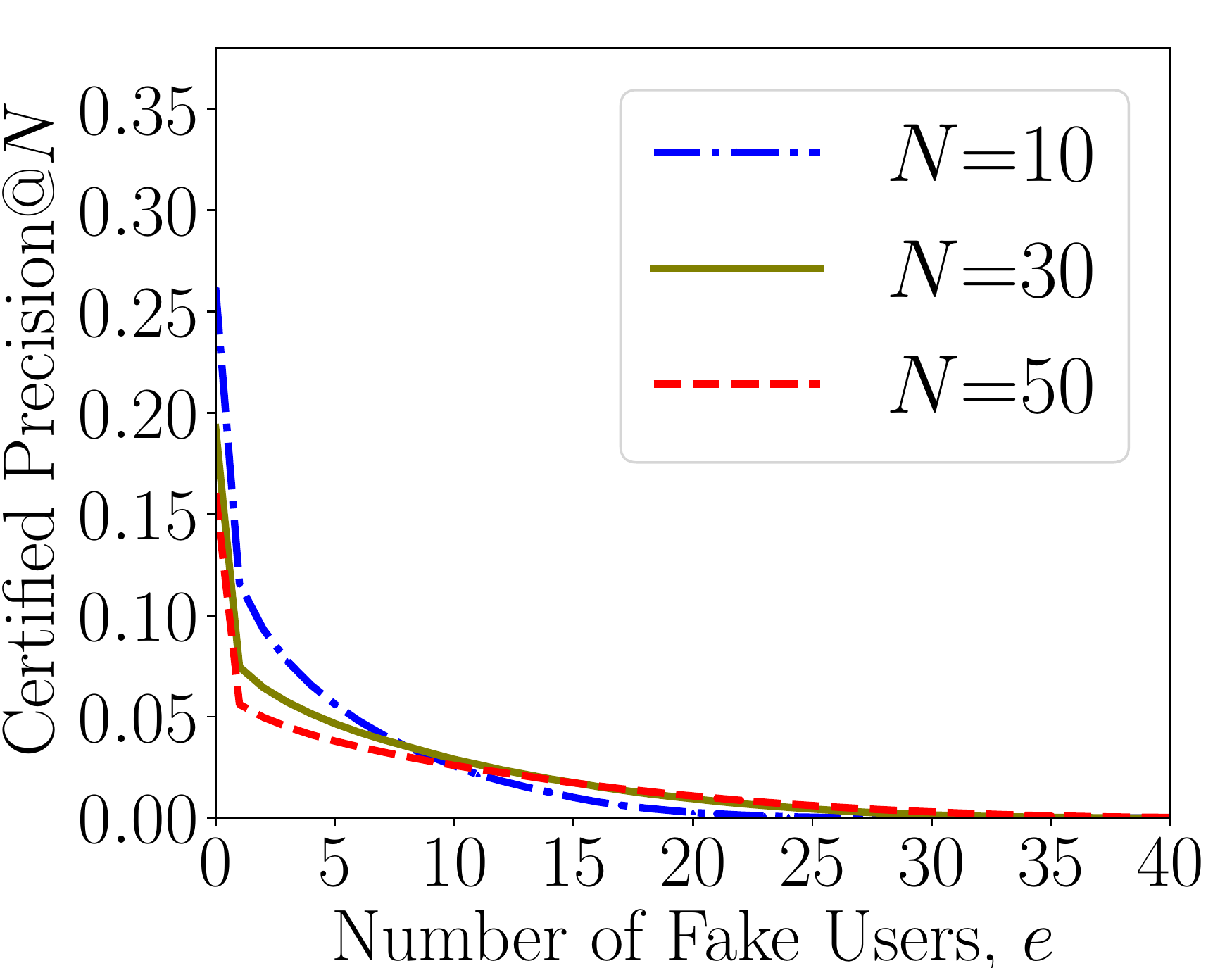}}
{\includegraphics[width=0.16\textwidth]{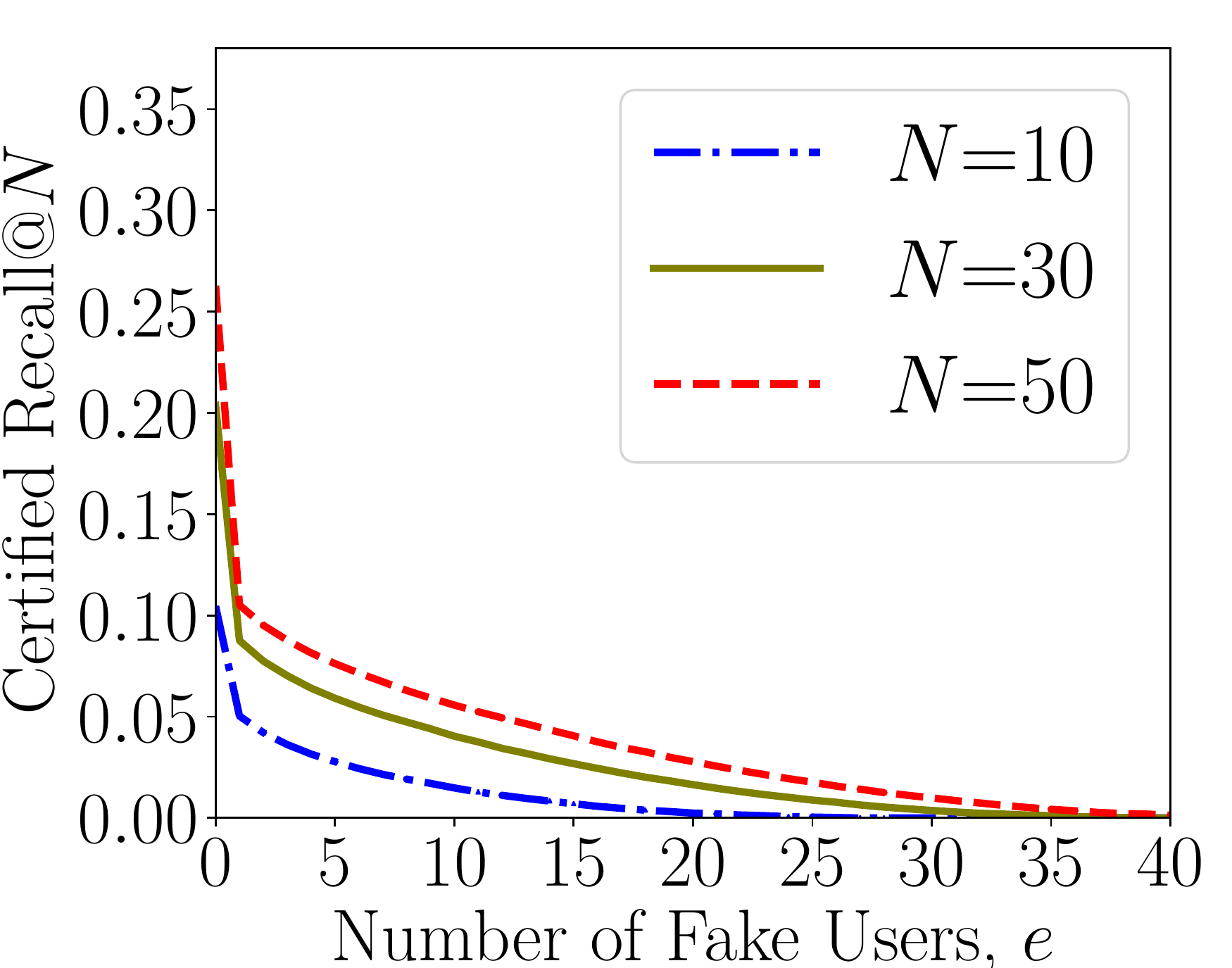}}
{\includegraphics[width=0.16\textwidth]{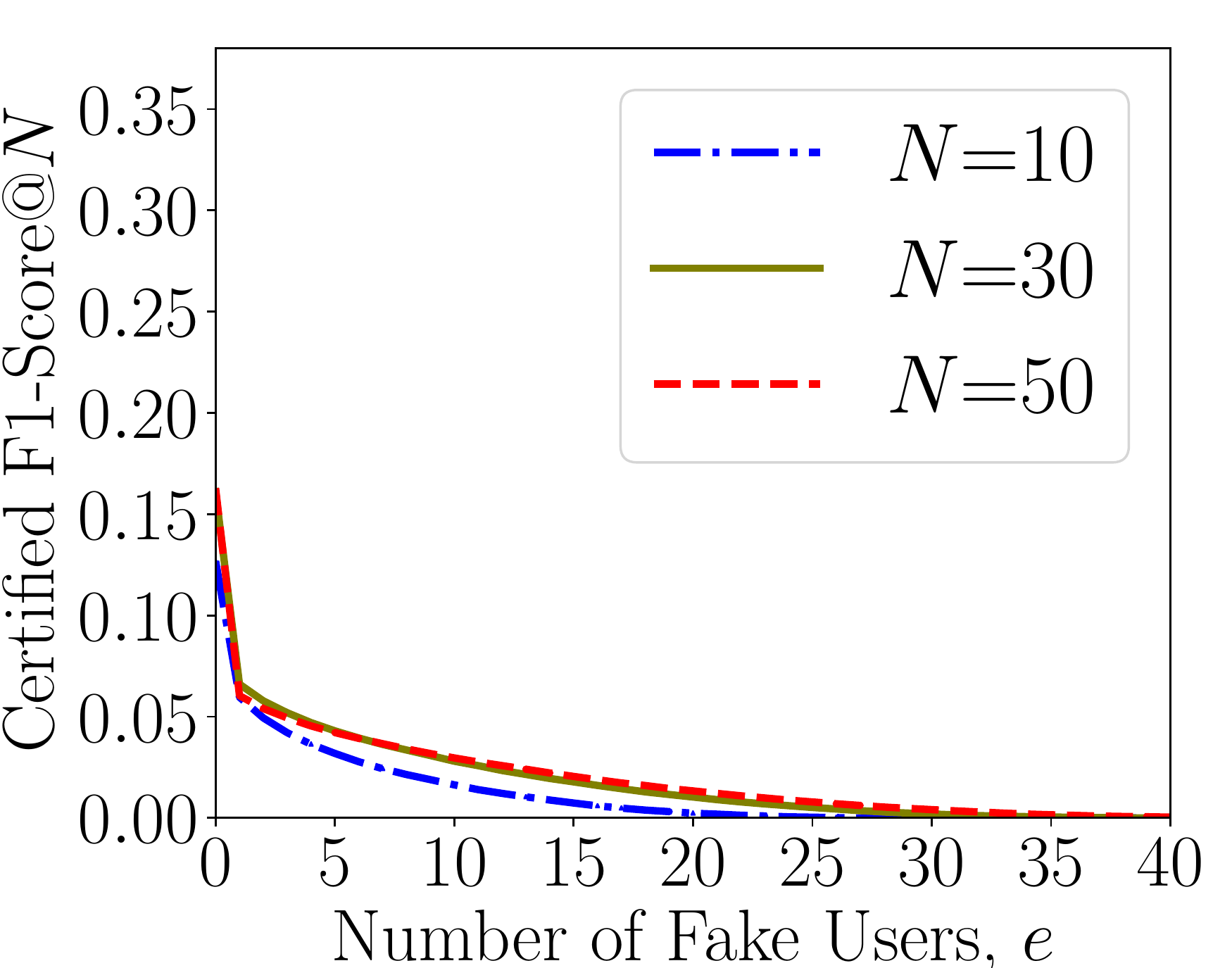}}
\vspace{-3mm}
\caption{Impact of $N$ on the certified Precision@$N$, certified Recall@$N$, and certified F1-Score@$N$ of ensemble IR for MovieLens-100k (left three) and MovieLens-1M (right three). }
\label{impact_of_N}
\end{figure*}

\begin{figure*}[!t]
	 \centering
{\includegraphics[width=0.16\textwidth]{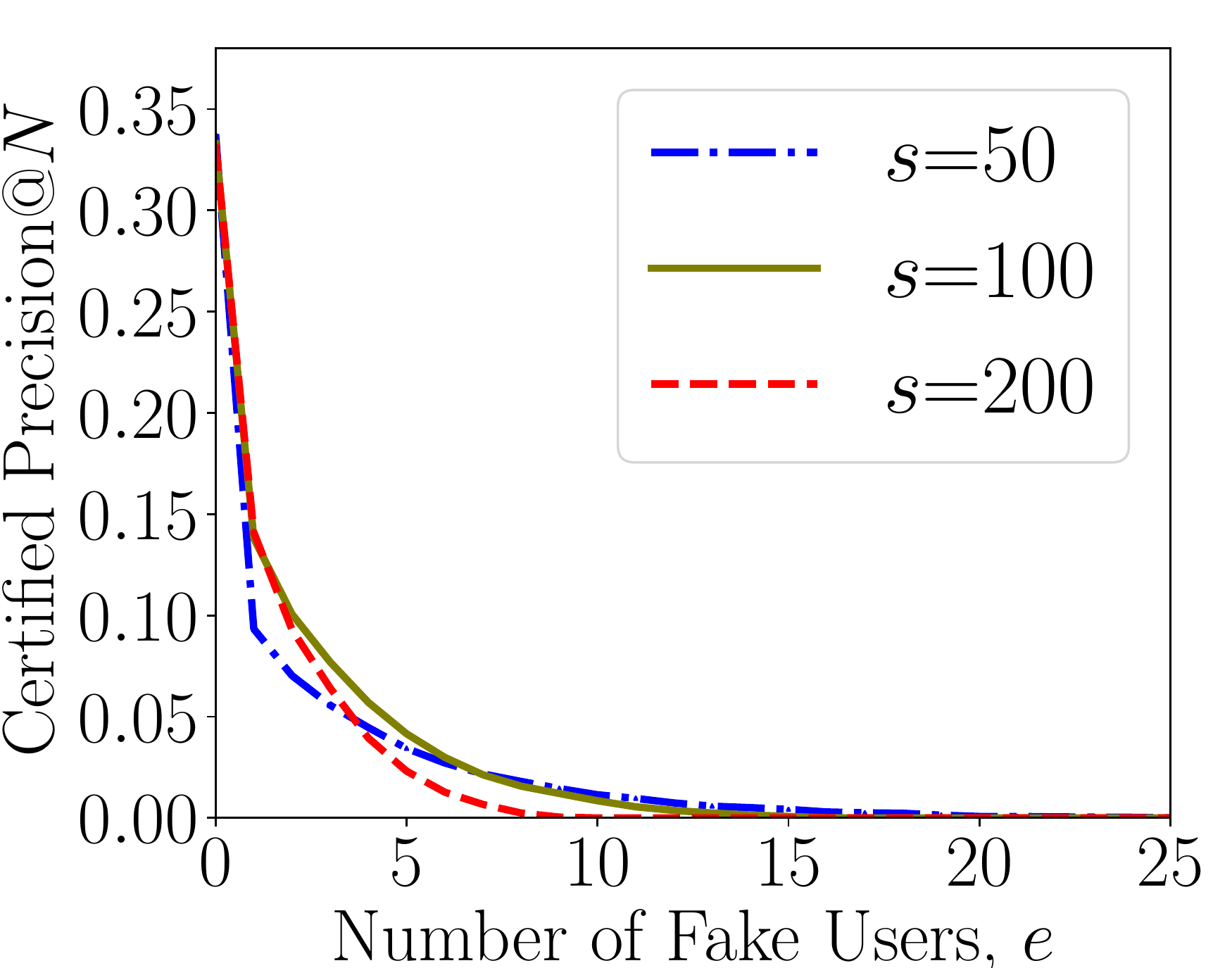}}
{\includegraphics[width=0.16\textwidth]{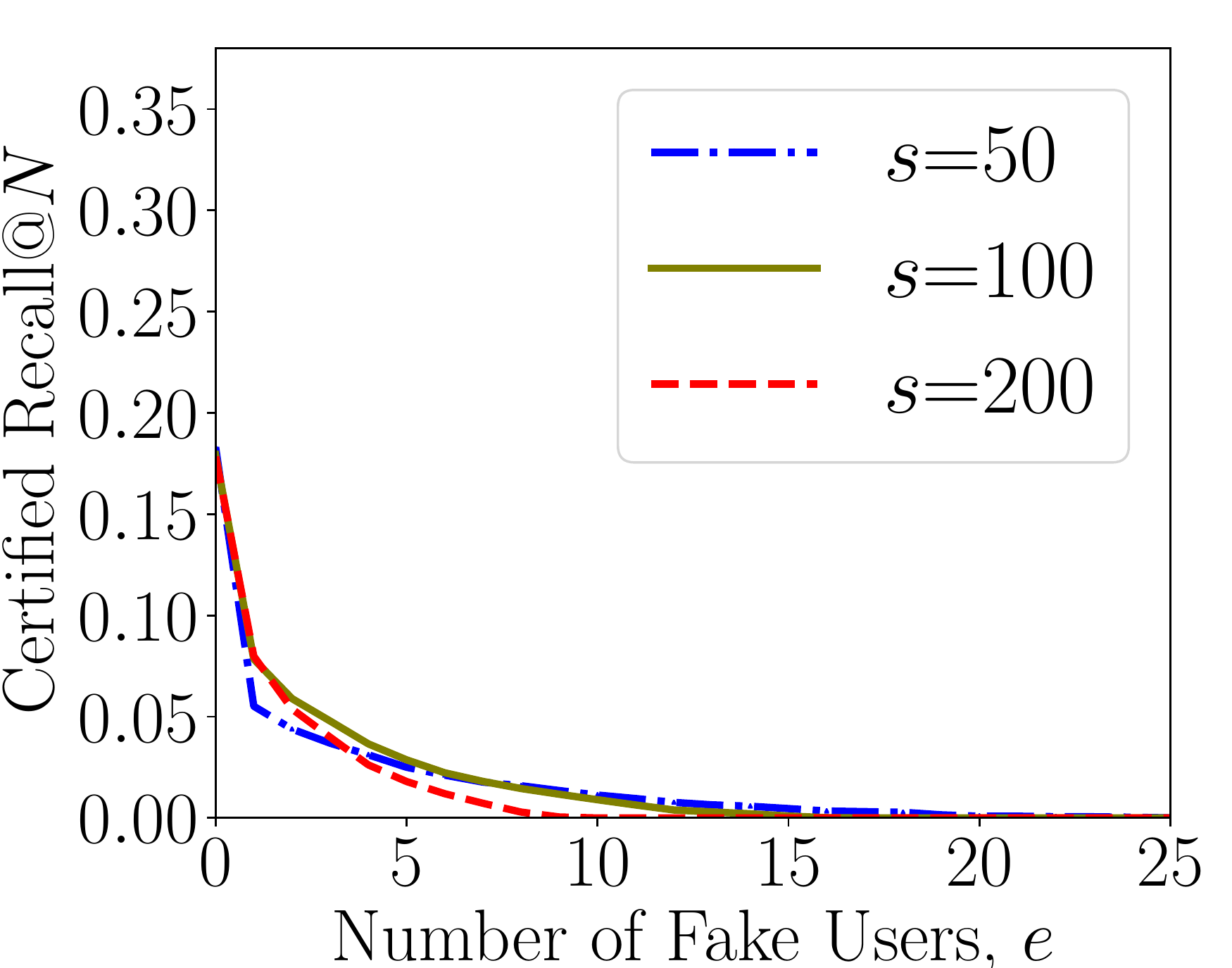}}
{\includegraphics[width=0.16\textwidth]{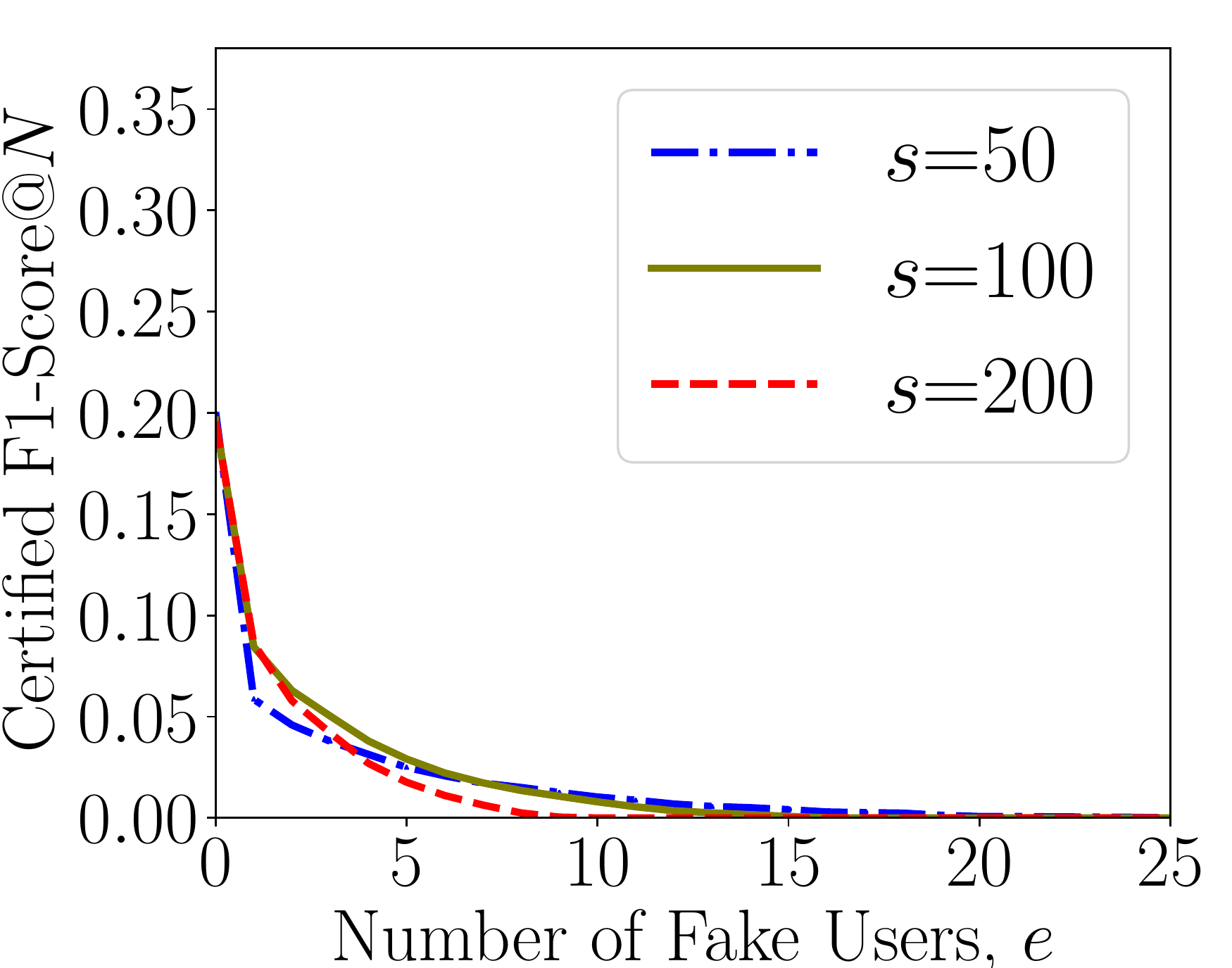} } 
{\includegraphics[width=0.16\textwidth]{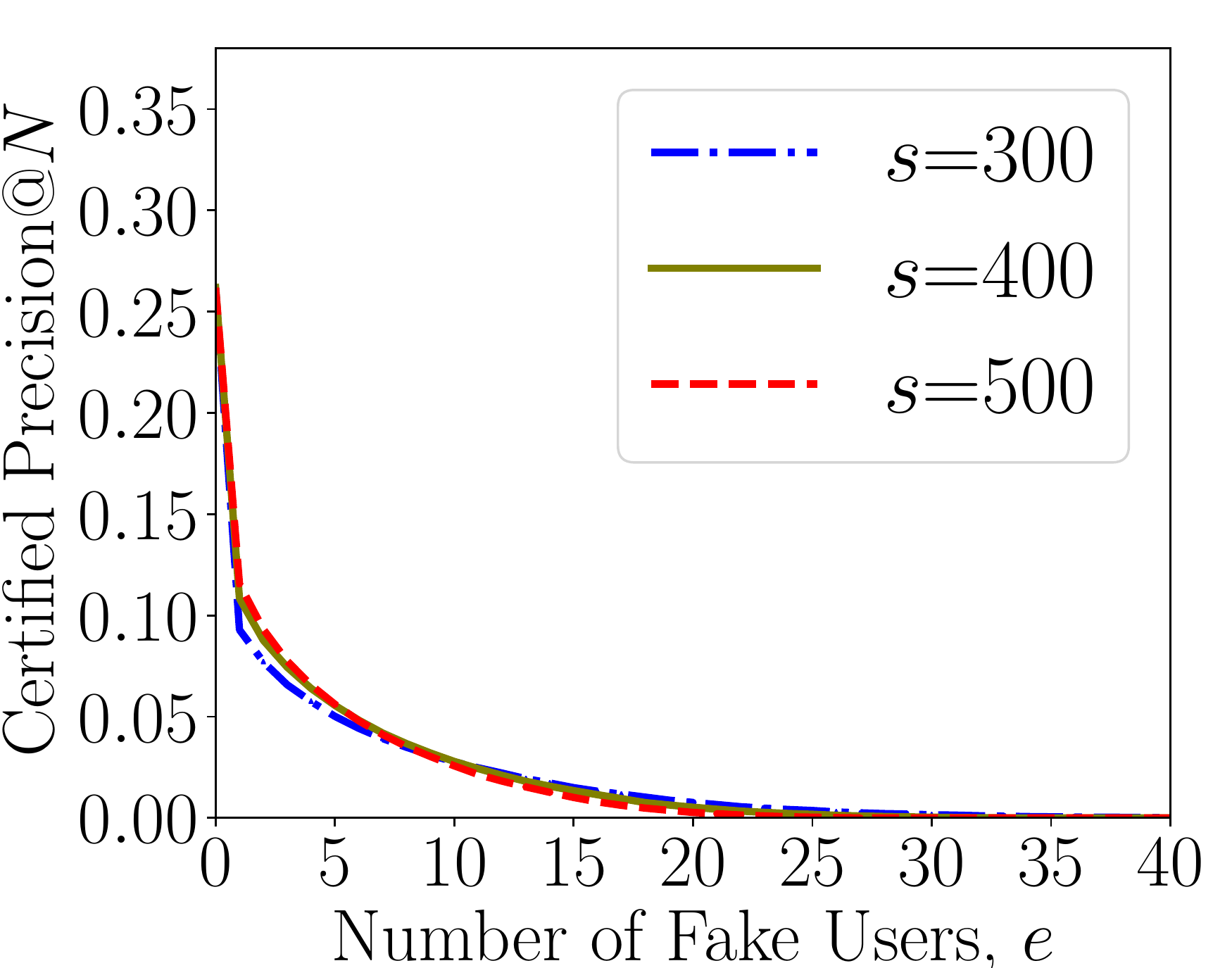}}
{\includegraphics[width=0.16\textwidth]{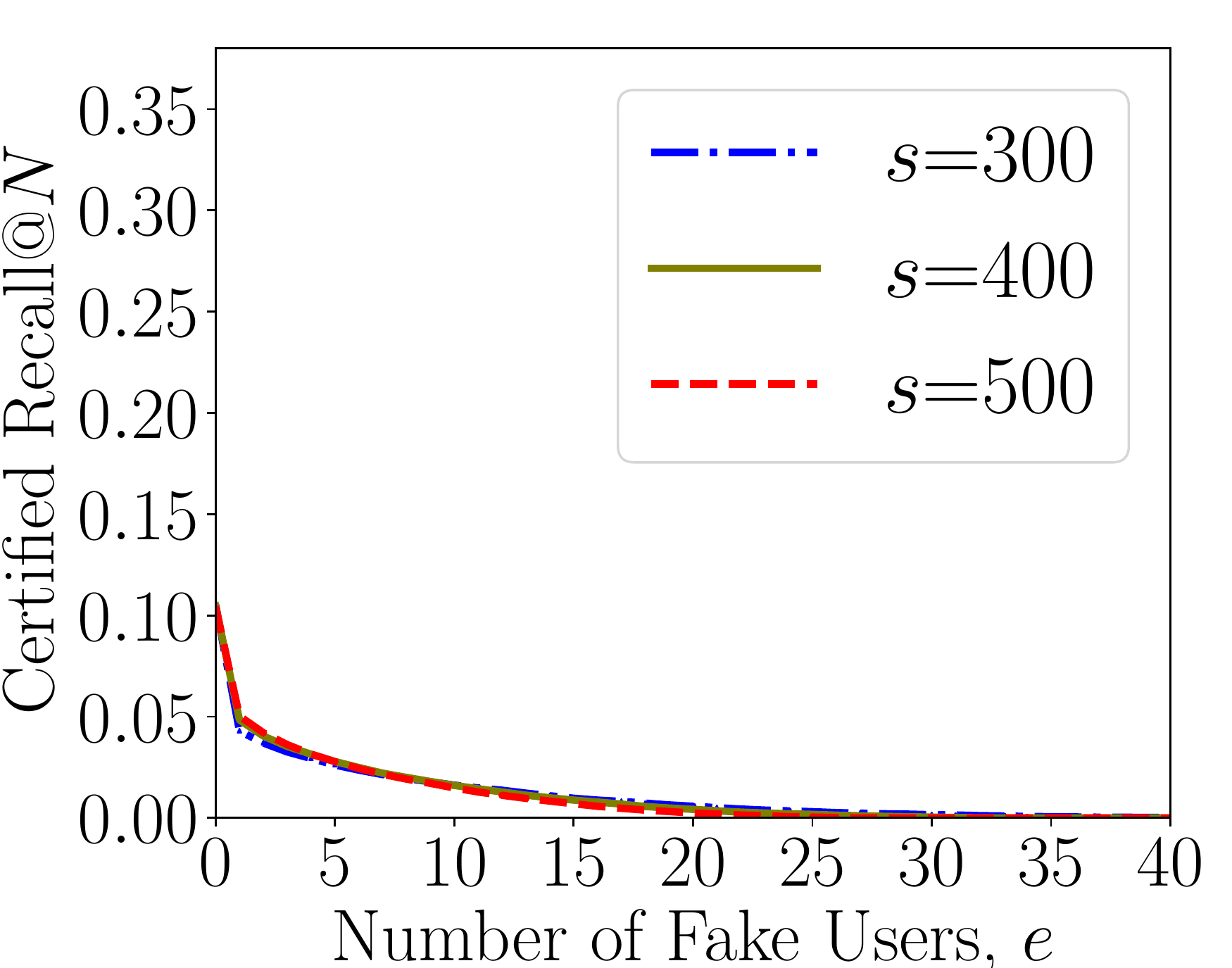}}
{\includegraphics[width=0.16\textwidth]{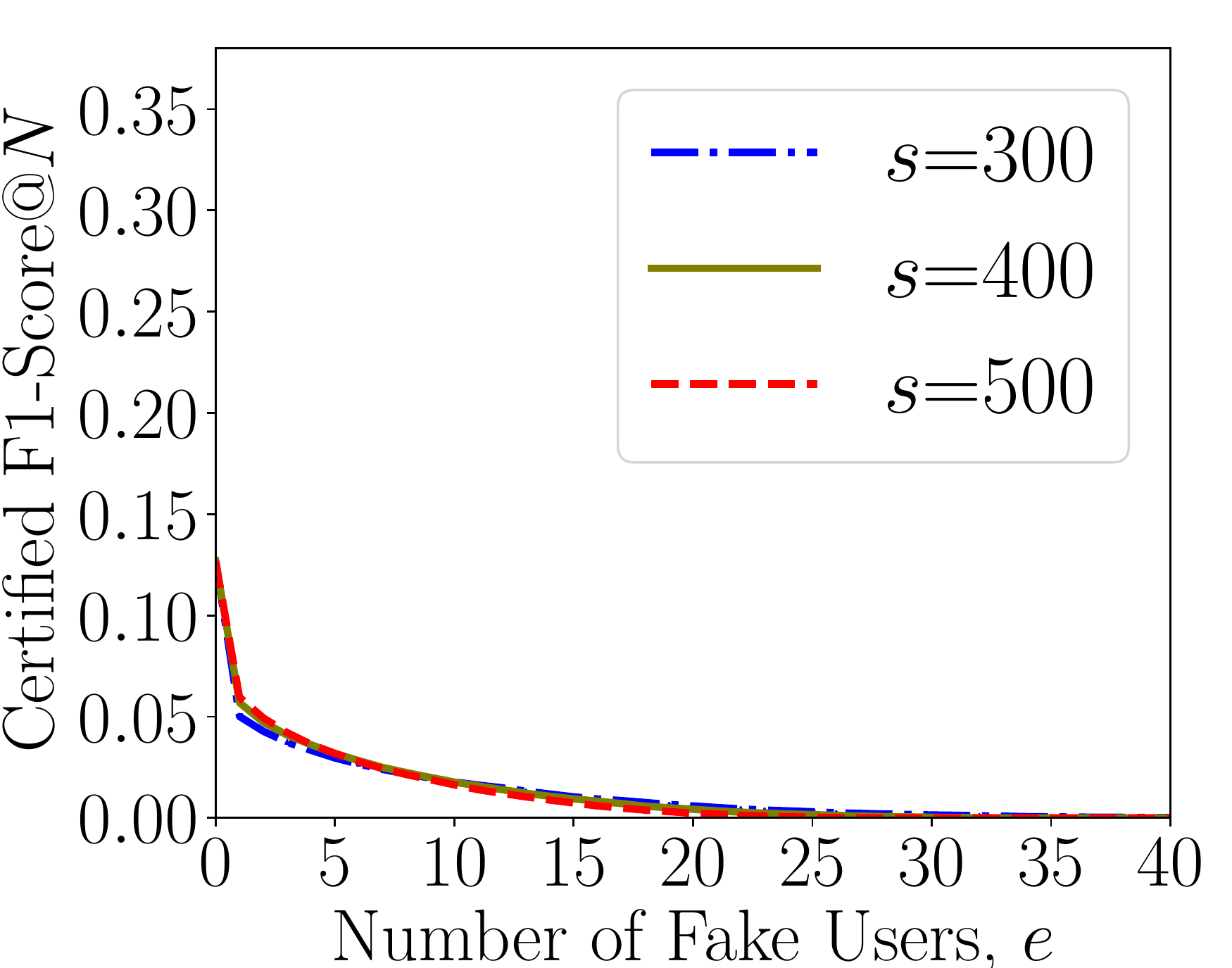}}
\vspace{-3mm}
\caption{Impact of $s$ on the certified Precision@$N$, certified Recall@$N$, and certified F1-Score@$N$ of ensemble IR for MovieLens-100k (left three) and MovieLens-1M (right three), where $N=10$.}
\label{impact_of_s}
\end{figure*}

\myparatight{Parameter setting} {PORE has the following parameters: $N'$ is the number of items recommended by a base recommender system for a user, $N$ is the number of items recommended by our ensemble recommender system for a user, $T$ is the number of base recommender systems, $1-\alpha$ is the confidence score, and $s$ is the number of rows sampled from the  rating-score matrix in each submatrix}. Unless otherwise mentioned, we adopt the following default parameter settings: $N'=1$, $N=10$, $T=100,000$, $\alpha =0.001$, $s=200$ for MovieLens-100k, and  $s=500$ for MovieLens-1M.  We will study the impact of each parameter on the certified Precision@$N$, certified Recall@$N$, and certified F1-Score@$N$ of PORE while fixing the remaining parameters to their default settings.  
We call our ensemble recommender system \emph{ensemble IR} (or \emph{ensemble BPR}) when the base  algorithm is IR (or BPR). By default, we use IR as the base  algorithm because of its scalability.

\subsection{Experimental Results} 

 We report Precision@$N$/Recall@$N$/F1-Score@$N$ under no attacks (i.e., $e=0$), while we report certified Precision@$N$ /Recall@$N$/F1-Score@$N$ under attacks (i.e., $e\geq 1$).

\myparatight{Our PORE outperforms bagging~\cite{jia2020intrinsic}} \CR{Figure~\ref{compare_with_existing_defense} compares our PORE with bagging on the two datasets. We find that our PORE substantially outperforms bagging when extended from classifiers to recommender systems. The reason is that PORE  jointly considers multiple items when deriving the certified intersection size. In contrast, bagging can only consider each item independently when extended to recommender systems, and thus achieve a suboptimal certified intersection size.} 

\myparatight{Impact of $N'$} Figure~\ref{impact_of_N_prime} shows the impact of $N'$. We have two observations. First, our method has similar Precision@$N$/Recall@$N$/F1-Score@$N$ for different $N'$ when there are no attacks (i.e., $e=0$).  Second, a smaller $N'$ achieves a lower certified Precision@$N$, certified Recall@$N$, or certified F1-Score@$N$ when $e$ is small (e.g., $e=1$), but the curve has a longer tail. In other words, a smaller $N'$ is more robust against data poisoning attacks as the number of fake users $e$ increases. The reason is that an attack has a smaller manipulation space when $N'$ is smaller. This observation is also consistent with our theoretical result in Equation~(\ref{optimization_problem_theorem_1}). Specifically, given the same item-probability lower/upper bounds, a smaller $N'$ may lead to a larger certified intersection size. 
Therefore, we set $N'$ to be $1$ by default in our experiments. 

\begin{table}[tp]\renewcommand{\arraystretch}{1.3}

  \centering
  \fontsize{7.5}{8}\selectfont
  \caption{Precision@$10$, Recall@$10$, and F1-Score@$10$ of IR, Ensemble IR, BPR, and Ensemble BPR under no attacks. } 
  \vspace{-3mm}
 \subfloat[MovieLens-100k]{ \begin{tabular}{|c|c|c|c|}
    \hline
    {Algorithm}   & Precision@10 & Recall@10 & F1-Score@10 \\
    \hline
    IR & 0.330753 & 0.176385 &0.193783  \\
    \hline
    Ensemble IR & 0.332556 & 0.178293 & 0.195624  \\
    \hline
    BPR & 0.349841 & 0.181807 & 0.199426  \\
    \hline
    Ensemble BPR & 0.352280 & 0.173296 & 0.193362 \\
    \hline
  \end{tabular}}
  
   \subfloat[MovieLens-1M]{ \begin{tabular}{|c|c|c|c|}
    \hline
    {Algorithm}    & Precision@10 & Recall@10 & F1-Score@10 \\
    \hline
    IR & 0.270116 & 0.104350 & 0.127704 \\
    \hline
    Ensemble IR & 0.262616 & 0.103638 & 0.126191 \\
    \hline
    BPR & 0.324449 & 0.118385 & 0.144765 \\
    \hline
    Ensemble BPR & 0.362945 & 0.119441 & 0.151509 \\
    \hline
  \end{tabular}}

  \label{asr-train-test}
  \end{table}

\begin{figure*}[!t]
	 \centering
	 \vspace{-2mm}
{\includegraphics[width=0.16\textwidth]{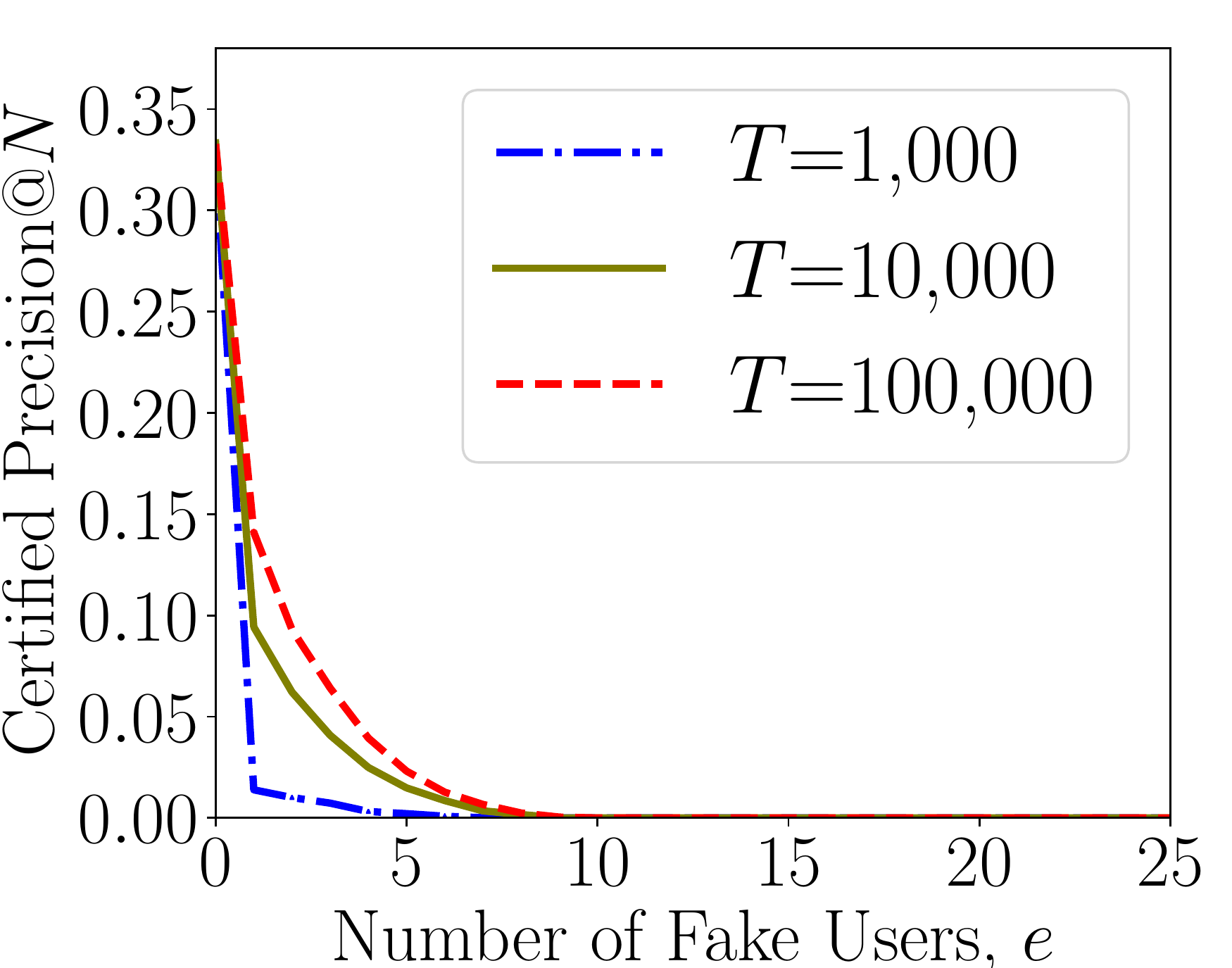}}
{\includegraphics[width=0.16\textwidth]{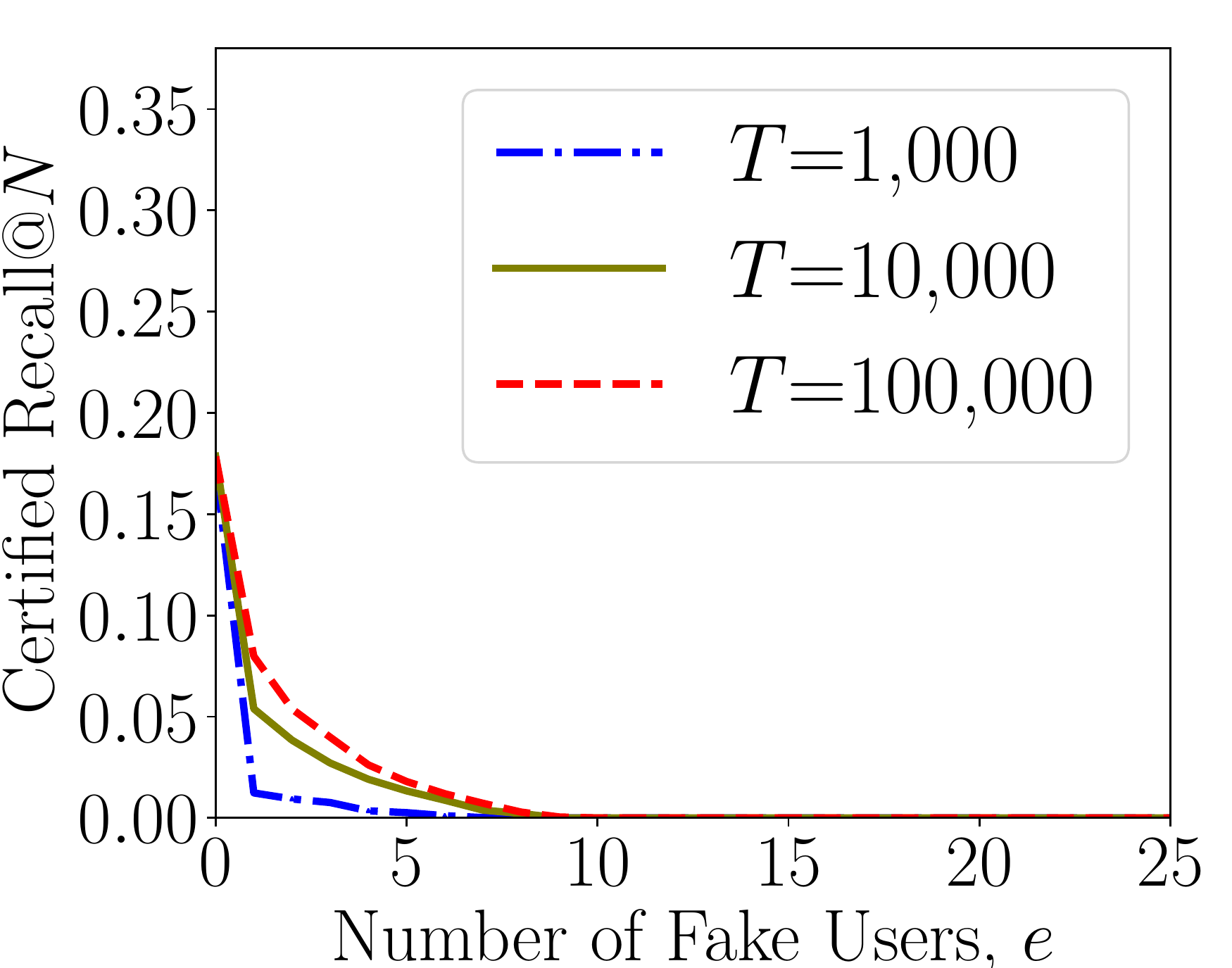}}
{\includegraphics[width=0.16\textwidth]{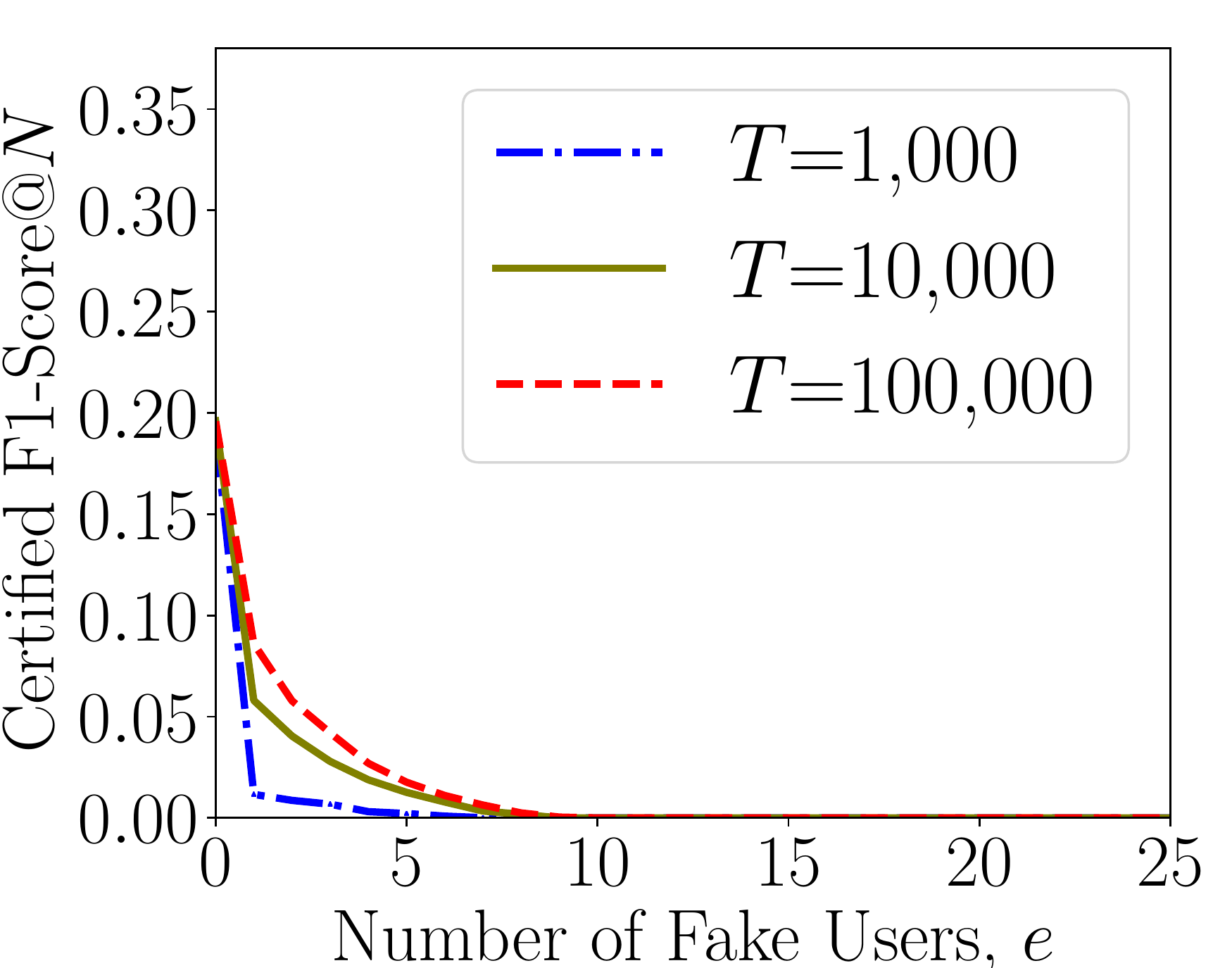}}
{\includegraphics[width=0.16\textwidth]{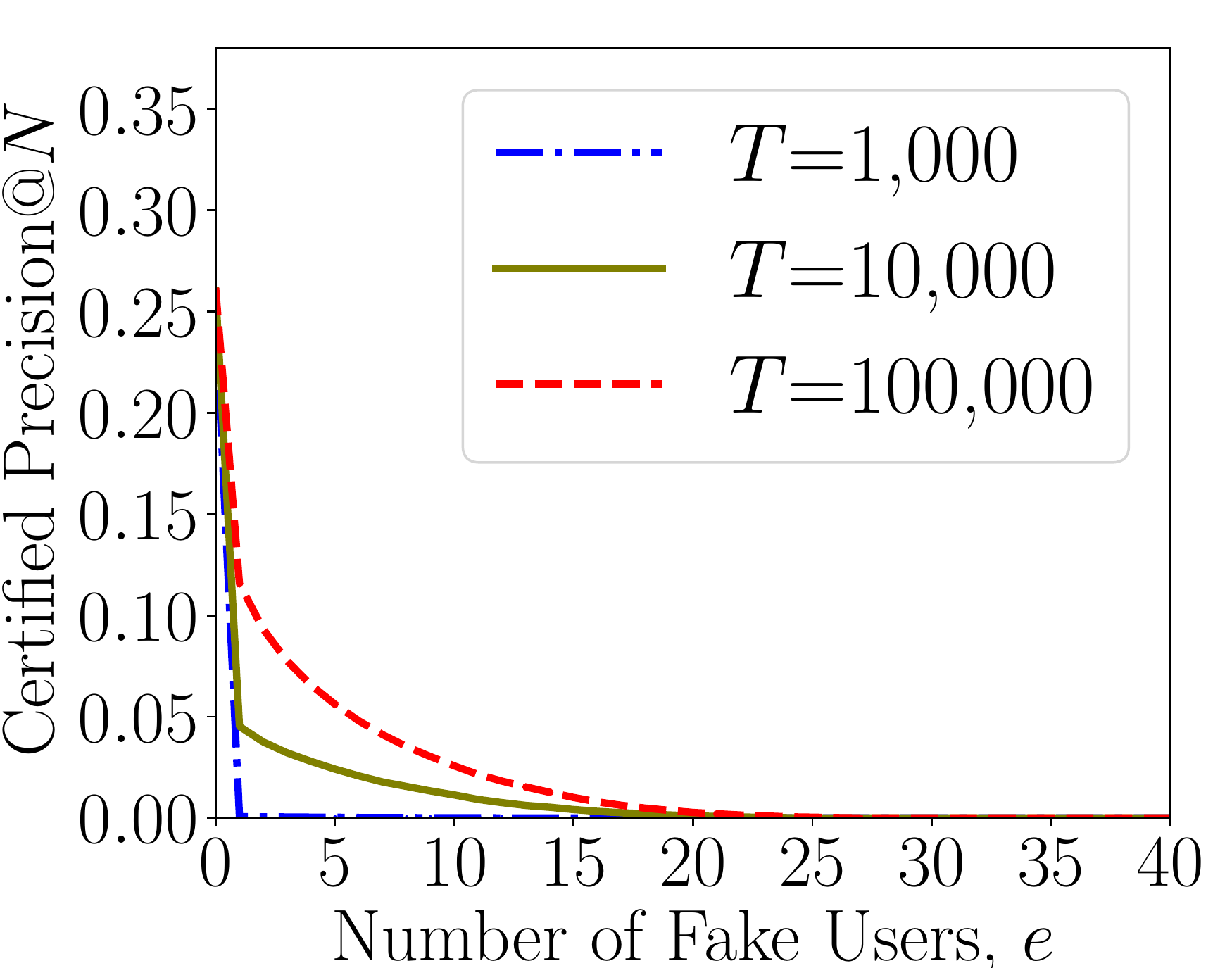}}
{\includegraphics[width=0.16\textwidth]{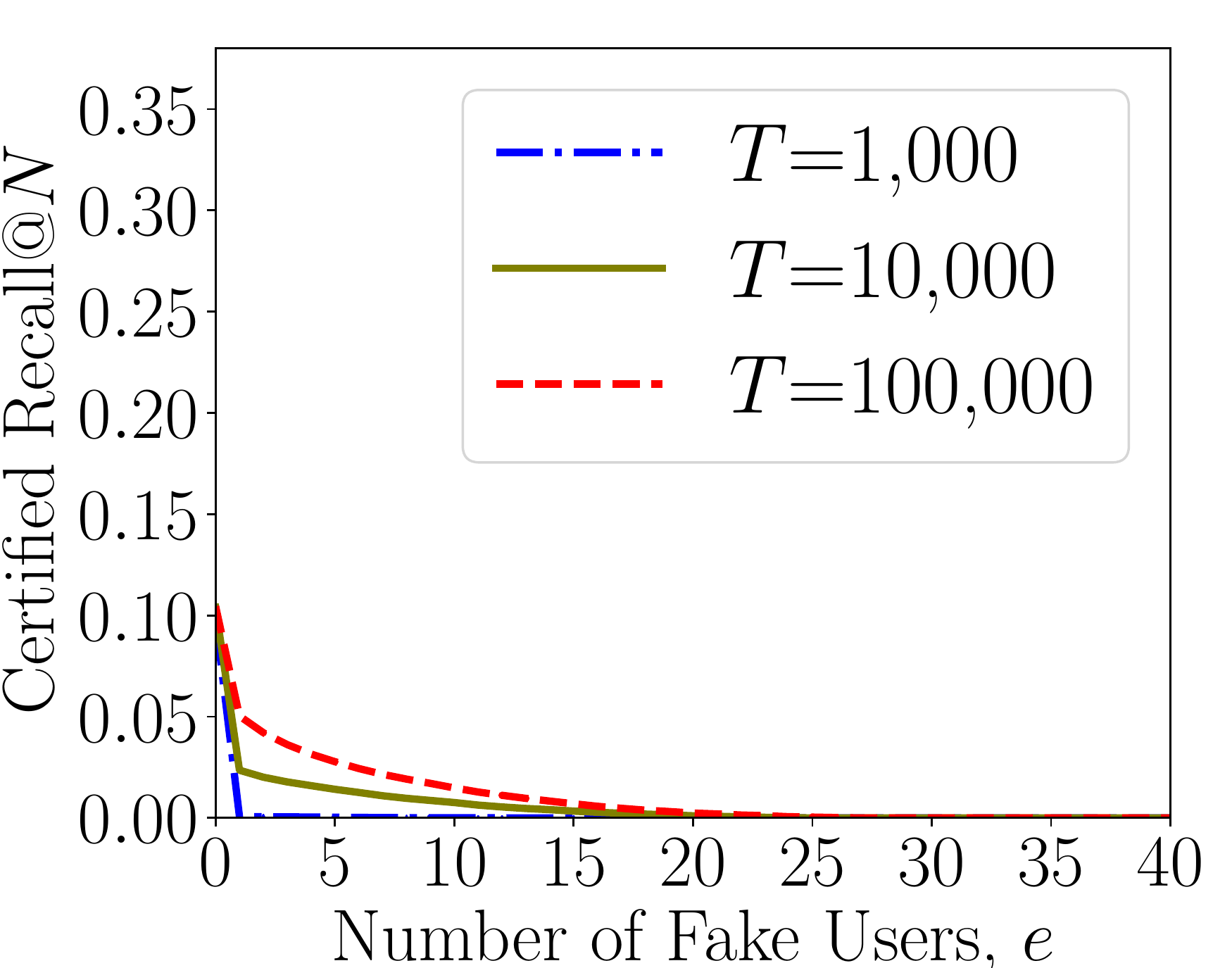}}
{\includegraphics[width=0.16\textwidth]{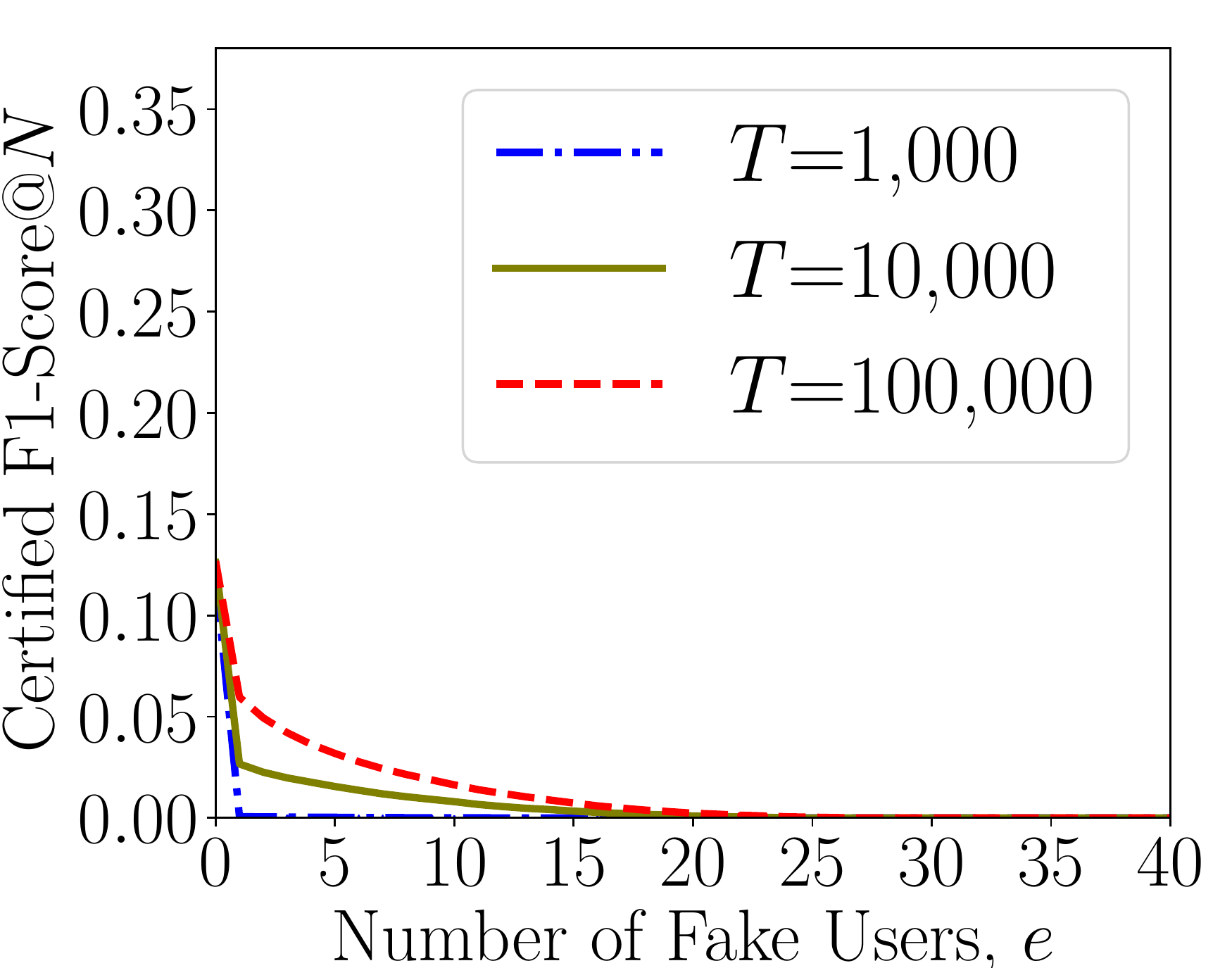}}
\vspace{-3mm}
\caption{Impact of $T$ on the certified Precision@$N$, certified Recall@$N$, and certified F1-Score@$N$ of ensemble IR for MovieLens-100k (left three) and MovieLens-1M (right three), where $N=10$.}
\label{impact_of_T}
\end{figure*}

\begin{figure*}[!t]
	 \centering
   \vspace{-3mm}
{\includegraphics[width=0.16\textwidth]{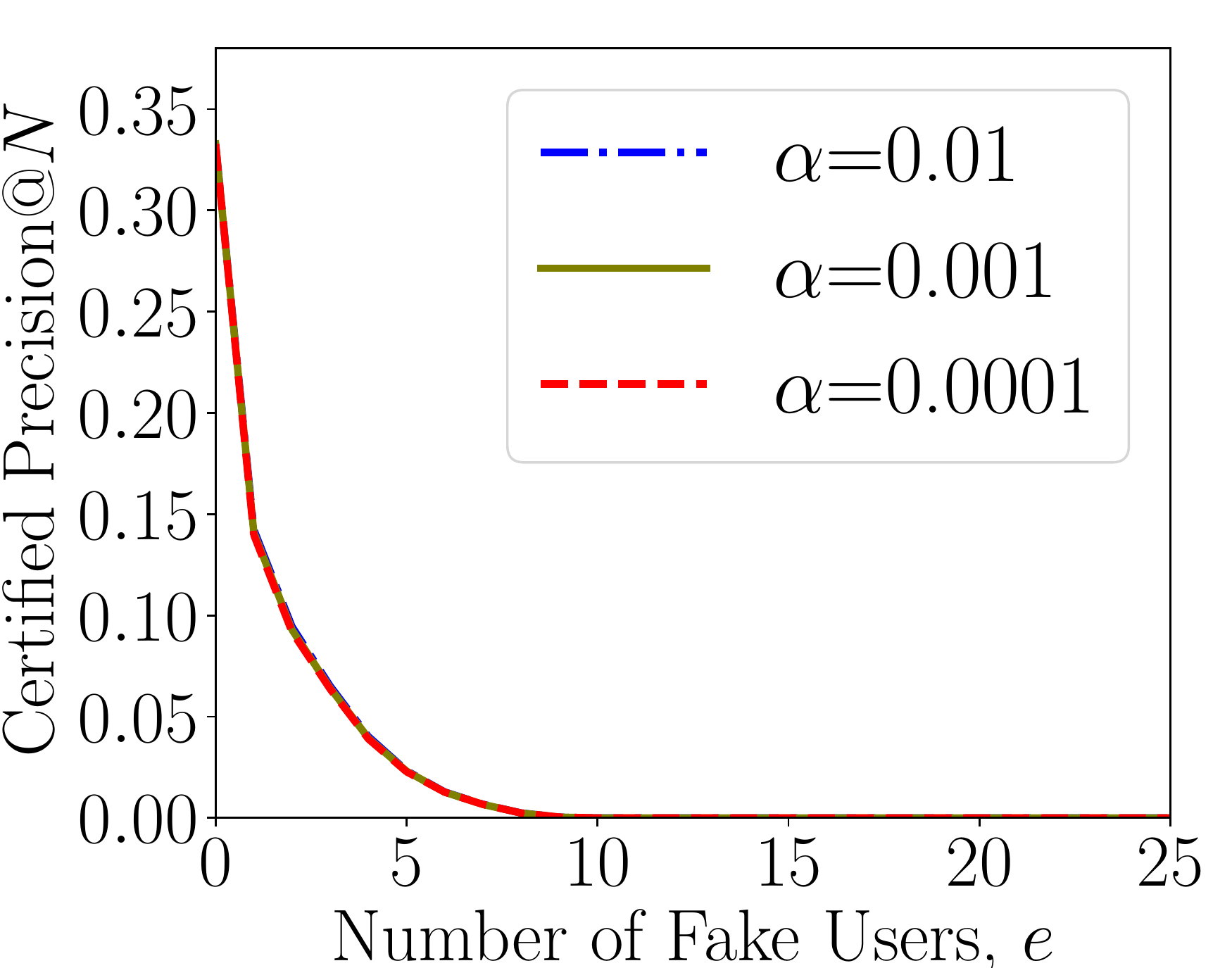}}
{\includegraphics[width=0.16\textwidth]{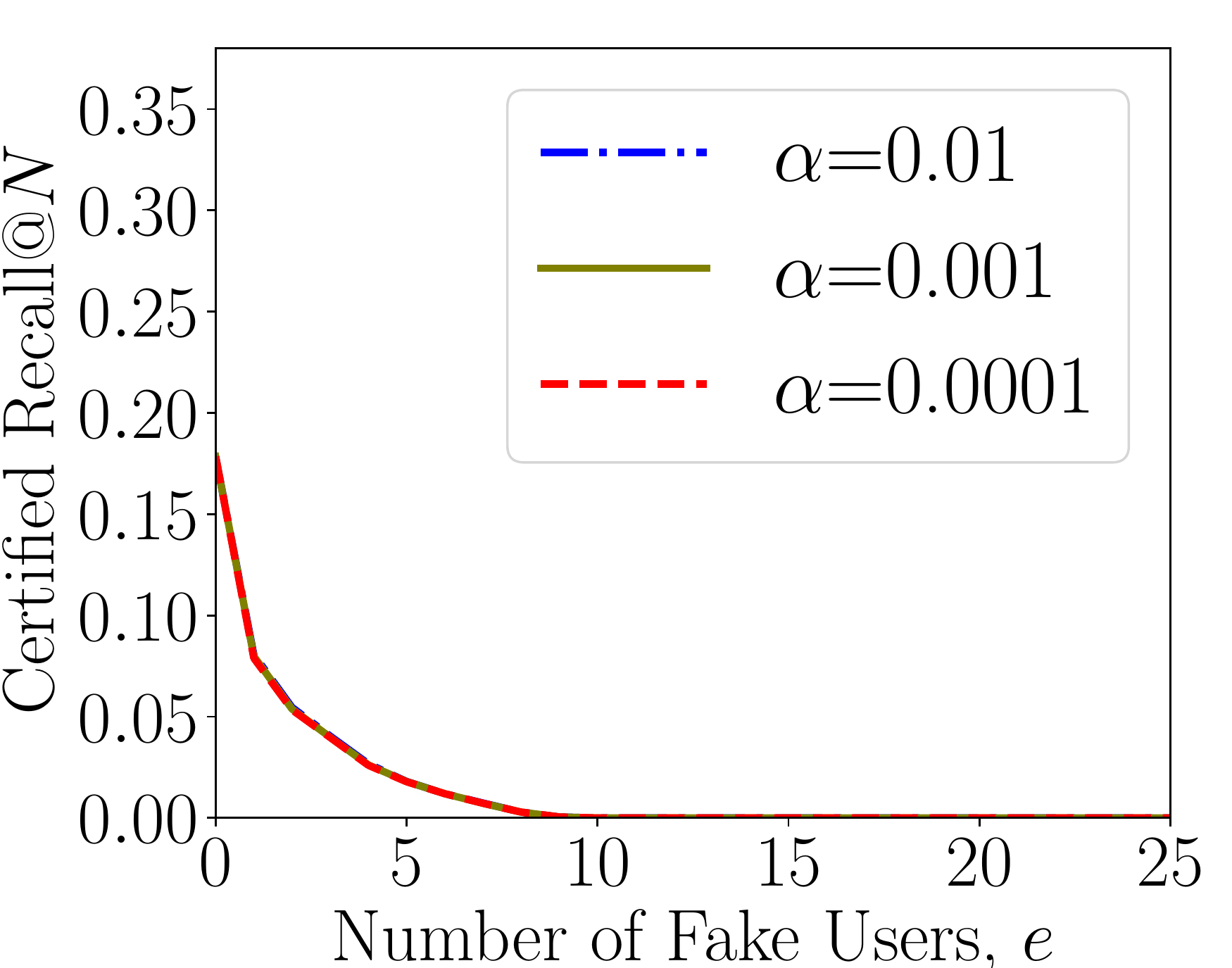}}
{\includegraphics[width=0.16\textwidth]{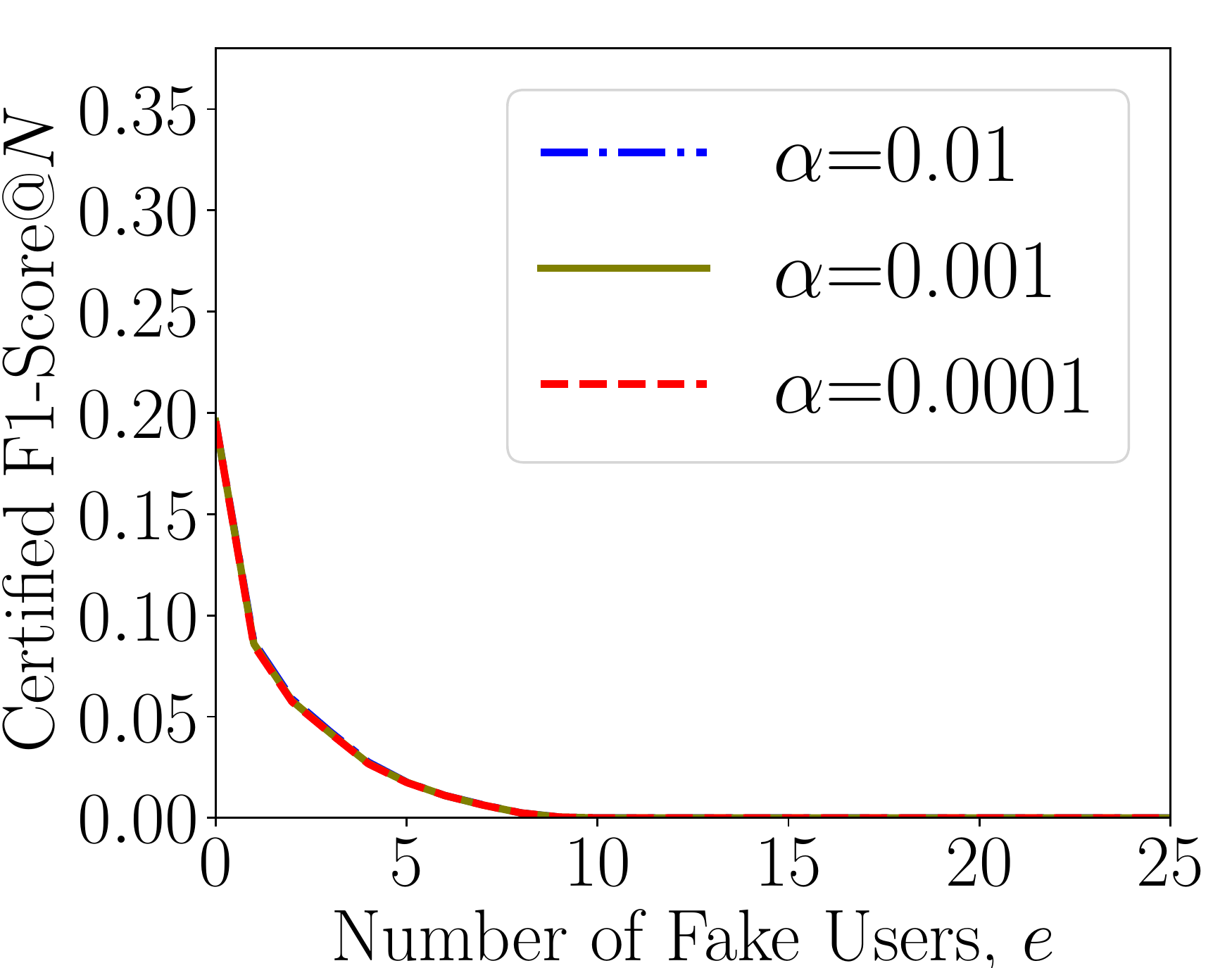}}
{\includegraphics[width=0.16\textwidth]{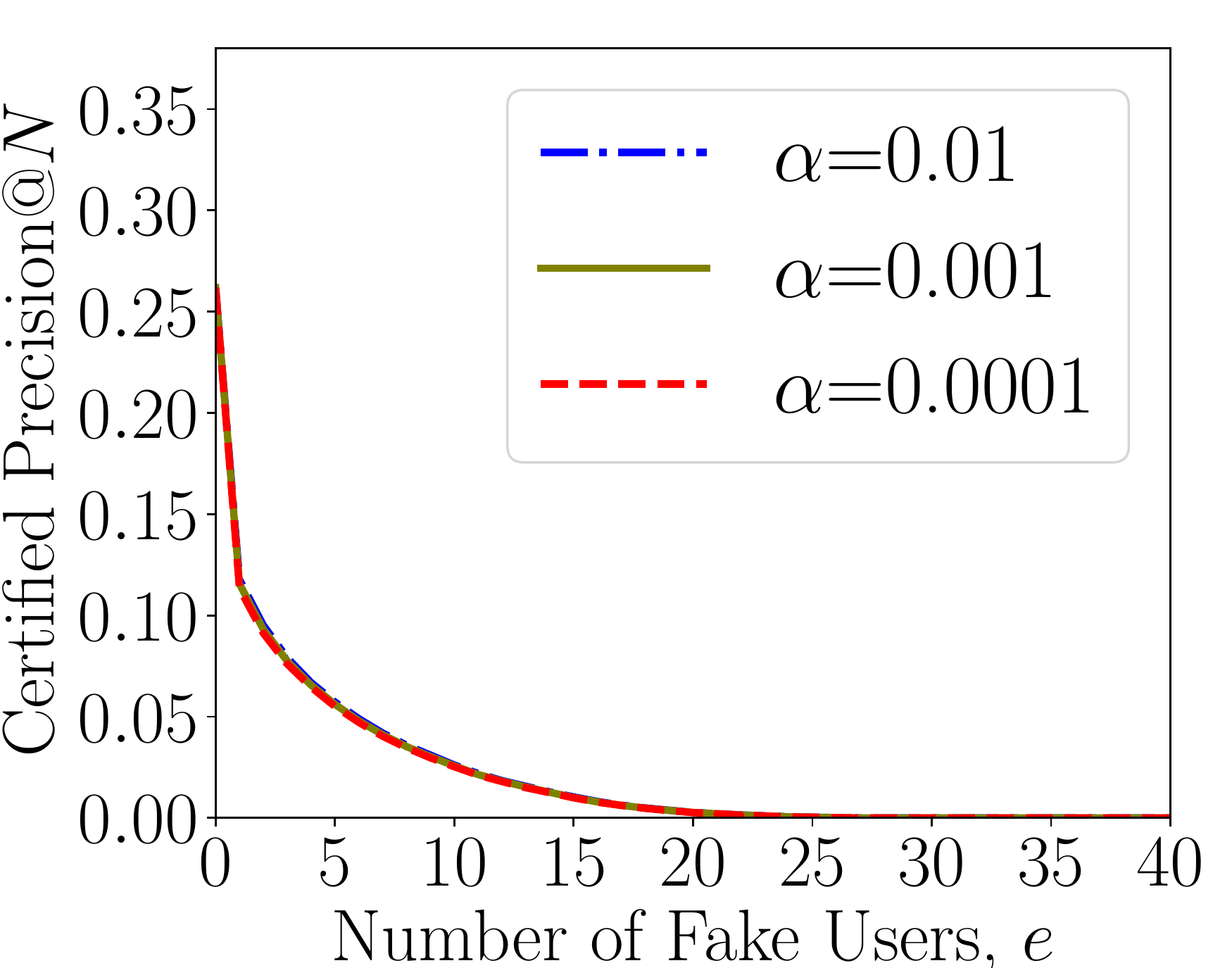}}
{\includegraphics[width=0.16\textwidth]{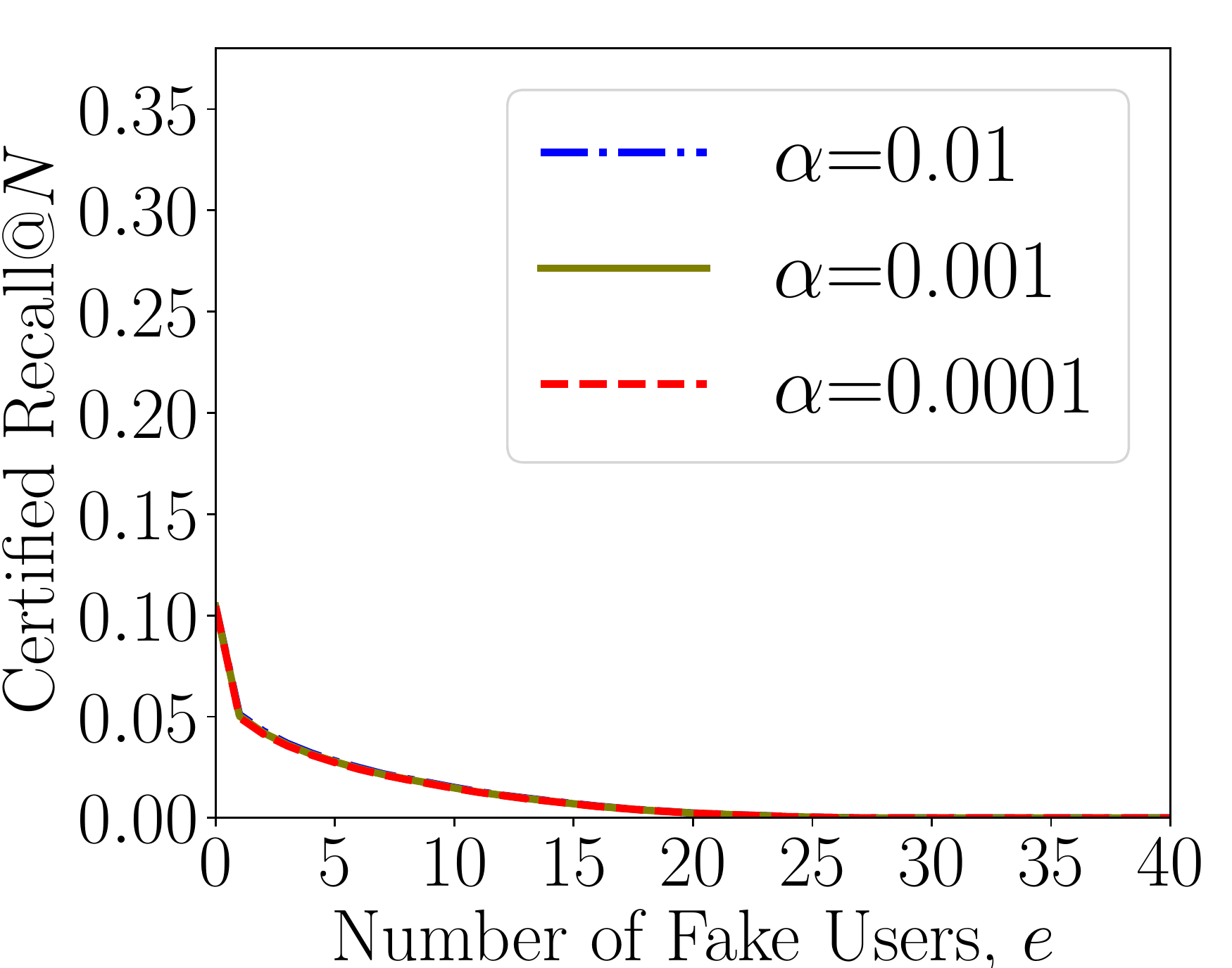}}
{\includegraphics[width=0.16\textwidth]{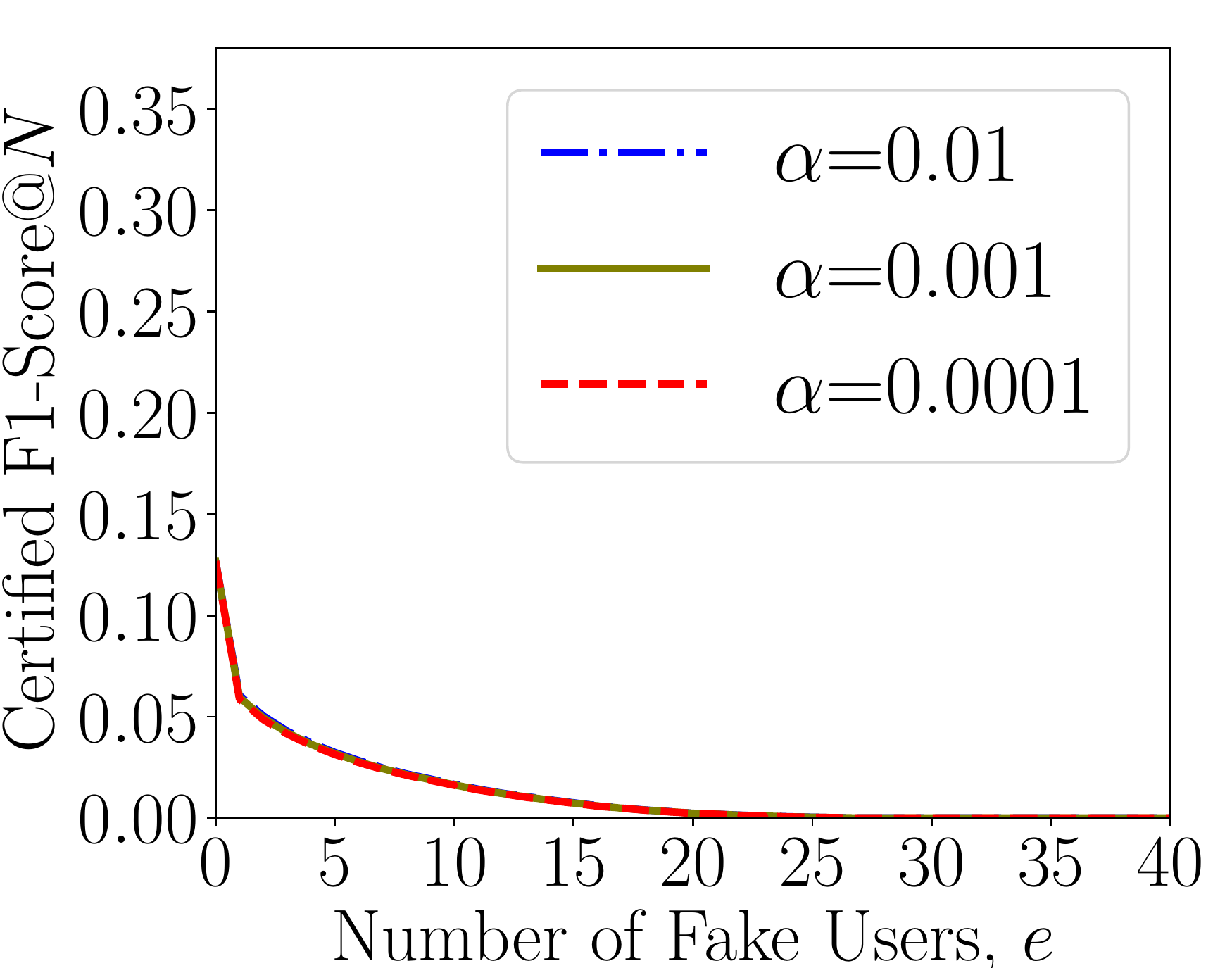}}
\vspace{-3mm}
\caption{Impact of $\alpha$ on the certified Precision@$N$, certified Recall@$N$, and certified F1-Score@$N$ of ensemble IR for MovieLens-100k (left three) and MovieLens-1M (right three), where $N=10$.}
\label{impact_of_alpha}
\end{figure*}

\begin{figure*}[!t]
	 \centering
  \vspace{-3mm}
{\includegraphics[width=0.16\textwidth]{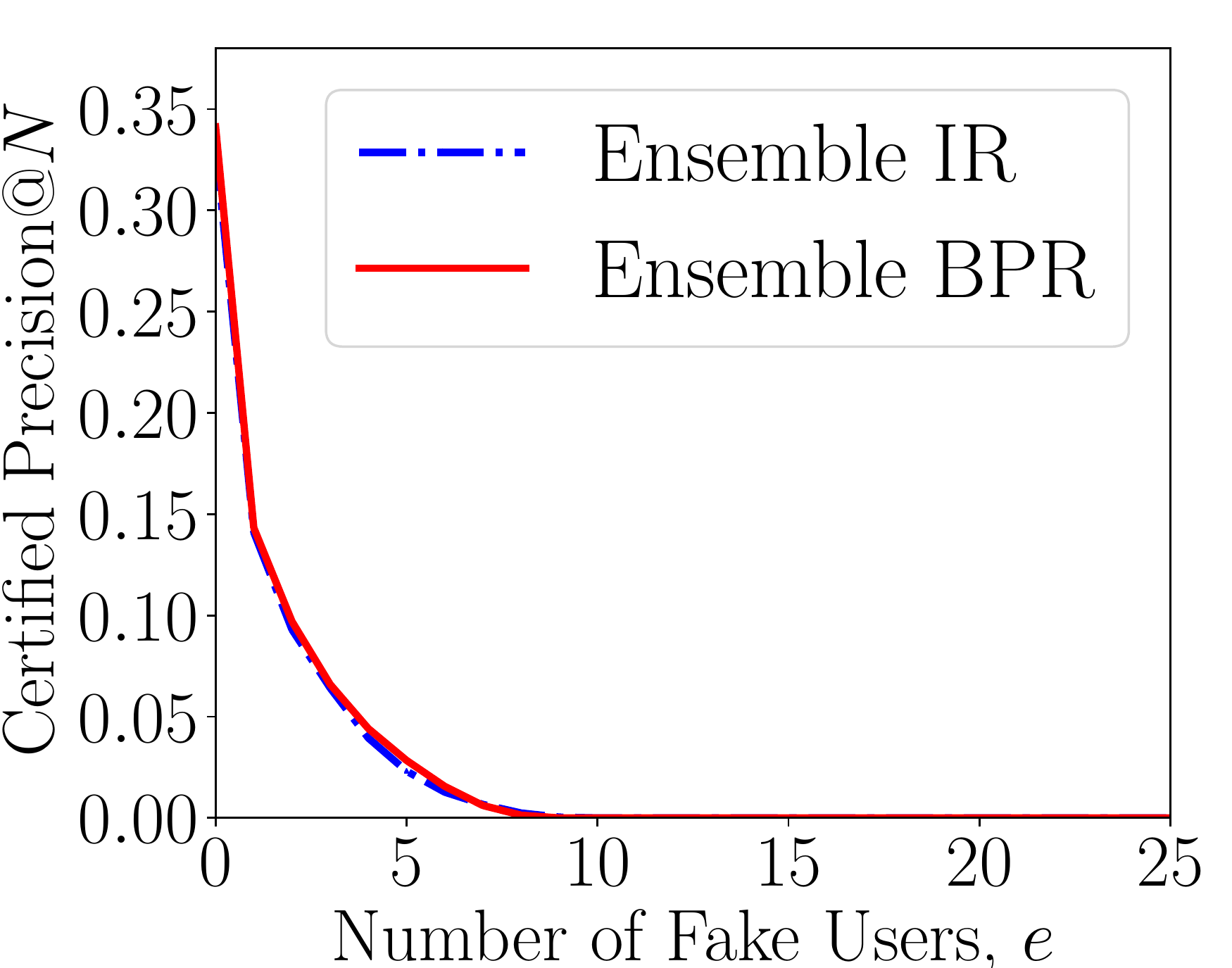}}
{\includegraphics[width=0.16\textwidth]{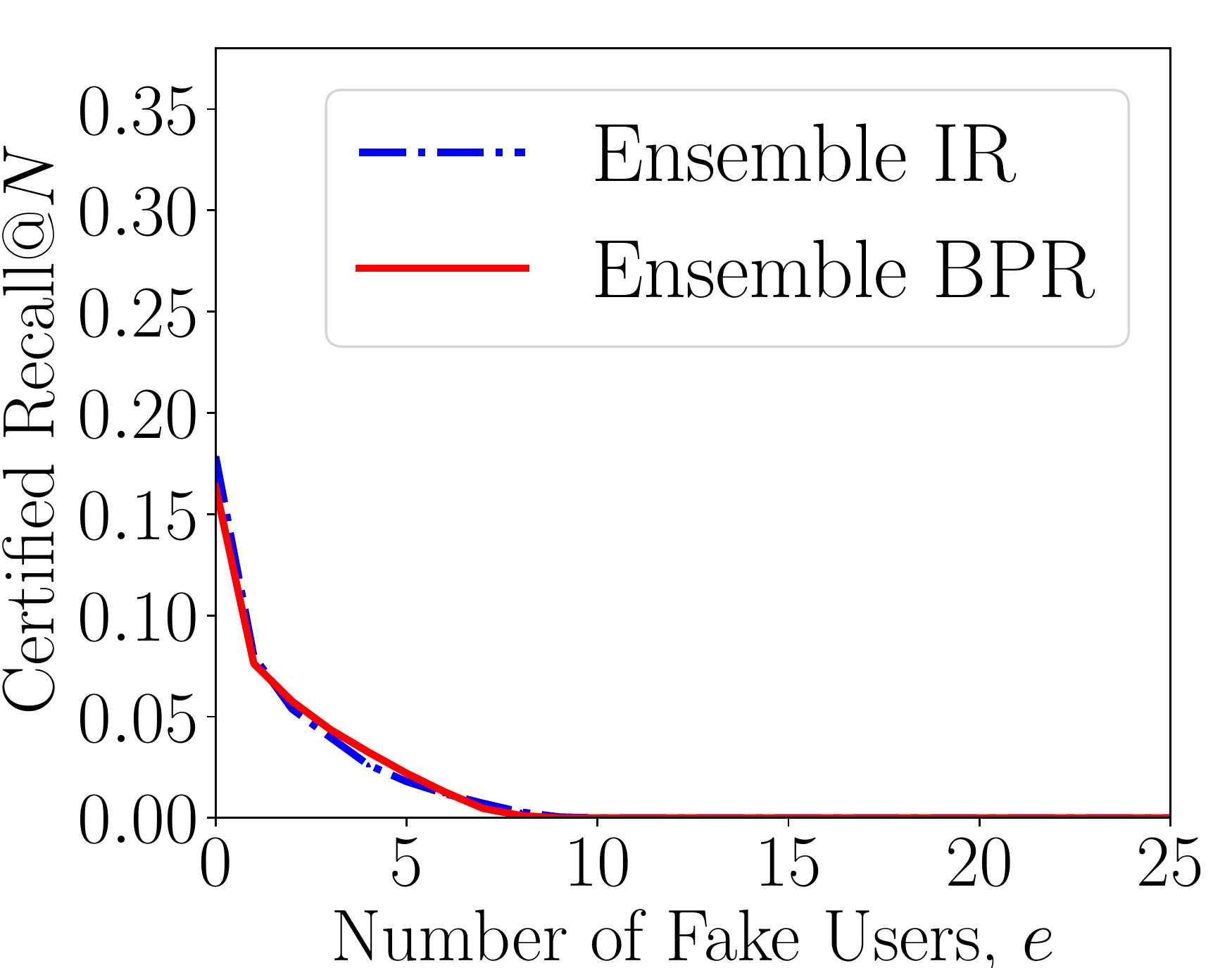}}
{\includegraphics[width=0.16\textwidth]{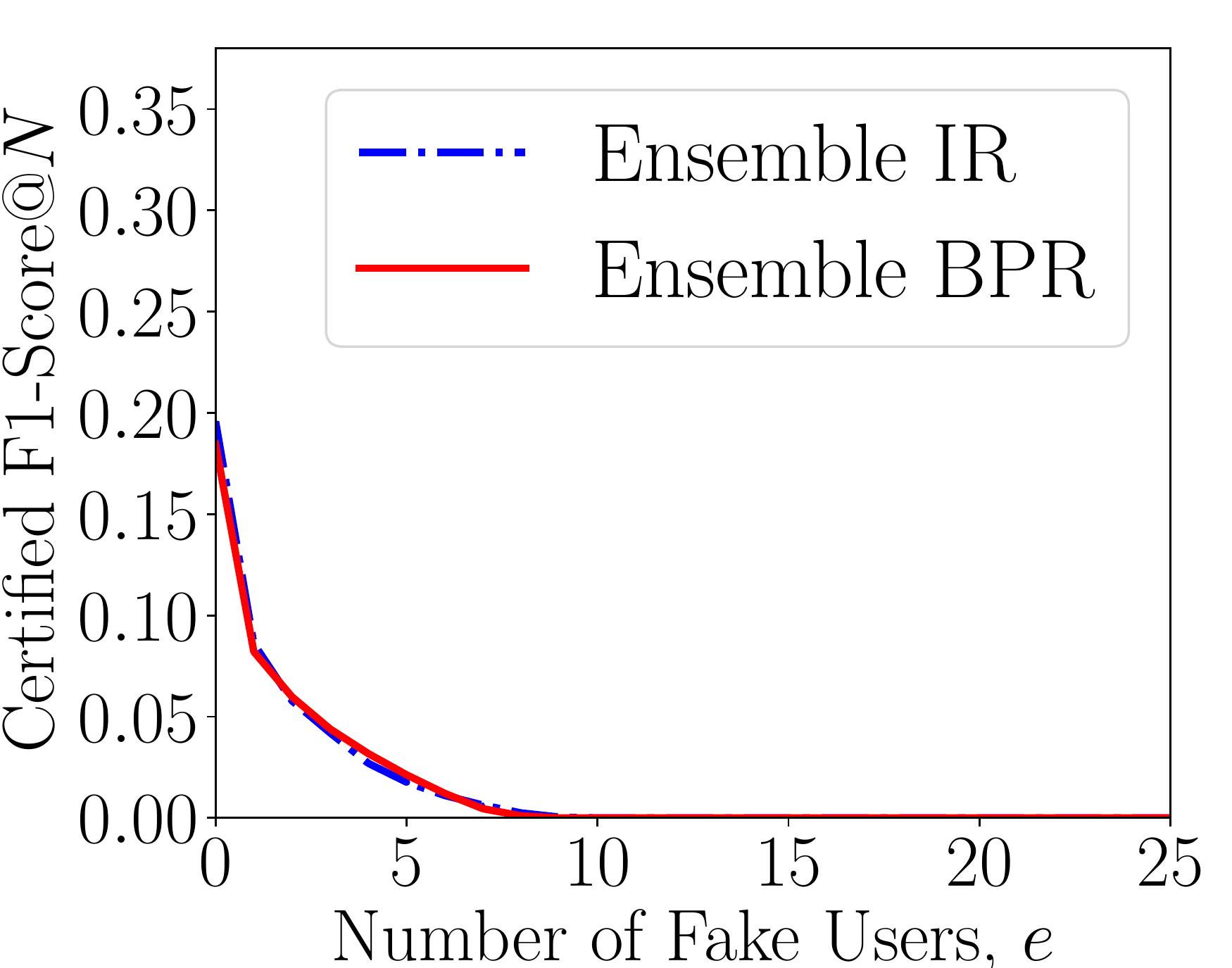}}
{\includegraphics[width=0.16\textwidth]{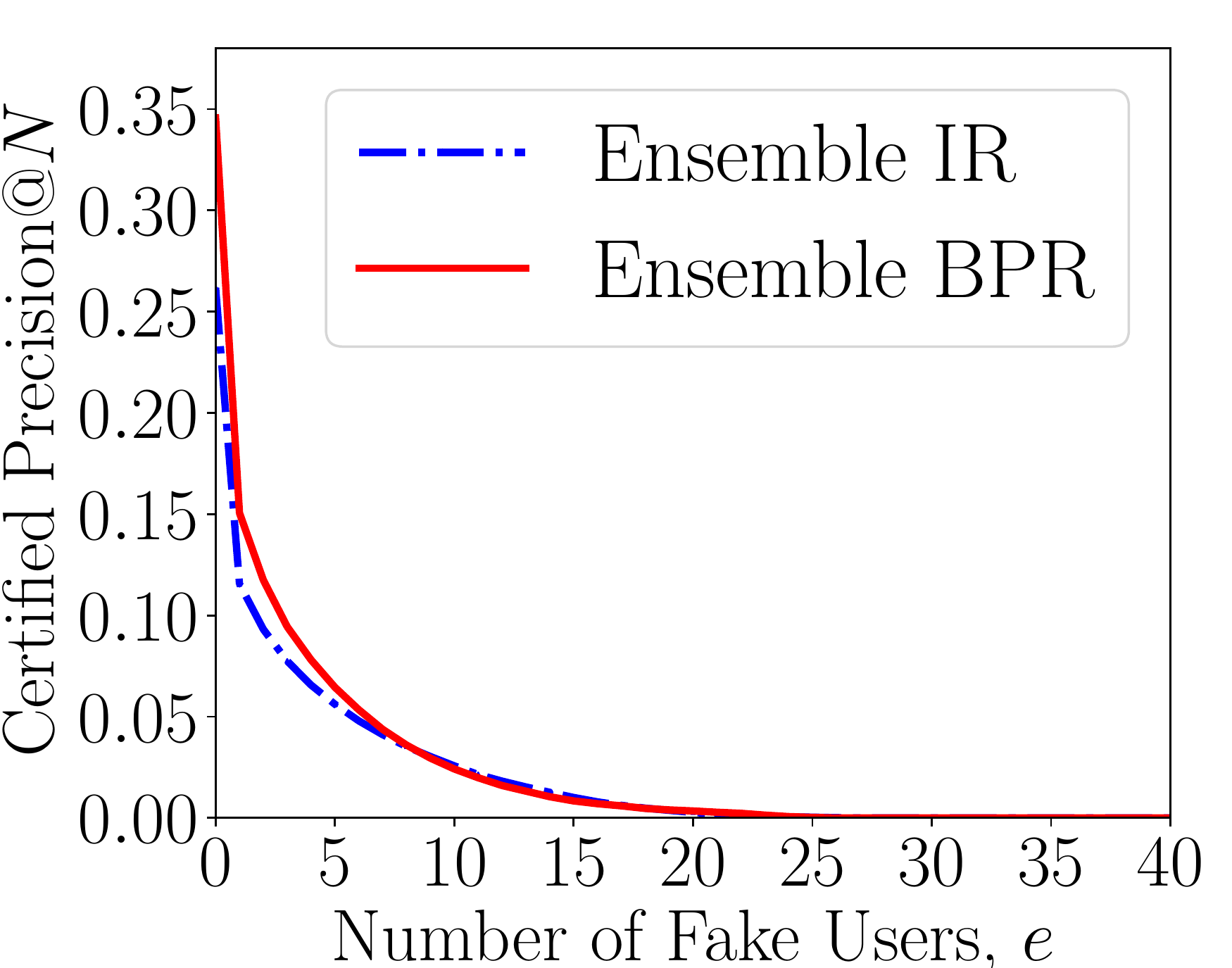}}
{\includegraphics[width=0.16\textwidth]{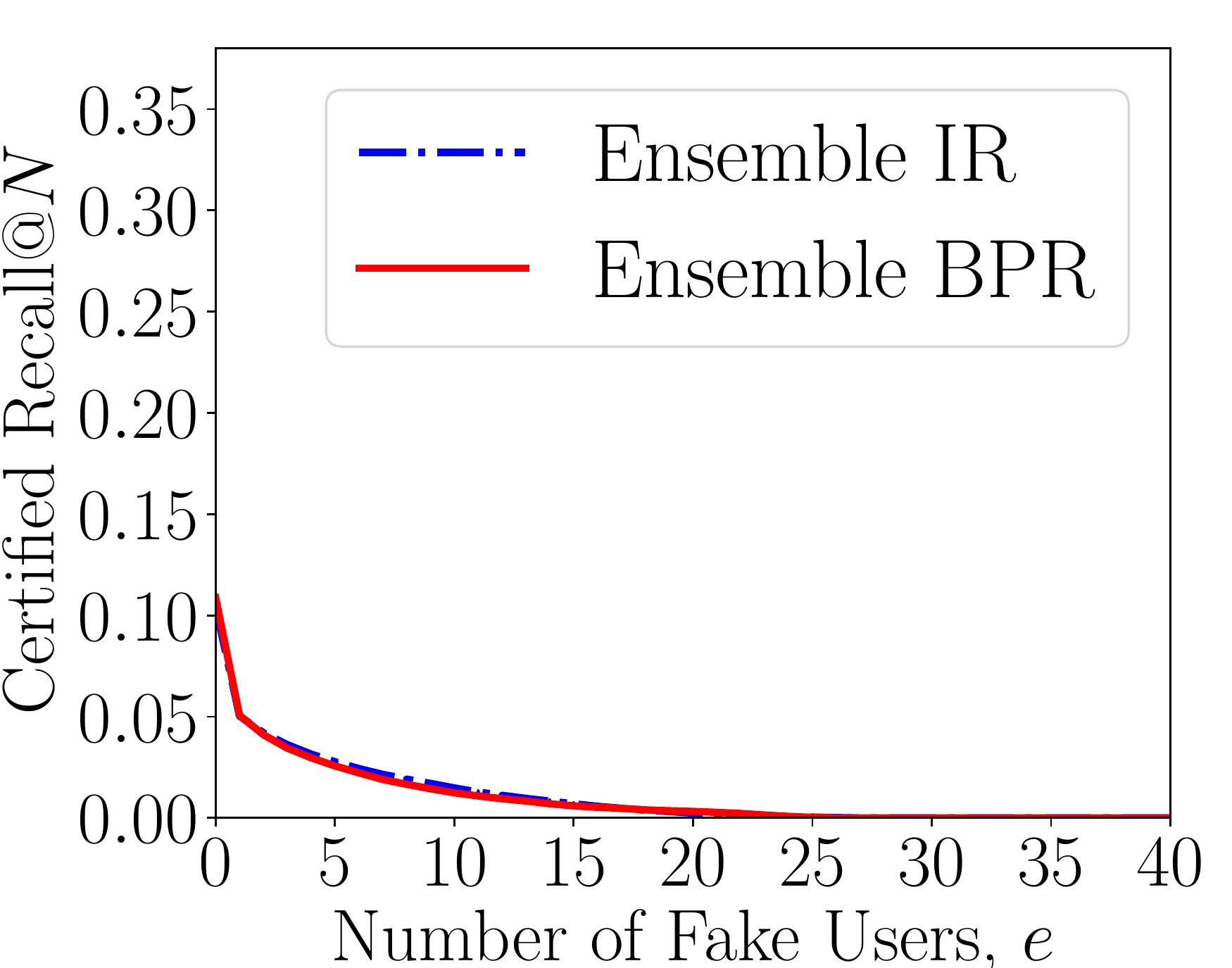}}
{\includegraphics[width=0.16\textwidth]{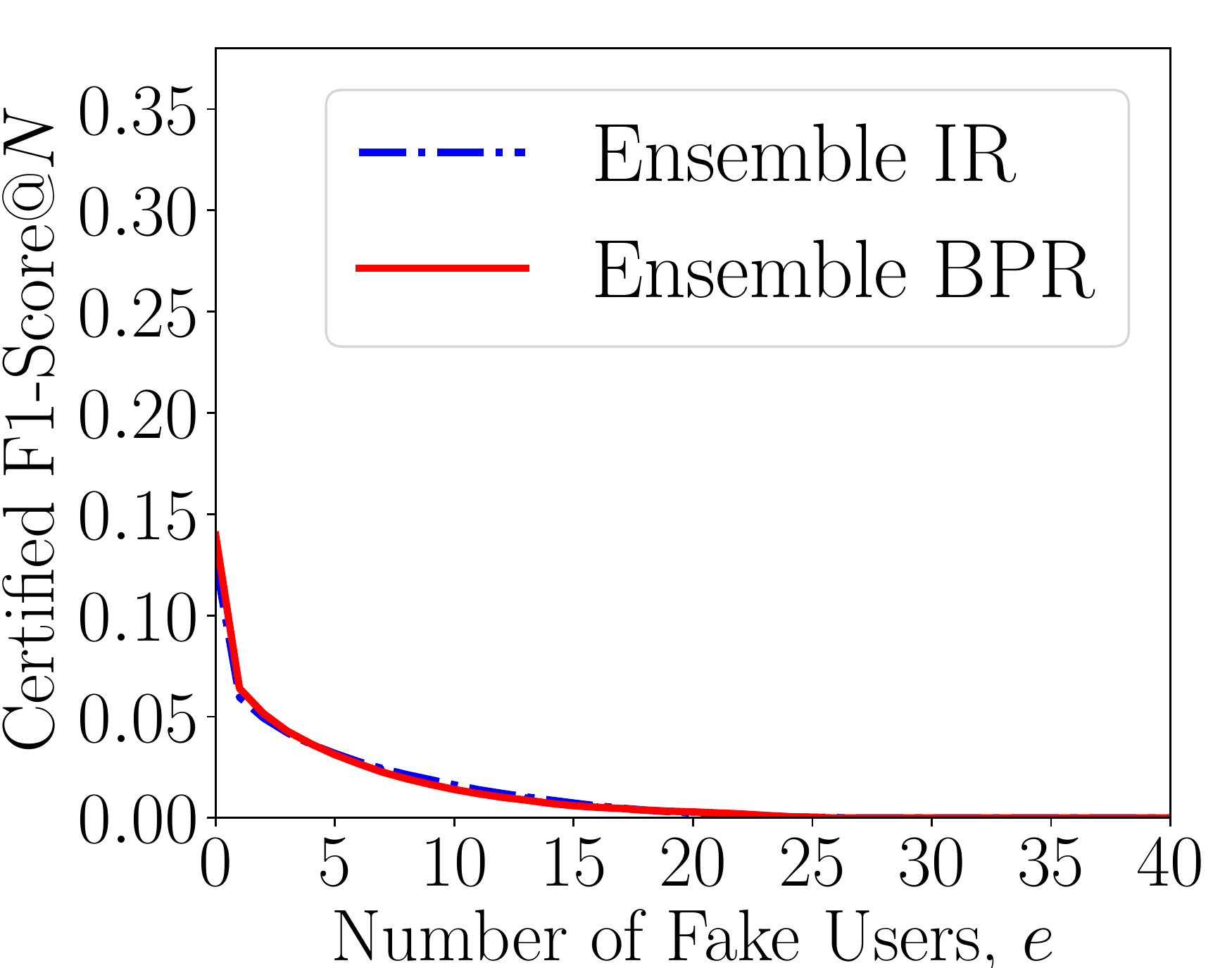}}
\vspace{-3mm}
\caption{Comparing the certified Precision@$N$, certified Recall@$N$, and certified F1-Score@$N$ of ensemble IR and ensemble BPR for MovieLens-100k (left three) and MovieLens-1M (right three), where $N=10$.}
\label{compare_sar_bpr}
\end{figure*}

\vspace{-1mm}
\myparatight{Impact of $N$} Figure~\ref{impact_of_N} shows the impact of $N$. The results show that  $N$ achieves a tradeoff between Precision@$N$ under no attacks (i.e., $e=0$) and robustness under attacks. Specifically, a smaller $N$ achieves a higher Precision@$N$ under no attacks but the certified Precision@$N$ decreases more quickly as $e$ increases. The certified Recall@$N$ increases as $N$ increases. The reason is that more items are recommended to each user as $N$ increases. 
The certified F1-Score@$N$ drops more quickly as $e$ increases when $N$ is smaller, because the certified Recall@$N$ drops more quickly.

\myparatight{Impact of $s$} Figure~\ref{impact_of_s} shows the impact of $s$. We have two observations. First, our method achieves similar Precision@$N$, Recall@$N$, and F1-Score@$N$ for different $s$ when there are no attacks (i.e., $e=0$). Second, a larger $s$ achieves a larger certified Precision@$N$, certified Recall@$N$, or certified F1-Score@$N$ when $e$ is small, but they decrease more quickly as $e$ increases. This is because it's more likely to sample fake users in a submatrix when $s$ is larger.

\myparatight{Impact of $T$ and $\alpha$}  Figure~\ref{impact_of_T} and~\ref{impact_of_alpha} show the impact of $T$ and $\alpha$, respectively. We have the following observations. First, Precision@$N$, Recall@$N$, or F1-Score@$N$ is similar for different $T$ when there are no attacks. In other words, a small $T$ is enough for our ensemble recommender system to achieve good recommendation performance when there are no attacks. Second,  certified Precision@$N$, certified Recall@$N$, or certified F1-Score@$N$ increases as $T$ or $\alpha$ increases. The reason is that a larger $T$ or $\alpha$ can produce tighter estimated item-probability lower/upper bounds, based on which we may compute larger certified intersection sizes $r$ in our Algorithm~\ref{alg:certify}. Therefore, we use a larger $T$ by default in our experiments to better show the certified Precision@$N$, certified Recall@$N$, and certified F1-Score@$N$ of PORE. We also observe that the certified Precision@$N$, certified Recall@$N$, and certified F1-Score@$N$ are insensitive to $\alpha$ once it is small enough. 

\myparatight{Ensemble IR vs. ensemble BPR}
Figure~\ref{compare_sar_bpr} compares ensemble IR and ensemble BPR. The results show that they achieve similar certified Precision@$N$/Recall@$N$/F1-Score@$N$. One exception is that ensemble BPR achieves higher certified Precision@$N$ on MovieLens-1M dataset. 

\myparatight{Standard recommender system vs. ensemble recommender system under no attacks} Table~\ref{asr-train-test} compares  standard recommender systems and our ensemble recommender systems with respect to the standard Precision@$N$, Recall@$N$, and F1-Score@$N$ when there are no attacks, where  $s=300$ for MovieLens-100k and $s=1,000$ for MovieLens-1M. A standard recommender system leverages IR (or BPR) to train a single recommender system on the entire rating-score matrix.  The results show that our ensemble recommender system achieves comparable performance with a standard recommender system when there are no attacks. \neil{Ensemble BPR achieves higher Precision@N than ensemble IR on both datasets. The reason is that  BPR achieves higher precision than  IR at training the base recommender systems.}

\begin{figure}[!t]
	 \centering
{\includegraphics[width=0.23\textwidth]{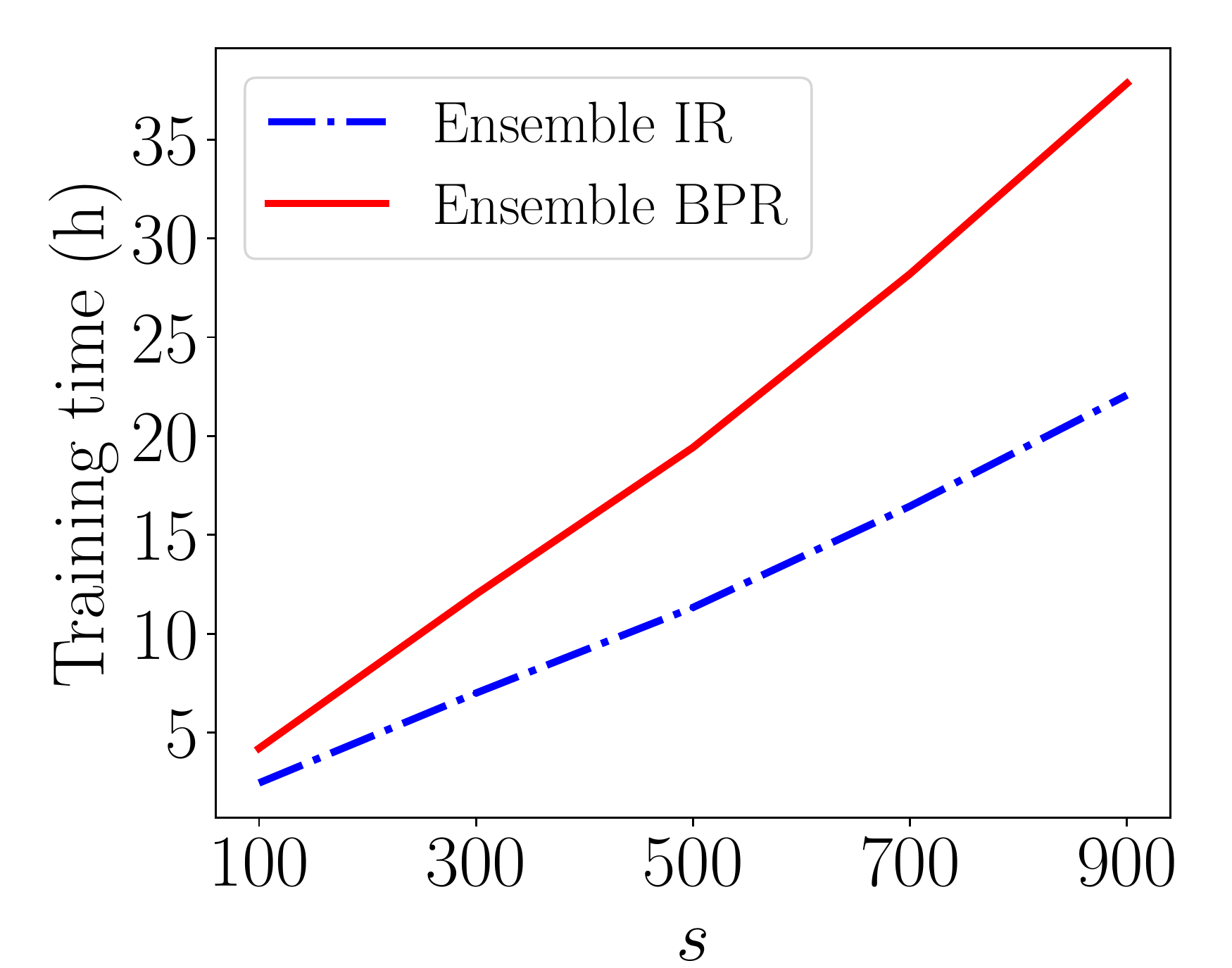}}
{\includegraphics[width=0.23\textwidth]{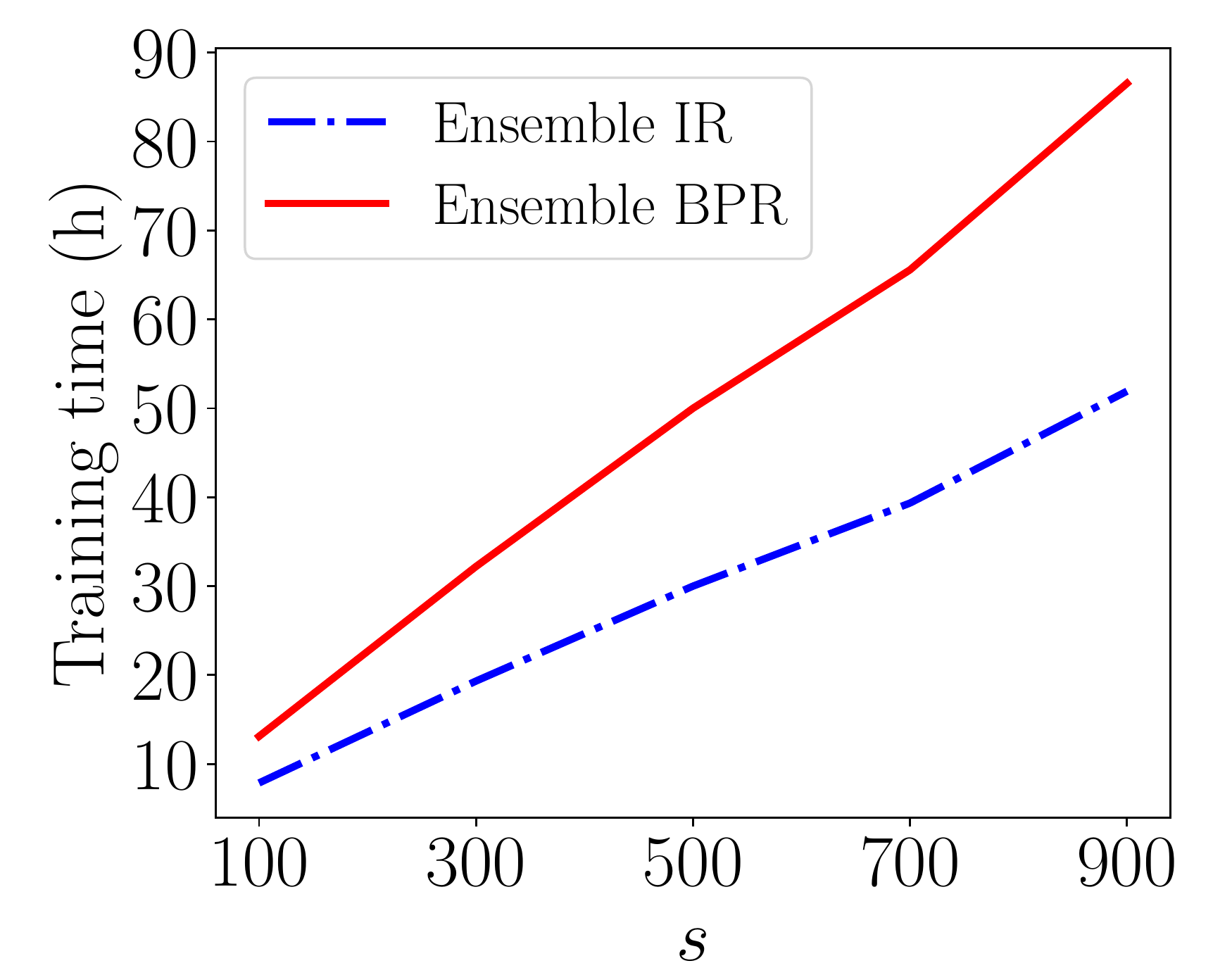}}
\vspace{-1mm}
\caption{\neil{Training time of PORE as a function of $s$ on MovieLens-100k (left) and MovieLens-1M (right).}}
\label{computation_complexity_s}
\end{figure}

\begin{figure}[!t]
	 \centering
{\includegraphics[width=0.23\textwidth]{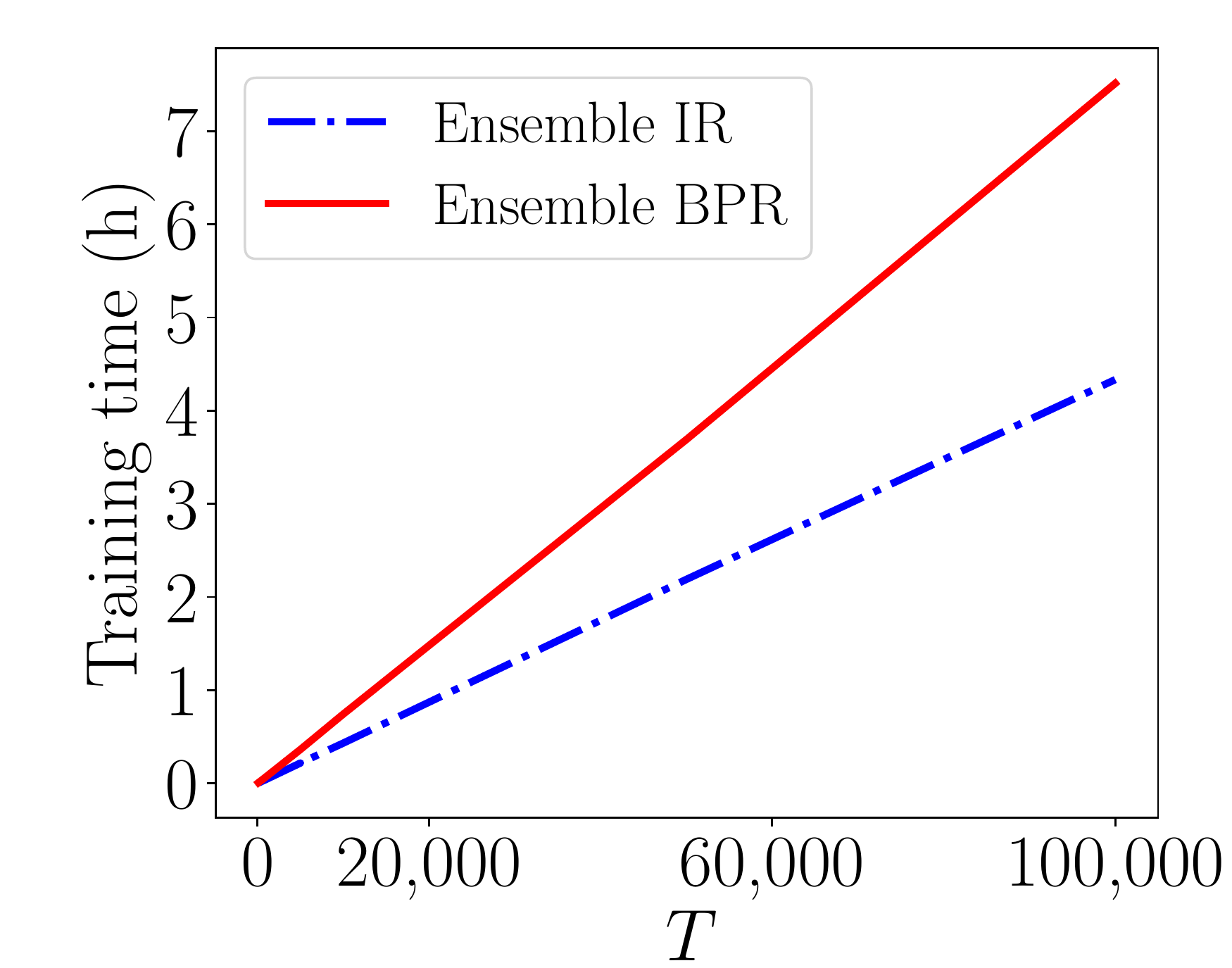}}
{\includegraphics[width=0.23\textwidth]{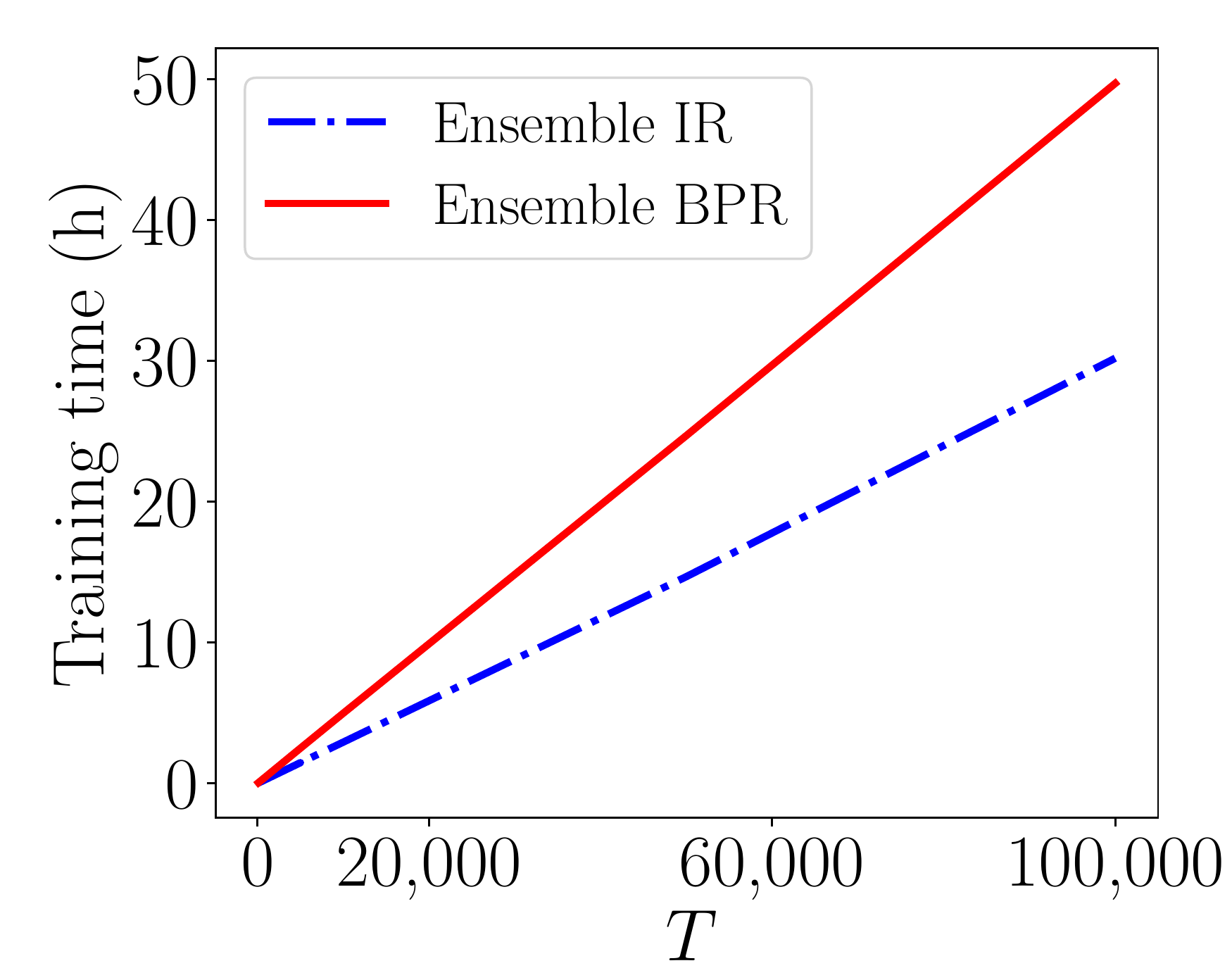}}
\vspace{-1mm}
\caption{\neil{Training time of PORE as a function of $T$ on MovieLens-100k (left) and MovieLens-1M (right).}}
\label{computation_complexity_T}
\end{figure}

\myparatight{\neil{Training time of ensemble IR and BPR}} \neil{PORE trains $T$ base recommender systems, each of which is trained using rating scores of $s$ users. 
Figure~\ref{computation_complexity_s} shows the training time of ensemble IR/BPR as a function of $s$, while Figure~\ref{computation_complexity_T} shows the impact of $T$ on the  training time of ensemble IR/BPR. As expected, the training time of ensemble IR/BPR increases linearly as $s$ or $T$ increases. This is because a larger $s$ means each base recommender system is trained using more data, while a larger $T$ means more base recommender systems are trained. Ensemble IR is more efficient than ensemble BPR because IR is more efficient than BPR. }

\begin{figure}[!t]
	 \centering
{\includegraphics[width=0.33\textwidth]{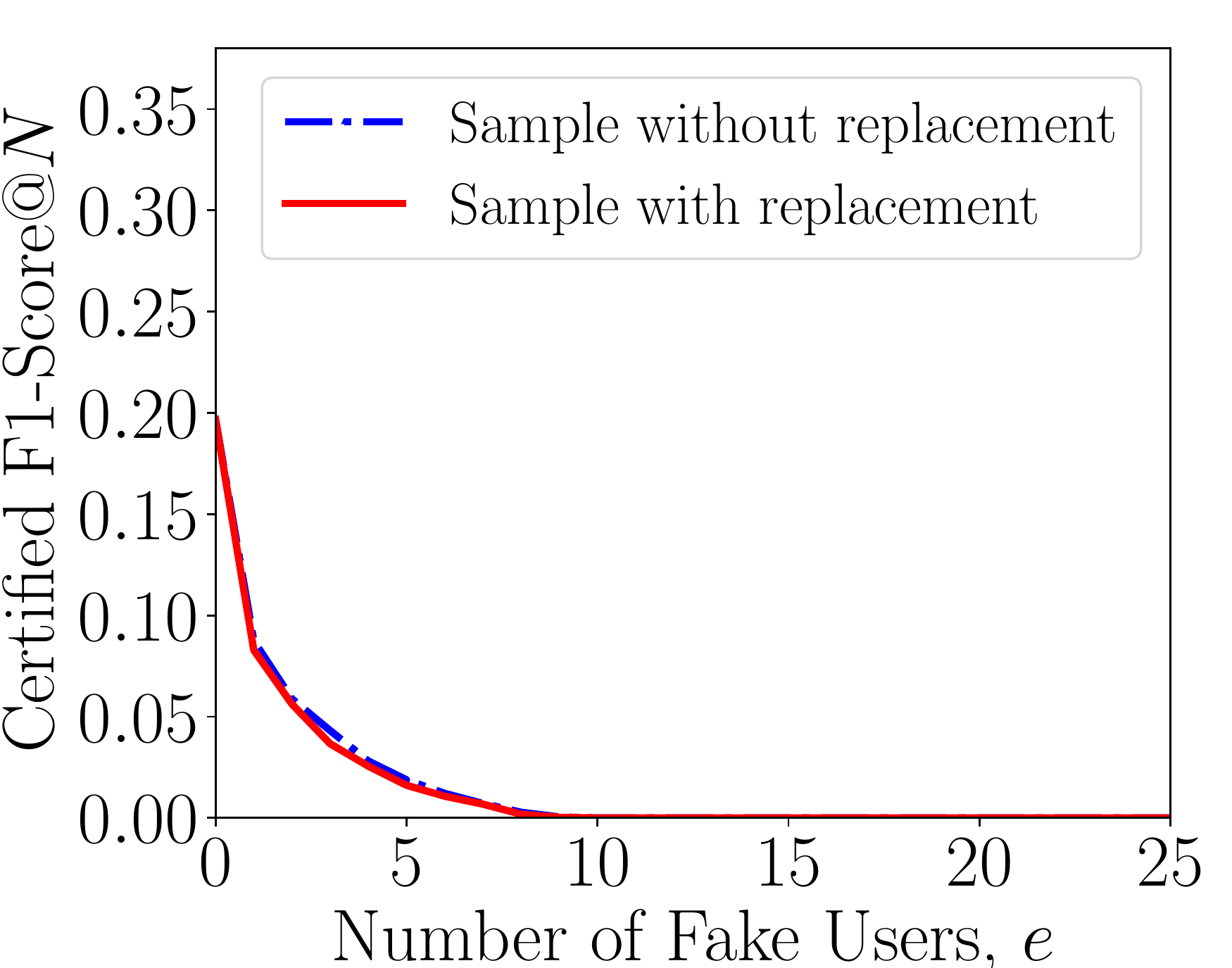}}
\vspace{-1mm}
\caption{Comparing  certified F1-Score@$N$ of sampling with and without replacement for ensemble IR on MovieLens-100k, where $N=10$.}
\label{samplingcomparison}
\end{figure}

\begin{figure}[!t]
	 \centering
{\includegraphics[width=0.33\textwidth]{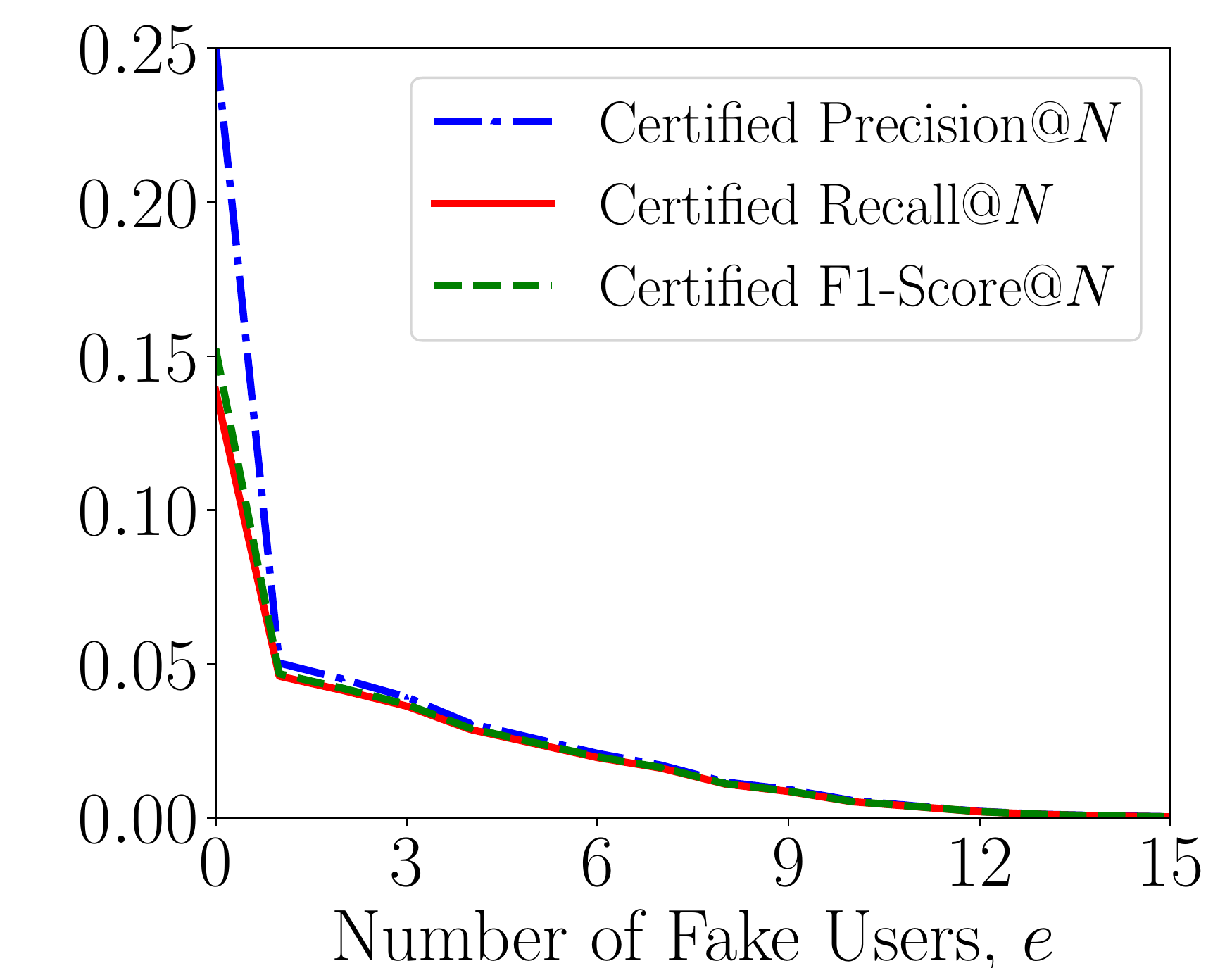}}
\vspace{-1mm}
\caption{\neil{Certified Precision@$N$, certified Recall@$N$, and certified F1-Score@$N$ of ensemble IR for MovieLens-10M, where $N=10$.}}
\label{10m}
\end{figure}

\myparatight{\neil{Sampling with vs. without replacement}}  \neil{ PORE randomly samples a submatrix without replacement, i.e., the sampled $s$ users are different in a submatrix. Sampling with replacement means that the $s$ users may have duplicates, i.e., a submatrix may include rating scores of less than $s$ unique users. We can extend our theoretical guarantees to sampling with replacement. The theoretical analysis is similar to  sampling without replacement, so we omit it for simplicity. Figure~\ref{samplingcomparison} compares the certified F1-Score@$N$ of sampling with and without replacement for ensemble IR on MovieLens-100k dataset. Our results show that sampling with and without replacement achieves comparable certified F1-Score@$N$. They also achieve comparable certified Precision@$N$ and certified Recall@$N$, which we omit for simplicity. }

\myparatight{\neil{Evaluation on a large dataset}} \neil{We also evaluate  PORE on MovieLens-10M~\cite{movielens}, which consists of around $10,000,000$ rating scores. We set $s=5,000$, $T=1,000$, and IR as the base recommender system algorithm, and we adopt the default settings for other parameters.  We set a smaller $T$ than our previous experiments because it is more expensive to train each base recommender system on MovieLens-10M.  Figure~\ref{10m} shows the results. Our results indicate that our method is applicable to a large dataset and can derive certified performance guarantees against data poisoning attacks. We note that the certified Precision@$N$/Recall@$N$/F1-Score@$N$ of bagging also reduces to 0 with just 1 fake user.}

%% file: relatedwork.tex
\section{Related Work}
\label{related}

\myparatight{Data poisoning attacks to recommender systems} Many data poisoning attacks to recommender systems~\cite{lam2004shilling,burke2005segment,mobasher2007toward,seminario2014attacking,li2016data,yang2017fake,fang2018poisoning,christakopoulou2019adversarial,zhang2020practical,fang2020influence,song2020poisonrec,wadhwa2020data,tang2020revisiting,huangdata2021,zhang2021data,wu2021triple} have been proposed. Early attacks are algorithm-agnostic, i.e., the crafted rating scores of fake users do not depend on the recommender system algorithm~\cite{lam2004shilling,burke2005segment,mobasher2007toward}. 
Recently,  more advanced data poisoning attacks~\cite{li2016data,yang2017fake,fang2018poisoning,fang2020influence,song2020poisonrec,huangdata2021,zhang2021data,wu2021triple} have been optimized for  specific recommender system algorithms. For instance, Yang et al.~\cite{yang2017fake} proposed to inject fake co-visitations to poison association rule based recommender systems, Fang et al. proposed optimized data poisoning attacks to graph based recommender systems~\cite{fang2018poisoning} and matrix factorization based recommender systems~\cite{fang2020influence}, and Huang et al.~\cite{huangdata2021} proposed data poisoning attacks optimized to deep learning based recommender systems. 

\myparatight{Empirical defenses against data poisoning attacks to recommender systems} 
One family of defenses~\cite{burke2006classification,zhang2006attack,wu2012hysad,zhang2014hht,fang2018poisoning} aim to detect  fake users via analyzing their abnormal rating score patterns. The key assumption is  that the rating scores of fake users and genuine users have different patterns. 
Fo instance, Burke~\cite{burke2006classification} extracted features from each user's rating scores and trained a classifier to predict whether a user is fake or not. 
\neil{Another family of defenses~\cite{sandvig2007robustness,mehta2007robust,tang2019adversarial,chen2019adversarial,yuan2019adversarial,liu2020certifiable,hidano2020recommender} try to train more robust recommender systems. For instance, adversarial training~\cite{goodfellow2015explaining}, which was developed to train robust machine learning classifiers, has been extended to train robust recommender systems by multiple work~\cite{tang2019adversarial,yuan2019adversarial}. However, none of the above defenses provides provable robustness guarantees. In particular, they cannot derive certified Precision@$N$, certified Recall@$N$, and certified F1-Score@$N$. As a result, they can still be attacked by adaptive attacks.} 

\myparatight{Ensemble recommender systems}
Ensemble methods have been explored to improve the empirical performance of  recommender systems~\cite{bell2008bellkor,toscher2009bigchaos,wu2007collaborative}. For instance, around a decade ago, the winning teams~\cite{bell2008bellkor,toscher2009bigchaos} in the well-known Netflix competition on 
predicting user rating scores for movies  blended multiple base recommender systems built by different base algorithms. 
However,  these studies are different from ours. The key difference is that they didn't derive the provable robustness guarantees for their ensemble recommender systems. 
Moreover,  they use different ways to aggregate the base recommender systems, e.g., they aggregated the rating scores predicted by the base recommender systems.  

\myparatight{Provably robust classifiers against data poisoning attacks} Several works~\cite{jia2020intrinsic,jia2022certified} proposed  certified defenses against data poisoning attacks for machine learning classifiers. The key difference between classifiers and recommender systems is that an input has a single ground-truth label in a classifier while a user has multiple ground-truth items in a recommender system. As a result, these methods achieve sub-optimal certified robustness guarantees when generalized to recommender systems as shown in our experimental results.

%% file: discussion.tex
\section{\CR{Discussion and Limitations}}
\label{discussion_limitation}

\vspace{-2mm}
\myparatight{\neil{Theoretical guarantees}} \neil{For any given number of fake users, our method can derive a certified intersection size $r$ for each user. Note that our theoretical guarantee still holds even if fake users rate new items.  However, when the fraction of fake users is large (e.g., 49\%), the derived $r$ and the corresponding certified Precision@$N$/Recall@$N$/F1-Score@$N$ may reduce to 0. As the first step on provably robust recommender systems, our method can derive a non-zero $r$ against a moderate number of fake users. It is still an open challenge to derive a non-trivial $r$ for a large fraction of fake users. Essentially, there is a trade-off between performance without attack and robustness, which is controlled by  $s$ (number of users in each submatrix). When the fraction of fake users is large, using a small $s$ makes the ensemble recommender system more robust but less accurate.

There are several directions to further improve the theoretical guarantees, i.e., derive a larger $r$ for a given fraction of fake users. First,  we considered a very strong threat model, where each fake user can arbitrarily rate all items. For instance, in MovieLens-100k, each fake user can rate up to 1,682 items, which accounts for 1.7\% of the rating scores from all genuine users. Fake users that rate a large number of items can be easily detected, as genuine users often rate a small number of items. Therefore, one way to further improve theoretical guarantee is to consider fake users that rate a bounded number of items. Second,  PORE is applicable to any base recommender system algorithm without considering the knowledge of the base algorithm. Therefore, the second possible way to derive better theoretical guarantees is to consider the domain knowledge of a specific base algorithm. }

\myparatight{\neil{Base algorithms and voting mechanisms}} \neil{We focus on using the same base algorithm to train each base recommender system in this work. We note that  the base recommender systems can be trained using different base algorithms. In particular, our theoretical guarantee holds for any (randomized) base algorithm. Therefore, given a set of base algorithms, we can randomly pick one to train a base recommender system. Moreover, we can view each base recommender system is trained using a randomized base algorithm sampled from the set of base algorithms. PORE uses hard voting when aggregating the items recommended by the base recommender systems. Hard voting has to be used to derive the theoretical guarantee of the ensemble recommender system.}
 
\myparatight{\neil{Targeted data poisoning attacks}} \neil{In this work, we focus on untargeted data poisoning attacks, which aim to reduce the overall performance of a recommender system. Our method can guarantee a lower bound of recommendation performance against any untargeted data poisoning attacks. Targeted data poisoning attacks aim to promote specific attacker-chosen items (called \emph{target items})~\cite{yang2017fake}. It is an interesting future work to derive provable robustness guarantees against such attacks. Specifically, given a fraction of fake users, we aim to derive an upper bound of the number of genuine users, to which the target items are recommended. }

%% file: conclusion.tex
\section{Conclusion and Future Work}
\label{conclusion}
In this work, we show that  PORE can turn an arbitrary base recommender system algorithm to be provably robust against data poisoning attacks via ensembling multiple base recommender systems built by the base algorithm on random subsamples of the rating-score matrix. 
Our ensemble recommender system guarantees a certain fraction of its top-$N$ items recommended to a user \neil{is} unaffected by fake users no matter how the attacker crafts their rating scores. Our empirical evaluation confirms that our ensemble recommender system provides provable robustness guarantees. Interesting future work includes deriving better provable robustness guarantees by bounding the number of items that fake users can rate and incorporating the knowledge of the base recommender system algorithm as well as deriving provable robustness guarantees against targeted data poisoning attacks.

\section*{Acknowledgements}
We thank the anonymous reviewers and shepherd for their constructive comments.  This work was supported by NSF under grant No. 2131859, 2112562, 2131859, 1937786, and 1937787, as well as ARO grant No. W911NF2110182.

%% file: appendix.tex
\appendix

\section{Proof of Theorem 3.1}
\label{proof_of_certified_theorem}
Recall that, given a rating-score matrix $\mathbf{M}$ and its poisoned version $\mathbf{M}^{\prime}$, $\mathbf{X}$ and $\mathbf{Y}$ are respectively two submatrices with $s$ rows randomly sampled from $\mathbf{M}$ and $\mathbf{M}^{\prime}$ without replacement. \neil{We use $\Phi$ to denote the domain space of $\mathbf{Y}$, i.e., each element in $\Phi$ is a submatrix with $s$ rows sampled from $\mathbf{M}'$. Note that since $\mathbf{M}$ is a submatrix of $\mathbf{M}'$, the domain space of $\mathbf{X}$ is a subset of $\Phi$.}  For simplicity, we define the following notations. Suppose we have $\mathbf{Z} \in \Phi$ and another rating-score matrix $\mathbf{W}$, we say $\mathbf{Z}  \prec \mathbf{W}$ (or $\mathbf{Z}  \nprec \mathbf{W}$) if $\mathbf{Z}$ is (or is not) in the domain space created by sampling $s$ rows from $\mathbf{W}$.  
We have $\mathbf{X} \prec \mathbf{M}$ and $\mathbf{Y} \prec \mathbf{M}'$ based on our defined notations. Similarly, give a user $u$ and $\mathbf{Z}\in \Phi$, we say $u \vdash \mathbf{Z}$ if $\mathbf{Z}$ contains rating scores of user $u$. We say  $u \nvdash \mathbf{Z}$ if $\mathbf{Z}$ does not contain user $u$'s rating scores.

The following lemma generalizes the Neyman-Pearson Lemma~\cite{neyman1933ix} to multiple functions:
\begin{lemma}
\label{lemma_np_general}
Let $\mathbf{X}$, $\mathbf{Y}$ be two random variables with probability densities $\text{Pr}(\mathbf{X}=\mathbf{Z})$ and $\text{Pr}(\mathbf{Y}=\mathbf{Z})$, where $\mathbf{Z} \in \Phi$. Let $g_{1}, g_{2}, \cdots, g_{\gamma}:\Phi \xrightarrow{} \{0,1\}$ be $\gamma$ random or deterministic functions. Let $\eta$ be an integer such that: $\sum_{l=1}^{\gamma}g_{l}(1|\mathbf{Z})\leq \eta, \forall \mathbf{Z}\in \Phi,$
where $g_{l}(1|\mathbf{Z})$ denotes the probability that $g_{l}(\mathbf{Z})=1$. Similarly, we use $g_{l}(0|\mathbf{Z})$ to denote the probability that $g_{l}(\mathbf{Z})=0$. Then, we have the following:  

(1) If $\Phi^{\prime} \subseteq \{\mathbf{Z}\in \Phi:g_{1}(1|\mathbf{Z})=g_{2}(1|\mathbf{Z})=\cdots=g_{\gamma}(1|\mathbf{Z})=0 \}$, $O_1=\{\mathbf{Z}\in \Phi \setminus \Phi^{\prime}: \text{Pr}(\mathbf{Y}=\mathbf{Z})<\rho \cdot \text{Pr}(\mathbf{X}=\mathbf{Z})  \}$ and $O_2=\{\mathbf{Z}\in \Phi\setminus \Phi^{\prime}: \text{Pr}(\mathbf{Y}=\mathbf{Z})=\rho \cdot \text{Pr}(\mathbf{X}=\mathbf{Z})   \}$ for some $\rho > 0$. Assuming we have $O_3 \subseteq O_2$, and $O=O_1 \cup O_3$. Then, if we have $\frac{\sum_{l=1}^{\gamma}\text{Pr}(g_{l}(\mathbf{X})=1)}{\eta}\geq \text{Pr}(\mathbf{X}\in O)$, then $\frac{\sum_{l=1}^{\gamma}\text{Pr}(g_{l}(\mathbf{Y})=1)}{\eta}\geq \text{Pr}(\mathbf{Y}\in O)$. 

(2) If $\Phi^{\prime} \subseteq \{\mathbf{Z}\in \Phi:g_{1}(1|\mathbf{Z})=g_{2}(1|\mathbf{Z})=\cdots=g_{\gamma}(1|\mathbf{Z})=0 \}$,  $O_1=\{\mathbf{Z}\in \Phi: \text{Pr}(\mathbf{Y}=\mathbf{Z})> \rho \cdot \text{Pr}(\mathbf{X}=\mathbf{Z})  \}$ and $O_2=\{\mathbf{Z}\in \Phi: \text{Pr}(\mathbf{Y}=\mathbf{Z})=\rho \cdot \text{Pr}(\mathbf{X}=\mathbf{Z}) \}$ for some $\rho > 0$. Assuming we have $O_3 \subseteq O_2$, and $O=O_1 \cup O_3 $. Then, if we have $\frac{\sum_{l=1}^{\gamma}\text{Pr}(g_{l}(\mathbf{X})=1)}{\eta}\leq \text{Pr}(\mathbf{X}\in O)$, then $\frac{\sum_{l=1}^{\gamma}\text{Pr}(g_{l}(\mathbf{Y})=1)}{\eta}\leq \text{Pr}(\mathbf{Y}\in O)$. 
\end{lemma}
\begin{proof}
We first prove part (1). For convenience, we denote the complement of $O$ as $O^{c}$. Then, we have the following: 
\allowdisplaybreaks
{\small 
\begin{align}
  & \frac{\sum_{l=1}^{\gamma}\text{Pr}(g_{l}(\mathbf{Y})=1)}{\eta}- \text{Pr}(\mathbf{Y}\in O) \\
=&\int_{\Phi}\frac{\sum_{l=1}^{\gamma}g_{l}(1|\mathbf{Z})}{\eta}\cdot \text{Pr}(\mathbf{Y}=\mathbf{Z})d\mathbf{Z} - \int_{O}\text{Pr}(\mathbf{Y}=\mathbf{Z})d\mathbf{Z} \\
=&\int_{O^{c}}\frac{\sum_{l=1}^{\gamma}g_{l}(1|\mathbf{Z})}{\eta}\cdot \text{Pr}(\mathbf{Y}=\mathbf{Z})d\mathbf{Z} \nonumber \\
& + \int_{O}\frac{\sum_{l=1}^{\gamma}g_{l}(1|\mathbf{Z})}{\eta}\cdot \text{Pr}(\mathbf{Y}=\mathbf{Z})d\mathbf{Z} - \int_{O}\text{Pr}(\mathbf{Y}=\mathbf{Z})d\mathbf{Z} \\
=&\int_{O^{c}}\frac{\sum_{l=1}^{\gamma}g_{l}(1|\mathbf{Z})}{\eta}\cdot \text{Pr}(\mathbf{Y}=\mathbf{Z})d\mathbf{Z} \nonumber\\
&- \int_{O}(1-\frac{\sum_{l=1}^{\gamma}g_{l}(1|\mathbf{Z})}{\eta})\cdot\text{Pr}(\mathbf{Y}=\mathbf{Z})d\mathbf{Z} \\
=&\int_{O^{c}\setminus \Phi^{\prime}}\frac{\sum_{l=1}^{\gamma}g_{l}(1|\mathbf{Z})}{\eta}\cdot \text{Pr}(\mathbf{Y}=\mathbf{Z})d\mathbf{Z} + \int_{ \Phi^{\prime}}\frac{\sum_{l=1}^{\gamma}g_{l}(1|\mathbf{Z})}{\eta} \nonumber \\
\label{lemma_np_general_e1}
& \cdot \text{Pr}(\mathbf{Y}=\mathbf{Z})d\mathbf{Z} 
- \int_{O}(1-\frac{\sum_{l=1}^{\gamma}g_{l}(1|\mathbf{Z})}{\eta})\cdot\text{Pr}(\mathbf{Y}=\mathbf{Z})d\mathbf{Z} \\
\geq&\rho \cdot[\int_{O^{c}\setminus  \Phi^{\prime}}\frac{\sum_{l=1}^{\gamma}g_{l}(1|\mathbf{Z})}{\eta}\cdot \text{Pr}(\mathbf{X}=\mathbf{Z})d\mathbf{Z} + \int_{ \Phi^{\prime}}\frac{\sum_{l=1}^{\gamma}g_{l}(1|\mathbf{Z})}{\eta}\nonumber \\
\label{lemma_np_general_e2}
&\cdot \text{Pr}(\mathbf{X}=\mathbf{Z})d\mathbf{Z}
- \int_{O}(1-\frac{\sum_{l=1}^{\gamma}g_{l}(1|\mathbf{Z})}{\eta})\cdot\text{Pr}(\mathbf{X}=\mathbf{Z})d\mathbf{Z}] \\
=&\rho \cdot[\int_{O^{c}}\frac{\sum_{l=1}^{\gamma}g_{l}(1|\mathbf{Z})}{\eta}\cdot \text{Pr}(\mathbf{X}=\mathbf{Z})d\mathbf{Z} \nonumber \\
& + \int_{O}\frac{\sum_{l=1}^{\gamma}g_{l}(1|\mathbf{Z})}{\eta}\cdot \text{Pr}(\mathbf{X}=\mathbf{Z})d\mathbf{Z} - \int_{O}\text{Pr}(\mathbf{X}=\mathbf{Z})d\mathbf{Z}] \\
=&\rho \cdot[\int_{\Phi}\frac{\sum_{l=1}^{\gamma}g_{l}(1|\mathbf{Z})}{\eta}\cdot \text{Pr}(\mathbf{X}=\mathbf{Z})d\mathbf{Z} - \int_{O}\text{Pr}(\mathbf{X}=\mathbf{Z})d\mathbf{Z}] \\
=&\rho \cdot [\frac{\sum_{l=1}^{\gamma}\text{Pr}(g_{l}(\mathbf{X})=1)}{\eta}- \text{Pr}(\mathbf{X}\in O)] \\
\geq & 0.
\end{align}
}
\begin{figure}[!t]
\centering
{\includegraphics[width=0.28\textwidth]{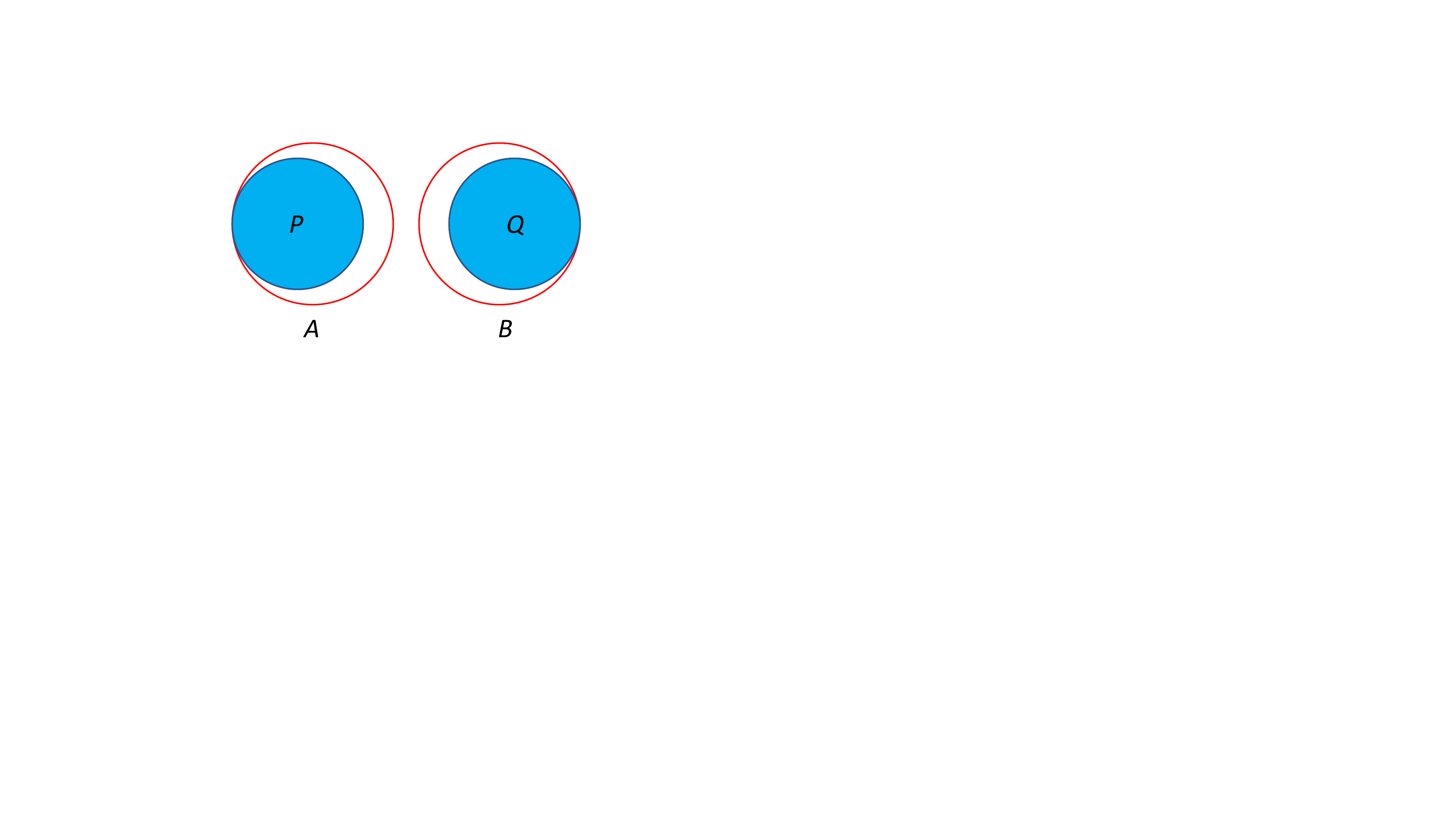}}
\caption{Illustration of subset $P$, $Q$, $A$, and $B$.}
\vspace{-4mm}
\label{illustration_figure}
\end{figure}
We have Equation~(\ref{lemma_np_general_e2}) from~(\ref{lemma_np_general_e1}) due to the fact that $\forall \mathbf{Z}\in O^{c}\setminus \Phi^{\prime}, \text{Pr}(\mathbf{Y}=\mathbf{Z} \geq \rho \cdot \text{Pr}(\mathbf{X}=\mathbf{Z})$, $\forall \mathbf{Z} \in \Phi^{\prime}, g_{1}(1|\mathbf{Z})=g_{2}(1|\mathbf{Z})=\cdots=g_{\gamma}(1|\mathbf{Z})=0$,  $1-\frac{\sum_{l=1}^{\gamma}g_{l}(1|\mathbf{Z})}{\eta} \geq 0$, and $\text{Pr}(\mathbf{Y}=\mathbf{Z}) \leq \rho \cdot \text{Pr}(\mathbf{X}=\mathbf{Z}), \forall \mathbf{Z}\in O$.
Similarly, we can prove the part (2). We omit the details for simplicity. 
\end{proof}
Given a user $u$ and $\Phi$, we define the following subsets: 
{\small 
\begin{align}
\label{definition_of_q_app}
&  P = \{\mathbf{Z}\in \Phi|\mathbf{Z}\prec \mathbf{M},u \vdash \mathbf{Z}\} ,   Q = \{\mathbf{Z}\in \Phi|\mathbf{Z}\prec \mathbf{M},u \nvdash \mathbf{Z}\} \\
\label{definition_of_b_app}
 &     A = \{\mathbf{Z}\in \Phi|\mathbf{Z}\prec \mathbf{M}',u \vdash \mathbf{Z}\},    B = \{\mathbf{Z}\in \Phi|\mathbf{Z}\prec \mathbf{M}',u \nvdash \mathbf{Z}\}
 \end{align}
 }
It is easy to verify that $P \cap Q = \emptyset$ and $A \cap B = \emptyset$.  Moreover, since $M'$ contains rating scores for all the user in $\mathbf{M}$, we have $P \subseteq A$ and $Q \subseteq B$. Figure~\ref{illustration_figure} shows an illustration of these four subsets.
Based on our sampling without replacement, we have the following probability mass functions for $\mathbf{X}$ and $\mathbf{Y}$ in $\Phi$:
{\small 
\begin{align}
\label{probability_density_start_app}
&\text{Pr}(\mathbf{X} = \mathbf{Z}) = 
\begin{cases}
 \frac{1}{{n \choose s}}, &\text{ if } \mathbf{Z} \in P  \cup Q, \\
 0, &\text{ otherwise}.
\end{cases} \\
\label{probability_density_end_app}
&\text{Pr}(\mathbf{Y} = \mathbf{Z} ) = 
\begin{cases}
 \frac{1}{{n' \choose s}}, &\text{ if } \mathbf{Z} \in A  \cup B, \\
 0, &\text{ otherwise}.
\end{cases} 
\end{align}
}
Given the probability mass functions, we have the following probabilities: 
\allowdisplaybreaks
{\small 
\begin{align}
\label{probability_of_x_in_region_start_app}
 &   \text{Pr}(\mathbf{X}\in P)= 1 - {n-1 \choose s}\cdot  \frac{1}{{n \choose s}} =  \frac{s}{n}, 
 & \text{Pr}(\mathbf{X}\in A) =  \frac{s}{n}, \\
  &   \text{Pr}(\mathbf{Y}\in P)  =  \frac{s}{n} \cdot \frac{{n \choose s}}{{n' \choose s}}, 
   \label{probability_of_x_in_region_end_app}
 & \text{Pr}(\mathbf{Y}\in A) = \frac{s}{n'}.
\end{align}
}
We have $\text{Pr}(\mathbf{X}\notin P)={n-1 \choose s}\cdot  \frac{1}{{n \choose s}}$ because there are overall ${n-1 \choose s}$ possible submatrices if the user $u$ is not sampled from $\mathbf{M}$. Then, we can compute $ \text{Pr}(\mathbf{X}\in P)$ based on the fact that $ \text{Pr}(\mathbf{X}\in P)+\text{Pr}(\mathbf{X}\notin P)=1$. We obtain $\text{Pr}(\mathbf{X}\in A)$ based on the fact that $P \subseteq A$ and $A \cap B = \emptyset$. Similarly, we can compute the probability of $\mathbf{Y}$ in these subsets. 

We will leverage the law of contraposition to prove our theorem. Suppose we have a statement: $U \longrightarrow V$, then, its contraposition is $\neg V \longrightarrow \neg U$. The law of contraposition tells us that a statement is true if and only if its contraposition is true. We define the following two predicates: 
{\small 
\begin{align}
 U: &  ~ \underline{p^{*}_{\mu_{r}}} >   \min ( \min_{c=1}^{N-r +1}\frac{ N' \cdot(\overline{p}^{*}_{\mathcal{H}_{c}} + \sigma)}{c},  \overline{p}^{*}_{v_1}+ \sigma)\\
&\text{ and } 1 \leq r \leq \min(k, N), \nonumber \\
V:&~  |\mathcal{I}_{u} \cap \mathcal{T}(\mathbf{M}',u)| \geq r ,
\end{align}
}
where $\sigma = \frac{s}{n'} \cdot \frac{{n' \choose s}}{{n \choose s}} -\frac{s}{n}$ . 

\myparatight{Deriving the necessary condition} We assume $\neg V$ is true, i.e.,  $|\mathcal{I}_{u} \cap \mathcal{T}(\mathbf{M}',u)| < r$. If $r=0 $ or $ r > \min(k, N)$, then, $\neg U$ is true. Next, we consider $1 \leq r \leq \min(k, N)$ and we will show $\underline{p^{*}_{\mu_{r}}} \leq   \min ( \min_{c=1}^{N-r +1}\frac{ N' \cdot(\overline{p}^{*}_{\mathcal{H}_{c}} + \sigma)}{c},  \overline{p}^{*}_{v_1}+ \sigma)$ is true. If $|\mathcal{I}_{u} \cap \mathcal{T}(\mathbf{M}',u)| < r$, then, there exist at least $k-r +1$ items in $\mathcal{I}_{u}$ that are not recommended to the user $u$ by our ensemble recommender system when taking $\mathbf{M}'$ as input. In other words, there exist at least $N - r +1$ items in $\mathcal{I} \setminus \mathcal{I}_u$ appears in the recommended item set $\mathcal{T}(\mathbf{M}',u)$. For simplicity, we use $\mathcal{D}_{r}$ and $\mathcal{V}_{r}$ to denote subsets of $k - r +1$ items in $\mathcal{I}_{u}$ and $N - r + 1$ items in $\mathcal{I}\setminus\mathcal{I}_{u}$, respectively. Formally, we have the following: 
\begin{align}
    \exists \mathcal{D}_{r}, \mathcal{V}_{r}, \text{ s.t. } \mathcal{D}_{r}\cap \mathcal{T}(\mathbf{M}',u) =\emptyset, \mathcal{V}_{r} \subseteq \mathcal{T}(\mathbf{M}',u), 
\end{align}
The above equation means that the poisoned item probabilities of the items in $\mathcal{V}_{r}$ are no smaller than the poisoned item probabilities of the items in $\mathcal{D}_{r}$. In other words, we have the following for $\mathcal{D}_{r}$ and $\mathcal{V}_{r}$:
{\small 
\begin{align}
    \max_{i \in \mathcal{D}_{r}}\text{Pr}(i \in \mathcal{A}(\mathbf{Y},u)) \leq \min_{j \in \mathcal{V}_{r}}\text{Pr}(j \in \mathcal{A}(\mathbf{Y},u)).
\end{align}
}
Moreover, since there exist $\mathcal{D}_{r}$ and $\mathcal{V}_{r}$, we have the following necessary condition if $|\mathcal{I}_{u} \cap \mathcal{T}(\mathbf{M}',u)| < r$ and $1 \leq r \leq \min(k, N)$: 
{\small 
\begin{align}
\label{necessary_condition_of_nega_u_is_true_app}
    \min_{\mathcal{D}_{r}}\max_{i \in \mathcal{D}_{r}}\text{Pr}(i \in \mathcal{A}(\mathbf{Y},u)) \leq \max_{\mathcal{V}_{r}}\min_{j \in \mathcal{V}_{r}}\text{Pr}(j \in \mathcal{A}(\mathbf{Y},u)).
\end{align}
}
Next, we will derive the lower bound of left-hand side and the upper bound of the right-hand side for the above equation. 

\myparatight{Deriving a lower bound of $ \min_{\mathcal{D}_{r}}\max_{i \in \mathcal{D}_{r}}\text{Pr}(i \in \mathcal{A}(\mathbf{Y},u))$} For $\forall i \in \mathcal{D}_{r}$, based on Equation~(\ref{probability_upper_low_bound_theorem_1}), we have the following: 
{\small 
\begin{align}
\label{probability_upper_low_bound_theorem_1_app}
&\underline{p^{*}_i} \triangleq \frac{\lfloor \underline{p_i} \cdot {n \choose s}\rfloor}{{n \choose s}} \leq \underline{p_i} \leq \text{Pr}(i \in \mathcal{A}(\mathbf{X},u))
\end{align}
}
Based on the definition of $\mathbf{X}$, we have: 
{\small 
\begin{align}
\label{probability_upper_bound_proof}
    \text{Pr}(i \in \mathcal{A}(\mathbf{X},u)) \geq \underline{p^{*}_i}.
\end{align}
}
For $\forall i \in \mathcal{D}_{r}$, we define the function $g_{i}(\mathbf{Z})= \mathbb{I}(i \in \mathcal{A}(\mathbf{Z},u))$, where $\mathbb{I}$ is an indicator function. Specifically, we have $g_{i}(\mathbf{Z})=0$ for $\forall \mathbf{Z}\in B$ based on the definition of $B$ in Equation~(\ref{definition_of_b_app}), which will be used when we leverage Lemma~\ref{lemma_np_general} to derive the lower bound of $\text{Pr}(g_{i}(\mathbf{Y})=1)$. We have $\text{Pr}(g_{i}(\mathbf{X})=1) \geq \underline{p^{*}_i}$ based on Equation~(\ref{probability_upper_bound_proof}) and the definition of $g_i$.  For $\forall i \in \mathcal{D}_{r}$, we can find $C_{i} \subseteq P$ such that we have the following:
{\small 
\begin{align}
\label{construct_of_region_A_i}
    \text{Pr}(\mathbf{X} \in C_{i}) = \underline{p^{*}_i}. 
\end{align}
}
Note that we can find such a subset because $\underline{p^{*}_i}$ is an integer multiple of $\frac{1}{{n \choose s}}$. Then, we have the following:
{\small 
\begin{align}
\label{condition_of_np_lemma_i_u_of_x_app}
 \text{Pr}(g_{i}(\mathbf{X})=1) \geq   \text{Pr}(\mathbf{X} \in C_{i}).
\end{align}
}
For simplicity, we define the following quantity: 
{\small 
\begin{align}
    \tau = {n \choose s}/{n' \choose s}. 
\end{align}
}
Next, we will apply Lemma~\ref{lemma_np_general} to obtain the lower bound of $\text{Pr}(g_{i}(\mathbf{Y})=1)$. In particular, we let $\Phi^{\prime} =  B$ since we have $g_{i}(\mathbf{Z})=0$ for $\forall \mathbf{Z} \in B$. Note that we omit $Q$ since we have $Q \subseteq B$. Then, we have $A = \Phi \setminus \Phi^{\prime}$.  Based on Equation~(\ref{probability_density_start_app}) -~(\ref{probability_density_end_app}), we have $\text{Pr}(\mathbf{Y}=\mathbf{Z})= \tau \cdot \text{Pr}(\mathbf{X}=\mathbf{\mathbf{Z}})$ if $\mathbf{Z} \in P$ and $\text{Pr}(\mathbf{Y}=\mathbf{Z}) > \tau \cdot \text{Pr}(\mathbf{X}=\mathbf{\mathbf{Z}})$ if $\mathbf{Z} \in A \setminus P$.  We let $O_1 = \emptyset$ since there is no subset that satisfies $\text{Pr}(\mathbf{Y}=\mathbf{Z}) < \tau \cdot \text{Pr}(\mathbf{X}=\mathbf{\mathbf{Z}})$. Furthermore, we let $O_2 = P $, $O_3 = C_{i} \subseteq O_2$, $\eta = 1$, and $\gamma=1$. Finally, we let $O = O_1 \cup O_3 = C_i$. We can apply Lemma~\ref{lemma_np_general} based on the condition in  Equation~(\ref{condition_of_np_lemma_i_u_of_x_app}) and we have the following: 
\begin{align}
    \text{Pr}(g_{i}(\mathbf{Y})=1) \geq \text{Pr}(\mathbf{Y} \in C_i).
\end{align}
Based on the definition of $g_{i}$, we have the following: 
{\small 
\begin{align}
& \text{Pr}(i \in \mathcal{A}(\mathbf{Y},u)) \\
=& \text{Pr}(g_{i}(\mathbf{Y})=1) \\
\geq & \text{Pr}(\mathbf{Y} \in C_i) \\
= &   \text{Pr}(\mathbf{X} \in C_i)\cdot \tau \\
=& \underline{p^{*}_i} \cdot \tau.
\end{align}
}
For simplicity, we denote $\mathcal{D}_{r}=\{d^{\prime}_{1},d^{\prime}_{2}, \cdots, d^{\prime}_{z}\}$, where $z = k-r +1$. Without loss of generality, we assume $\underline{p_{d^{\prime}_{1}}} \geq \cdots \geq \underline{p_{d^{\prime}_{z}}}$. We have the following: 
\begin{align}
 \max_{i \in \mathcal{D}_{r}}\text{Pr}(i \in \mathcal{A}(\mathbf{Y},u)) 
\geq \max_{i \in \mathcal{D}_{r}} \underline{p^{*}_i} \cdot \tau
= \underline{p^{*}_{d^{\prime}_{1}}} \cdot \tau. 
\end{align}
Then, we have the following: 
\begin{align}
\min_{\mathcal{D}_{r}}\max_{i \in \mathcal{D}_{r}}\text{Pr}(i \in \mathcal{A}(\mathbf{Y},u)) 
=  \min_{\mathcal{D}_{r}} \underline{p^{*}_{d^{\prime}_{1}}} \cdot \tau
\end{align}
Therefore, when $\mathcal{D}_{r}=\{\mu_{r},\mu_{r+1},\cdots, \mu_{k}\}$, $\underline{p^{*}_{d^{\prime}_{1}}} \cdot \tau$ reaches the minimal value which is $\underline{p^{*}_{\mu_{r}}}\cdot \tau$. In other words, we have the following: 
\begin{align}
    \min_{\mathcal{D}_{r}}\max_{i \in \mathcal{D}_{r}}\text{Pr}(i \in \mathcal{A}(\mathbf{Y},u))  \geq \underline{p^{*}_{\mu_{r}}}\cdot \tau. 
\end{align}

\myparatight{Deriving an upper bound of $\max_{\mathcal{V}_{r}}\min_{j \in \mathcal{V}_{r}}\text{Pr}(j \in \mathcal{A}(\mathbf{Y},u))$} For $\forall j \in \mathcal{V}_r$, given Equation~(\ref{probability_upper_low_bound_theorem_1}), we have the following: 
{\small 
\begin{align}
\label{probability_upper_low_bound_theorem_1_app_low}
&\overline{p}^{*}_j = \frac{\lceil \overline{p}_j \cdot {n \choose s}\rceil}{{n \choose s}}  \geq \overline{p}_j \geq \text{Pr}(j \in \mathcal{A}(\mathbf{X},u)).
\end{align}
}
Suppose we have  $\mathcal{V}_{r}=\{v^{\prime}_1, v^{\prime}_2, \cdots, v^{\prime}_{w}\}$, where $w= N - r +1$. Without loss of generality, we assume the following:
\begin{align}
\label{without_loss_condition_2_app}
\overline{p}_{v^{\prime}_1} \leq \overline{p}_{v^{\prime}_2}\leq \cdots \leq \overline{p}_{v^{\prime}_w}.  
\end{align}
We first derive an upper bound of $\min_{j \in \mathcal{V}_{r}}$ $\text{Pr}(j \in \mathcal{A}(\mathbf{Y},u))$. Given an arbitrary item $j \in \mathcal{V}_{r}$, we have the following inequality based on Equation~(\ref{probability_upper_low_bound_theorem_1_app_low}) and our definition of $\mathbf{X}$:
\begin{align}
\label{probability_lower_bound_proof}
    \text{Pr}(j \in \mathcal{A}(\mathbf{X},u)) \leq \overline{p}^{*}_j.
\end{align}
Given an item $j \in \mathcal{V}_r$, we define the function $g_{j}(\mathbf{Z}) = \mathbb{I}(j \in \mathcal{A}(\mathbf{Z},u))$. We have $\text{Pr}(g_{j}(\mathbf{X}) = 1) \leq \overline{p}^{*}_j$ based on Equation~(\ref{probability_lower_bound_proof}) and the definition of $g_j$. Then, we can leverage Lemma~\ref{lemma_np_general} to derive an upper bound for $\text{Pr}(g_{j}(\mathbf{X}) = 1)$. In particular, we can find $C^{\prime}_j \subseteq P$ such that we have the following:
\begin{align}
    \text{Pr}(\mathbf{X} \in C^{\prime}_j) = \overline{p}^{*}_j.
\end{align}
We let $C_{j} = C^{\prime}_j \cup (A \setminus P)$. Since we have $\text{Pr}(\mathbf{X} \in A\setminus P)= \text{Pr}(\mathbf{X} \in A)-\text{Pr}(\mathbf{X} \in P)=0$, we have the following:
\begin{align}
    \text{Pr}(\mathbf{X} \in C_j) = \text{Pr}(\mathbf{X} \in C^{\prime}_j) + \text{Pr}(\mathbf{X} \in A\setminus P) = \overline{p}^{*}_j.
\end{align}
Based on $\text{Pr}(g_{j}(\mathbf{X}) = 1) \leq \overline{p}^{*}_j$, we have the following:
\begin{align}
\label{condition_of_np_lemma_j_u_of_x_app}
    \text{Pr}(g_{j}(\mathbf{X})=1) \leq \text{Pr}(\mathbf{X}\in C_j).
\end{align}
Next, we will apply Lemma~\ref{lemma_np_general} to obtain the upper bound of $\text{Pr}(g_{j}(Y)=1)$. In particular, we let $\Phi^{\prime} =  B$ since we have $g_{j}(\mathbf{Z})=0$ for $\forall \mathbf{Z} \in B$. Then, we have $A = \Phi \setminus \Phi^{\prime}$.  Based on Equation~(\ref{probability_density_start_app}) -~(\ref{probability_density_end_app}), we have $\text{Pr}(\mathbf{Y}=\mathbf{Z})= \tau \cdot \text{Pr}(\mathbf{X}=\mathbf{\mathbf{Z}})$ if $\mathbf{Z} \in P$ and $\text{Pr}(\mathbf{Y}=\mathbf{Z}) > \tau \cdot \text{Pr}(\mathbf{X}=\mathbf{\mathbf{Z}})$ if $\mathbf{Z} \in A \setminus P$.  We let $O_1 = A \setminus P$, $O_2 = P $, $O_3 = C^{\prime}_{j} \subseteq O_2$, $\eta=1$, and $\gamma=1$. Finally, we let $O = O_1 \cup O_3 = C_j$. We can apply Lemma~\ref{lemma_np_general} based on the condition in  Equation~(\ref{condition_of_np_lemma_j_u_of_x_app}) and we have the following: 
{\small 
\begin{align}
&\text{Pr}(j \in \mathcal{A}(\mathbf{Y},u)) \\
    =&\text{Pr}(g_{j}(\mathbf{Y})=1) \\
    \leq & \text{Pr}(\mathbf{Y} \in C_j) \\
    =& \text{Pr}(\mathbf{Y} \in C^{\prime}_j) + \text{Pr}(\mathbf{Y} \in A\setminus P) \\
    =& \text{Pr}(\mathbf{X} \in C^{\prime}_j)\cdot \tau + \text{Pr}(\mathbf{Y} \in A) - \text{Pr}(\mathbf{Y} \in P) \\
    =& \overline{p}^{*}_j \cdot \frac{{n \choose s}}{{n' \choose s}}+ \frac{s}{n'}-\frac{s}{n}\cdot \frac{{n \choose s}}{{n' \choose s}}.
\end{align}
}
Given Equation~(\ref{without_loss_condition_2_app}), we have the following:
{\small 
\begin{align}
\min_{j \in \mathcal{V}_{r}}\text{Pr}(j \in \mathcal{A}(\mathbf{Y},u)) \leq \overline{p}^{*}_{v^{\prime}_{1}} \cdot \frac{{n \choose s}}{{n' \choose s}}+ \frac{s}{n'}-\frac{s}{n}\cdot \frac{{n \choose s}}{{n' \choose s}}.
\end{align}
}
Therefore, when $\mathcal{V}_{r}$ contains the set of items in $\mathcal{I} \setminus \mathcal{I}_{u}$ that have the largest probability bounds, which we denote as $\mathcal{V}_{r}=\{v_1, v_2, \cdots, v_w\}$ where $\overline{p}_{v_1} \leq \overline{p}_{v_2}\leq \cdots \leq \overline{p}_{v}$, the upper bound of $\max_{\mathcal{V}_{r}}\min_{j \in \mathcal{V}_{r}}$ $ \text{Pr}(j \in \mathcal{A}(\mathbf{Y},u))$ reaches the maximum value. Formally, we have the following: 
{\small 
\begin{align}
\max_{\mathcal{V}_{r}}\min_{j \in \mathcal{V}_{r}}\text{Pr}(j \in \mathcal{A}(\mathbf{Y},u))  \leq \overline{p}^{*}_{v_{1}} \cdot \frac{{n \choose s}}{{n' \choose s}}+ \frac{s}{n'}-\frac{s}{n}\cdot \frac{{n \choose s}}{{n' \choose s}}. 
\end{align}
}

Next, we will derive another upper bound for $\text{Pr}(g_{j}(Y)=1)$ via jointly considering multiple items. In particular, we use $\mathcal{H}_{c}$ to denote an arbitrary subset of $\mathcal{V}_{r}$ that contains $c$ items, i.e., $\mathcal{H}_{c} \subseteq \mathcal{V}_{r}$. We denote $\overline{p}_{\mathcal{H}_{c}} = \sum_{j \in \mathcal{H}_{c}} \overline{p}_j$, i.e., the summation of probability upper bounds for the items in $\mathcal{H}_{c}$. Then, for each item $j \in \mathcal{H}_{c}$, we define $g_{j}(\mathbf{Z}) = \mathbb{I}(j \in \mathcal{A}(\mathbf{Z},u))$. Given a positive integer $\eta$, we define the following quantity: 
{\small 
\begin{align}
    \overline{p}^{*}_{\mathcal{H}_{c}} = \frac{\lceil (\overline{p}_{\mathcal{H}_{c}}/\eta)\cdot {n \choose s} \rceil}{ {n \choose s}  }.
\end{align}
}
We can find $C^{\prime}_{\mathcal{H}_{c}} \subseteq P$ such that we have the following: 
{\small 
\begin{align}
    \text{Pr}(\mathbf{X}\in C^{\prime}_{\mathcal{H}_{c}}) = \overline{p}^{*}_{\mathcal{H}_{c}}. 
\end{align}
}
Then, we define $C_{\mathcal{H}_{c}} =C^{\prime}_{\mathcal{H}_{c}} \cup (A \setminus P)$ and we have the following: 
{\small 
\begin{align}
    \text{Pr}(\mathbf{X} \in C_{\mathcal{H}_{c}}) = \text{Pr}(\mathbf{X} \in C^{\prime}_{\mathcal{H}_{c}}) + \text{Pr}(\mathbf{X} \in A \setminus P) = \overline{p}^{*}_{\mathcal{H}_{c}}. 
\end{align}
}
Based on the definition of $g_{j}(\mathbf{Z})$, we have the following: 
{\small 
\begin{align}
\label{condition_of_np_lemma_multiple_j_u_of_x_app}
    \frac{\sum_{j \in \mathcal{H}_{c}} \text{Pr}(g_{j}(X)=1)}{\eta}\leq \frac{\sum_{j \in \mathcal{H}_{c}} \overline{p}_j}{\eta} \leq  \overline{p}^{*}_{\mathcal{H}_{c}} =  \text{Pr}(\mathbf{X} \in C_{\mathcal{H}_{c}}). 
\end{align}
}
Next, we will leverage Lemma~\ref{lemma_np_general} to derive an upper bound for $\sum_{j \in \mathcal{H}_{c}} \text{Pr}(g_{j}(\mathbf{Y})=1)$. Given a rating-score matrix $\mathbf{Z}$ as input, the recommender system algorithm $\mathcal{A}$ recommends $N'$ items to a user. Therefore, we have  $\sum_{j \in \mathcal{H}_{c}} \mathbb{I}(j \in \mathcal{A}(\mathbf{Z},u)) \leq N'$, i.e., $\sum_{j \in \mathcal{H}_{c}} g_{j}(\mathbf{Z}) \leq N'$. Based on this, we let $\eta = N'$. Since there are $c$ items in $\mathcal{H}_{c}$, we let $\gamma =  c$. We let $\Phi^{\prime} =  B$ since we have $g_{j}(\mathbf{Z})=0$ for $\forall \mathbf{Z} \in B$. Then, we have $A = \Phi \setminus \Phi^{\prime}$.  Based on Equation~(\ref{probability_density_start_app}) -~(\ref{probability_density_end_app}), we have $\text{Pr}(\mathbf{Y}=\mathbf{Z})= \tau \cdot \text{Pr}(\mathbf{X}=\mathbf{\mathbf{Z}})$ if $\mathbf{Z} \in P$ and $\text{Pr}(\mathbf{Y}=\mathbf{Z}) > \tau \cdot \text{Pr}(\mathbf{X}=\mathbf{\mathbf{Z}})$ if $\mathbf{Z} \in A \setminus P$.  We let $O_1 = A \setminus P$, $O_2 = P $, and $O_3 = C^{\prime}_{\mathcal{H}_{c}} \subseteq O_2$. Finally, we let $O = O_1 \cup O_3 = C_{\mathcal{H}_{c}}$. We can apply Lemma~\ref{lemma_np_general} based on the condition in  Equation~(\ref{condition_of_np_lemma_multiple_j_u_of_x_app}) and we have the following: 
{\small 
\begin{align}
&\sum_{j \in \mathcal{H}_{c}}\text{Pr}(j \in \mathcal{A}(\mathbf{Y},u)) \\
    =&\sum_{j \in \mathcal{H}_{c}}\text{Pr}(g_{j}(\mathbf{Y})=1) \\
    \leq & N' \cdot \text{Pr}(\mathbf{Y} \in C_{\mathcal{H}_{c}}) \\
    =& N' \cdot( \text{Pr}(\mathbf{Y} \in C^{\prime}_{\mathcal{H}_{c}}) + \text{Pr}(\mathbf{Y} \in A\setminus P)) \\
    =& N' \cdot(\text{Pr}(\mathbf{X} \in C^{\prime}_{\mathcal{H}_{c}})\cdot \tau + \text{Pr}(\mathbf{Y} \in A) - \text{Pr}(\mathbf{Y} \in P)) \\
    =& N' \cdot (\overline{p}^{*}_{\mathcal{H}_{c}} \cdot \frac{{n \choose s}}{{n' \choose s}}+  (\frac{s}{n'}-\frac{s}{n}\cdot \frac{{n \choose s}}{{n' \choose s}})).
\end{align}
}
Then, we have the following: 
{\small 
\begin{align}
\label{2_derive_min_app_0}
& \min_{j \in \mathcal{V}_{r}}\text{Pr}(j \in \mathcal{A}(\mathbf{Y},u)) \\
\label{2_derive_min_app_1}
\leq & \min_{j \in \mathcal{H}_{c}}\text{Pr}(j \in \mathcal{A}(\mathbf{Y},u)) \\
\label{2_derive_min_app_2}
\leq & \frac{\sum_{j \in \mathcal{H}_{c}}\text{Pr}(j \in \mathcal{A}(\mathbf{Y},u))}{c} \\
\label{2_derive_min_app_3}
\leq & N' \cdot (\overline{p}^{*}_{\mathcal{H}_{c}} \cdot \frac{{n \choose s}}{{n' \choose s}}+  (\frac{s}{n'}-\frac{s}{n}\cdot \frac{{n \choose s}}{{n' \choose s}}))/c.
\end{align}
}
We have Equation~(\ref{2_derive_min_app_1}) from~(\ref{2_derive_min_app_0}) because $\mathcal{H}_{c} \subseteq \mathcal{V}_{r}$ and Equation~(\ref{2_derive_min_app_2}) from~(\ref{2_derive_min_app_1}) because the smallest value is no larger than the average value in a set. We note that the upper bound of $ \min_{j \in \mathcal{V}_{r}}\text{Pr}(j \in \mathcal{A}(\mathbf{Y},u)) $ is non-decreasing as $\overline{p}_{\mathcal{H}_{c}}$ increases. Based on Equation~(\ref{without_loss_condition_2_app}), the upper bounds reaches the minimal value when $\mathcal{H}_{c} = \{v^{\prime}_1, v^{\prime}_2,$ $  \cdots, v^{\prime}_{c}\}$. Taking all possible $c$ into consideration, we have the following: 
{\small 
\begin{align}
   & \min_{j \in \mathcal{V}_{r}}\text{Pr}(j \in \mathcal{A}(\mathbf{Y},u)) \\ 
    \leq & \min_{c=1}^{N -r +1} N' \cdot (\overline{p}^{*}_{\mathcal{H}_{c}} \cdot \frac{{n \choose s}}{{n' \choose s}}+  (\frac{s}{n'}-\frac{s}{n}\cdot \frac{{n \choose s}}{{n' \choose s}}))/c, \nonumber
\end{align}
}
where  $\mathcal{H}_{c}=\{v^{\prime}_1, v^{\prime}_2, \cdots, v^{\prime}_{c}\}$. 
Similarly, the upper bound of $\max_{\mathcal{V}_{r}}$ $\min_{j \in \mathcal{V}_{r}}$ $\text{Pr}(j \in \mathcal{A}(\mathbf{Y},u))$ reaches the maximum value when $\mathcal{V}_{r}$ contains the $N-r + 1$ items among all items in $\mathcal{I}\setminus \mathcal{I}_u$ that have the largest probability upper bounds, which we denote as $\mathcal{V}_{r}=\{v_1, v_2, \cdots, v_w\}$, where $\overline{p}_{v_1} \leq \overline{p}_{v_2}\leq \cdots \leq \overline{p}_{v_w}$ and $w=N -r +1$. Formally, we have the following: 
{\small 
\begin{align}
& \max_{\mathcal{V}_{r}}\min_{j \in \mathcal{V}_{r}}\text{Pr}(j \in \mathcal{A}(\mathbf{Y},u)) \\
\leq & \min_{c =1}^{N-r +1} N' \cdot(\overline{p}^{*}_{\mathcal{H}_{c}} \cdot \frac{{n \choose s}}{{n' \choose s}}+  (\frac{s}{n'}-\frac{s}{n}\cdot \frac{{n \choose s}}{{n' \choose s}}))/c, 
\end{align}
}
where $\mathcal{H}_{c}=\{v_1, v_2, \cdots, v_{c}\}$.
Since we have Equation~(\ref{necessary_condition_of_nega_u_is_true_app}) when $\neg U$ is true and $1 \leq r \leq \min(k, N)$, we have the following: 
{\small 
\begin{align}
\label{equation_of_the_middle}
    \underline{p^{*}_{\mu_{r}}} \cdot  \frac{{n \choose s}}{{n' \choose s}}  \leq  & \min ( \min_{c=1}^{N-r +1}\frac{ N' \cdot(\overline{p}^{*}_{\mathcal{H}_{c}} \cdot \frac{{n \choose s}}{{n' \choose s}}+  (\frac{s}{n'}-\frac{s}{n}\cdot \frac{{n \choose s}}{{n' \choose s}}))}{c}, \nonumber\\
    & \overline{p}^{*}_{v_1}\cdot \frac{{n \choose s}}{{n' \choose s}}+ \frac{s}{n'}-\frac{s}{n}\cdot \frac{{n \choose s}}{{n' \choose s}}).
\end{align}
}
where $\mathcal{V}_{r}=\{v_1, v_2, \cdots, v_{N - r +1}\}$  and $\mathcal{H}_{c}=\{v_1, v_2, \cdots, v_{c}\}$. 
The Equation~(\ref{equation_of_the_middle}) is equivalent to the following: 
\begin{align}
   \underline{p^{*}_{\mu_{r}}}  \leq   \min ( \min_{c=1}^{N-r +1}\frac{ N' \cdot(\overline{p}^{*}_{\mathcal{H}_{c}} + \sigma)}{c}, 
     \overline{p}^{*}_{v_1}+ \sigma),  
\end{align}
where $\sigma = \frac{s}{n'} \cdot \frac{{n' \choose s}}{{n \choose s}} -\frac{s}{n}$ . 

\myparatight{Applying the law of contraposition} We leverage the law of contraposition and we have the following: if we have $1 \leq r \leq \min(k, N)$ and the following: 
{\small 
\begin{align}
\label{optimization_prob_app}
   \underline{p^{*}_{\mu_{r}}} >   \min ( \min_{c=1}^{N-r +1}\frac{ N' \cdot(\overline{p}^{*}_{\mathcal{H}_{c}} + \sigma)}{c}, 
     \overline{p}^{*}_{v_1}+ \sigma),  
\end{align}
}
where $\sigma = \frac{s}{n'} \cdot \frac{{n' \choose s}}{{n \choose s}} -\frac{s}{n}$. Then, we have $|\mathcal{I}_{u} \cap \mathcal{T}(\mathbf{M}',u)| \geq r $.  The Equation~(\ref{optimization_prob_app}) is satisfied for $\forall \mathbf{M}' \in \mathcal{L}(\mathbf{M},e)$. We can find the maximum value of $r$, where $1 \leq r \leq \min(k, N)$, that satisfies the Equation~(\ref{optimization_prob_app}), which is essentially the optimization problem in the Equation~(\ref{optimization_problem_theorem_1}). We reach the conclusion.

\section{Proof of Theorem 3.2}
\label{proof_of_proposition}
Based on Equation~(\ref{cp_lower_bound})~-~(\ref{cp_upper_bound}) and Boole's inequality in probability theory, we have the following probability: 
\begin{align}
    \text{Pr}((p_i \geq \underline{p_i},\forall i \in \mathcal{I}_u) \wedge (p_j \leq \overline{p}_j,\forall j \in \mathcal{I}\setminus\mathcal{I}_u)) \geq 1 - \alpha_u, 
\end{align}
where $R \wedge S$ is true if and only if $R$ is true and $S$ is true. Note that there is no randomness in our optimization problem in Equation~(\ref{optimization_problem_theorem_1}). Therefore, the probability that our Algorithm~\ref{alg:certify} computes an incorrect $r_u$ 
for the user $u$ is at most $\alpha_u$. 
Recall that we set $\alpha_u = \frac{\alpha}{n}$ in our Algorithm~\ref{alg:certify}. 
Based on the Boole's inequality, we know the probability that our Algorithm~\ref{alg:certify} computes an incorrect $r_u$
for at least one user among  all users in $\mathcal{U}$ is at most $\alpha$.